\documentclass[american,aps,prl,reprint,superscriptaddress,nofootinbib]{revtex4-2}

\usepackage{amsmath,amssymb}

\usepackage{amsthm}
\usepackage{amsfonts}

\usepackage{color}
\usepackage{graphicx}
\usepackage[american]{babel}
\usepackage[utf8]{inputenc}
\usepackage{times}
\usepackage{braket} 				

\usepackage{subcaption}
\usepackage{verbatim}
\usepackage{quantikz}
\usepackage{enumitem}

\usepackage{float}
\usepackage{bbold}
\usepackage{mathrsfs}  
\usepackage{titlesec}

\definecolor{mygrey}{gray}{0.35}
\definecolor{myblue}{rgb}{0.2,0.2,0.8}
\definecolor{myzard}{cmyk}{0,0,0.05,0}
\definecolor{mywhite}{rgb}{1,1,1}
\definecolor{myred}{rgb}{0.9,0.1,0.}
\definecolor{goldenyellow}{rgb}{1.0, 0.87, 0.0}
\definecolor{cornellred}{rgb}{0.7, 0.11, 0.11}
\definecolor{textgreen}{RGB}{25,160,34}
\definecolor{dartmouthgreen}{rgb}{0.05, 0.5, 0.06}

\usepackage[colorlinks=true,citecolor=myblue,linkcolor=myblue,urlcolor=myblue]{hyperref}

\newtheoremstyle{customStyle1}  
{0pt}       
{0pt}       
{\normalfont}   
{\parindent}        
{\em}  
{. --}   	 
{.5em}       
{\thmname{#1}\thmnumber{ #2}\thmnote{ (#3)}}  

\theoremstyle{customStyle1}
 \titleformat{\section}[runin]{\itshape}{\thesection}{1em}{}[.--]
\titlespacing*{\section }{\parindent}{1ex}{1ex}[0pt]
\newcounter{theorems}

\newtheorem{thm}[theorems]{Theorem}

\newtheorem{prop}[theorems]{Proposition}

\newtheorem{corollary}[theorems]{Corollary}
\newtheorem{defin}[theorems]{Definition}
\newtheorem{lem}[theorems]{Lemma}
\newtheorem{lemma}[theorems]{Lemma}

\newtheorem{ex.}{Example}[theorems]

\newtheorem*{cor*}{Corollary}
\newtheorem*{thm*}{Theorem}
\newtheorem*{prop*}{Proposition}
\newtheorem*{lem*}{Lemma}
\newtheorem*{rem*}{Remark}

\usepackage{cancel}
\usepackage{soul}

\newcommand{\id}{{\mathbb{1}}}

\newcommand{\ketbra}[2]{|#1\rangle\!\langle#2|}

\DeclareMathOperator{\idChannel}{id}

\DeclareMathOperator{\MIO}{MIO}
\newcommand*\diff{\mathop{}\!\mathrm{d}}
\newcommand{\dPhi}{\frac{\diff \phi}{2\pi}}

\DeclareMathOperator{\comb}{Comb}

\newcommand{\Comb}[1]{\comb(#1)}

\DeclareMathOperator{\trace}{Tr}
\newcommand{\Tr}[1]{\trace\left[#1\right]}
\newcommand{\partTr}[2]{\trace_{#1}\left[#2\right]}
\newcommand{\Choi}{Choi-Jamiołkowski}

\DeclareMathOperator{\I}{I}
\DeclareMathOperator{\II}{II}

\DeclareMathOperator*{\essinf}{ess\,inf}
\DeclareMathOperator{\In}{in}
\DeclareMathOperator{\Out}{out}
\DeclareMathOperator{\all}{all}
\DeclareMathOperator{\odd}{odd}
\DeclareMathOperator{\even}{even}
\DeclareMathOperator{\suchthat}{s.t.}

\newcommand\newsubcap[1]{\phantomcaption%
	\caption*{\figurename~\thefigure\thesubfigure: #1}} 

\begin{document}

    \title{Coherence as a resource for phase estimation}
	\author{Felix Ahnefeld}
	\email{felix.ahnefeld@uni-ulm.de}
	\affiliation{Institute of Theoretical Physics, 			Universit{\"a}t Ulm, Albert-Einstein-Allee 11, D-89069 Ulm, Germany}
	\author{Thomas Theurer}
	\email{tth@math.ku.dk}
	\affiliation{Department of Mathematical Sciences, University of Copenhagen, Universitetsparken 5, 2100, Denmark}
	\author{Martin B. Plenio}
	\email{martin.plenio@uni-ulm.de}
	\affiliation{Institute of Theoretical Physics, Universit{\"a}t Ulm, Albert-Einstein-Allee 11, D-89069 Ulm, Germany}

	\begin{abstract}
Quantum phase estimation is a core task 
in quantum technologies ranging from metrology to quantum computing, where it appears as a key subroutine in various algorithms. Here, we quantitatively connect the performance of phase estimation protocols with quantum coherence. To achieve this, we construct and characterize resource theories of quantum networks that cannot generate coherence. Given multiple copies of a unitary encoding an unknown phase and access to a fixed coherent state, we estimate the phase using such networks. For a unified and general approach, we assess the quality of the estimate using a generic cost function that penalizes deviations from the true value. We determine the minimal average cost that can be achieved in this manner and explicitly derive optimal protocols. From this, we construct a family of coherence measures that directly connect a state's coherence with its value for phase estimation, demonstrating that every bit of coherence helps. This establishes coherence as a resource that quantifies the performance of phase estimation, and, thus, of any quantum technology relying on it as a subroutine.
	\end{abstract}
	\date{\today}
	\maketitle

\section{Introduction}
In the past decades, various quantum algorithms promising superpolynomial speedups compared to their classical counterparts have been discovered, most prominently Shor's factorization algorithm~\cite{Shor1997}. Many of these algorithms~\cite{Shor1997,Kitaev1995,Harrow2009} contain a crucial subroutine, namely \textit{quantum phase estimation}, as detailed in Refs.~\cite{Cleve1998, Kitaev1997,vanDam2007,Nielsen2010}. The goal of phase estimation is to determine an unknown phase $\phi$ encoded by a unitary as accurately as possible.  This also comprises a fundamental problem in metrology itself and has been studied in Ref.~\cite{Giovannetti2006}. Other applications are found in quantum chemistry and quantum simulation~\cite{Aspuru-Guzik2005}. This makes phase estimation a core problem for quantum technologies.

It is clear that access to coherence, i.e., quantum superposition, is necessary for phase estimation~\cite{MunozLahoz2022,Lecamwasam2024}. But how much coherence is required for accurate phase estimation or, equivalently, what are the ultimate limitations on the accuracy of any phase estimation protocol that emerge from a lack of coherence? Answering these questions quantitatively is the central goal of this work. For a rigorous approach, we turn to the framework of quantum resource theories~\cite{Chitambar2019}, and in particular to the resource theory of coherence~\cite{Baumgratz2014, Streltsov2017}.

The quality of a phase estimate depends on the context: For example, to solve integer factorization via Shor's algorithm, one requires estimates that are sufficiently close to the true value because this ensures that the classical post-processing succeeds with sufficiently high probability. In contrast, in metrology, one tries to optimize the mean squared error. For a unified treatment of the diverse applications of phase estimation, we adopt the approach of Ref.~\cite{vanDam2007} and consider generic cost functions that penalize deviations of our estimates $\hat{\phi}$ from the true value $\phi$. 

We evaluate the minimal average cost that an optimal phase estimation protocol consisting of a fixed coherent state and quantum networks that cannot generate coherence can achieve and show that this cost is fully determined by a family of coherence monotones. This establishes a quantitative relationship between coherence and the quality of phase estimation. Furthermore, we show that any amount of coherence provides a direct operational advantage in phase estimation, and give an explicit construction of optimal protocols achieving this advantage. 
For other examples that relate resource-theoretic quantities with the efficiency of specific quantum algorithms, see, e.g., Refs.~\cite{Bruss2011,Pashayan2015,Bravyi2016,Shi2017, Ahnefeld2022, Naseri2022, Anand2025}. Moreover, Refs.~\cite{Jozsa2003,Vidal2003} derived \textit{necessary} entanglement-based conditions for super-polynomial speed-ups in quantum algorithms and Ref.~\cite{Howard2014} linked contextuality and non-stabilizerness to quantum computational advantages. The relevance of such results lies in the fact that they provide a deeper understanding of the properties that underpin the advantages of both specific and general quantum algorithms.

Moreover, our work includes an alternative proof for the results of Ref.~\cite{vanDam2007} concerning optimal phase estimation without resource constraints and characterizes the structure of coherence non-generating networks and supermaps, which we will introduce in the next sections. 

\begin{figure*}
    \centering
    \includegraphics[width=1\linewidth]{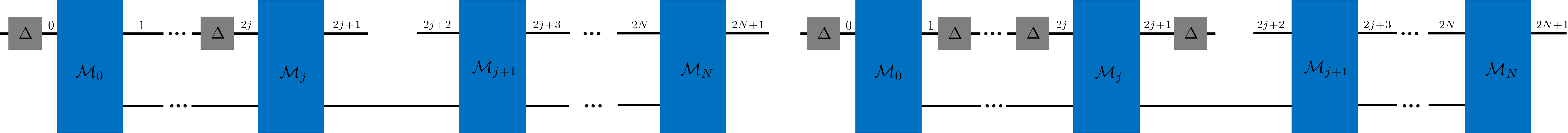}
    \caption{A MIO-compatible network satisfies the condition that incoherent inputs at the first $j$ inputs, i.e., systems $0,2,\ldots, 2j$, imply incoherent outputs at the output systems $1,3,\ldots,2j+1$ for all $0\leq j\leq N$.}
    \label{fig:DeltaConditions}
\end{figure*}

\section{Coherence as a resource} 
Resource theories of coherence~\cite{Baumgratz2014,Streltsov2017} quantify the value of superposition with respect to a fixed orthonormal basis $\{\ket{i}\}_i$ referred to as the incoherent basis. A quantum state $\sigma$ is called incoherent (or free) iff it is diagonal in the incoherent basis, i.e., iff $\Delta(\sigma)=\sigma$, where $\Delta(\rho)=\sum_i \ketbra{i}{i} \rho \ketbra{i}{i}$ denotes total dephasing in the incoherent basis $\{\ket{i}\}_i$. We denote the set of incoherent states by $\I$. States not contained in $\I$ are called coherent and considered resourceful. Since eigenstates of the unitary that encodes the phase we intend to estimate are invariant under its action, they cannot be used to infer any information about the phase. In the context of phase estimation, it is thus natural to choose this eigenbasis as the incoherent one, which we do. To use coherence optimally for phase estimation (or any other application), we need to be able to manipulate it, which is done with the help of free channels. The choice of free channels is not unique, and there are several relevant classes of free channels considered in the literature~\cite{Streltsov2017,Aberg2006,Baumgratz2014,Winter2016,Yadin2016,Chitambar2016,Chitambar2016b,Chitambar2017,Marvian2016} emerging, for example, from practical or conceptual considerations. A minimal requirement for a meaningful resource theory is that free channels cannot create coherence~\cite{Chitambar2019, Gour2024}. The largest set of channels that satisfies this requirement is called the set of \textit{maximally incoherent} operations, denoted by $\MIO$~\cite{Aberg2006, Liu2017, Diaz2018, Theurer2019}, which consists of all channels $\mathcal{M}$ that satisfy $\mathcal{M}\Delta=\Delta\mathcal{M}\Delta$. To determine the ultimate limitations that restricted access to coherence imposes on phase estimation, we choose MIO as the set of free channels, since this allows for the most advantageous manipulation and thus usage of coherence.

\section{ Free quantum networks}
The most general way to use a set of $N$ channels $\mathcal{N}_1,\ldots, \mathcal{N}_N$ in a laboratory is to insert them into a quantum network~\cite{Chiribella2008,Chiribella2009}, i.e., to interlace them with a sequence of channels $\mathcal{M}_0,\ldots ,\mathcal{M}_N$ transforming them to $\mathcal{S}_N[\mathcal{N}_1,\ldots,\mathcal{N}_N]\!=\!\mathcal{M}_{N}(\idChannel\! \otimes \mathcal{N}_N\!) \mathcal{M}_{N\!-\!1}(\idChannel\! \otimes \mathcal{N}_{N\!-\!1}) \mathcal{M}_{N\!-\!2}\ldots \mathcal{M}_1 (\idChannel\! \otimes \mathcal{N}_1)\mathcal{M}_0$, where $\idChannel$
refers to identity channels on auxiliary systems. Different networks can lead to the same effective transformation $\mathcal{S}_N$, which is called a supermap. Using a generalization of the Choi-Jamiołkowski isomorphism~\cite{Jamiolkowski1972, Choi1975}, such supermaps can be represented by quantum combs~\cite{Chiribella2008, Chiribella2009} (likewise, a quantum channel can be represented by a Choi state and implemented with different Stinespring dilations). Quantum supermaps will be our main tool for the study of optimal phase estimation protocols with limited coherence, which is only meaningful if the supermaps under consideration have well-defined restrictions on their ability to create coherence. To this end, we divide the set of quantum supermaps into free and non-free ones (as we did in the previous section for states and channels). Since quantum supermaps are implemented by quantum networks, it is intuitive and operational to consider the supermaps that can be implemented via a network of free channels as free. 
\begin{defin}\label{def:MIONetwork}     
    A quantum supermap $\mathcal{S}_N$ corresponds to a $\MIO$ network and is considered free if its action can be decomposed into a sequence of $\MIO$ channels $\mathcal{M}_0,\ldots, \mathcal{M}_N \in \MIO$, i.e., $\mathcal{S}_N[\mathcal{N}_1,\ldots, \mathcal{N}_N]=\mathcal{M}_{N}(\idChannel \otimes \mathcal{N}_N)\mathcal{M}_{N-1}\ldots \mathcal{M}_1 (\idChannel \otimes \mathcal{N}_1)\mathcal{M}_0$ for all channels $\mathcal{N}_1,\ldots, \mathcal{N}_N$.  
\end{defin}

It is clear (see also the Supplemental Material (SM)~\cite{SM}) that the above definition is also solid from a structural perspective in the sense that the composition of MIO networks leads to another MIO network. In particular, this implies that a MIO network cannot transform free states and channels into resourceful ones, and thus cannot generate resources for free. For further details on how to build resource theories of quantum supermaps in general, see the SM and also Refs.~\cite{Pollock2018,Liu2019,Gour2020,Gour2021b,Berk2021,Berk2023, Taranto2025}. However, since it seems that the set of MIO networks is hard to characterize, as a technical tool, we will utilize the following relaxation.

\begin{defin}\label{def:DeltaConditionsNetwork} 
A quantum supermap $\mathcal{S}_N$ is called $\MIO$-compatible if it satisfies the conditions depicted in Fig.~\ref{fig:DeltaConditions} for all $0\le j\le N$, i.e., if the first $j$ inputs are incoherent, then the first $j$ outputs are incoherent too.
\end{defin}

The intuition behind this definition is that a free network should not be able to output coherence before coherence is fed into it. In the SM, we show that the set of MIO-compatible supermaps is the largest set of superchannels that is consistent with MIO channels from a resource-theoretical perspective. This implies that a superchannel (i.e., $\mathcal{S}_N=\mathcal{S}_1$) is completely MIO-preserving (see also e.g. Refs.~\cite{Gour2020,Gour2021b}), i.e., it maps every MIO channel to a MIO channel in a complete sense, iff it is MIO-compatible. Importantly, a supermap $\mathcal{S}_N$ is MIO-compatible iff its comb $J_{\mathcal{S}_N}$ satisfies
\begin{align}
    \Delta_{0,2\ldots 2j} J_{\mathcal{S}_N}\! =\!\Delta_{0,2\ldots 2j} \Delta_{1,3\ldots 2j+1} J_{\mathcal{S}_N} \, \forall j\!:0\leq j\leq N,
\end{align}
where $\Delta_{x_1,...,x_n}$ denotes the application of a total dephasing to subsystems $x_1,...,x_n$. This is a semidefinite constraint which allows the use of the powerful methods of the comb representation~\cite{Chiribella2008,Chiribella2009} and semidefinite programming~\cite{Vandenberghe1996,Boyd2004}. In the SM, we show that MIO networks are a subset of MIO-compatible supermaps. However, we will see later that the two sets are equally powerful for phase estimation.

\begin{figure*}
     \centering
     \begin{subfigure}[b]{0.6\linewidth}
         \centering
         \includegraphics[width=1\linewidth]{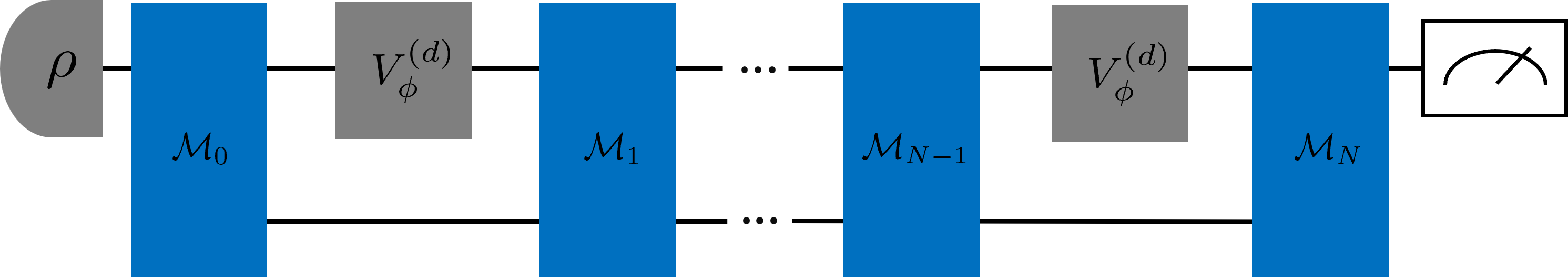}
         \newsubcap{Phase estimation protocol composed of a network of MIO operations probing $N$ copies of $V_\phi^{(d)}$ and a (resourceful) input state $\rho$.}
         \label{fig:PhaseEstimationProtocols}
     \end{subfigure}%
     \begin{subfigure}[b]{0.4\textwidth}
         \centering
         \includegraphics[width=0.8\linewidth]{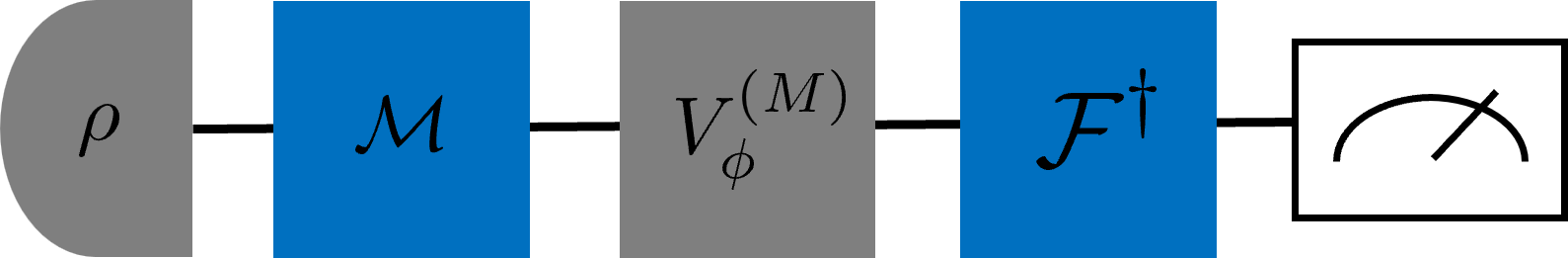}
		 \newsubcap{Protocol that achieves the minimal average cost in Theorem~\ref{thm:OptAvgCost}, where $\mathcal{F}$ denotes the quantum Fourier transform, $\mathcal{M} \in \MIO$, $M=(d-1)N+1$, and the measurement is a projective measurement in the computational basis. }
		 \label{fig:SimpleCircuit}
     \end{subfigure}
\end{figure*}

\section{Quantum phase estimation}
Phase estimation is a central subroutine in many quantum algorithms.  Given $N$ copies of the black-box unitary $V_\phi^{(d)}=\sum_{n=0}^{d-1} e^{i\phi n} \ketbra{n}{n}$, the goal is to estimate the unknown phase $\phi$ as accurately as possible. As shown in Ref.~\cite{Cleve1998}, many quantum algorithms can be phrased in this manner. For example, Shor's integer factorization algorithm, in which we try to estimate a phase containing information that allows us to find a prime factor, is one of them (see Ref.~\cite{Cleve1998} and the SM for further information). For a unified description of different applications of phase estimation, we therefore evaluate the performance of a phase estimation protocol using a generic cost function $C(\phi-\hat{\phi})$ that only depends on the difference between $\phi$ and our estimate $\hat{\phi}$ and satisfies the following mild assumptions: (i) $C(\phi)\geq 0$, (ii) $C(\phi+2\pi) = C(\phi)$, (iii) $\int_0^{2\pi} \diff\phi\ C(\phi)<\infty$.

To investigate limitations that restricted access to coherence imposes on phase estimation, we supply a fixed state $\rho$ and $N$ copies of the unitary $V_\phi^{(d)}$, which we insert into a MIO network as in Definition~\ref{def:MIONetwork} followed by a measurement; see Fig.~\ref{fig:PhaseEstimationProtocols}. Based on the measurement outcomes $x$ with probability distribution $p^{(d,N)}(x|\phi,\rho)$ we assign phase estimates $\hat{\phi}_x$. 
Assuming a uniform prior distribution of the phase, our goal is to determine the minimal average cost  
\begin{align}\label{eq:AvgCostMainText}
    C_{\min}^{(d,N)}(\rho)=  \inf \sum_{x} \int_{0}^{2\pi}\dPhi  C(\phi -\hat{\phi}_x) p^{(d,N)}(x|\phi,\rho).
\end{align}
The infimum is taken over all MIO networks, all measurements (with a potentially unbounded number of outcomes) leading to the probability distribution $p^{(d,N)}(x|\phi,\rho)$, and all possible assignment rules for the phase estimate $x\mapsto \hat{\phi}_x$. A MIO network is considered optimal if it achieves $C_{\min}^{(d,N)}(\rho)$.

Since MIO is the largest set of free channels, no other set of free channels allows for better utilization of the coherence contained in $\rho$. If we choose any subset of MIO as free, the correspondingly defined minimal average cost cannot be lower than the one defined via MIO. Moreover, the minimal average cost in Eq.~\eqref{eq:AvgCostMainText} is not limited to a specific protocol because we optimize over all MIO networks, measurements, and assignment rules - and it is in this sense that Eq.~\eqref{eq:AvgCostMainText} determines the ultimate limitations that restricted access to coherence imposes on phase estimation. Due to the generality of our approach, we are modeling neither a specific phase estimation protocol nor hardware-dependent experimental restrictions. Nonetheless, any cost that can be obtained with a protocol that corresponds to a MIO network is lower bounded by Eq.~\eqref{eq:AvgCostMainText}. By choosing MIO as free, we consider an idealization in the sense that, in a specific experiment, not all MIO channels are necessarily easy to implement. On the other hand, in such a specific experiment, it might be reasonably easy to implement operations outside of MIO. By allowing for such operations, we do, however, no longer consider coherence as a restricted resource but are modeling other constraints instead (which are experiment-dependent). Since we are interested in a general approach, we do not pursue this direction, although it can be highly relevant for the specific experiment under consideration. This idealization allows us to single out and rigorously investigate the role of coherence in phase estimation. A similar approach has been pursued in the investigation of the role of entanglement in quantum information theory~\cite{Bennett1996a,PlenioV07,Horodecki2009}, where local operations and classical communication are often considered free, although many local operations are difficult to implement in practice but entangling two qubits close to each other is often feasible.

\section{Main Results}
We now show that the optimal achievable average cost using an arbitrary protocol as depicted in Fig.~\ref{fig:PhaseEstimationProtocols} is achieved by the surprisingly simple phase estimation protocol depicted in Fig.~\ref{fig:SimpleCircuit}.

\begin{thm}\label{thm:OptAvgCost}
    Let $Y^{(M)}\in \mathbb{C}^{M\times M}$ be the Toeplitz matrix defined by \begin{align}\label{eq:CostMatrixDefinition}
       Y^{(M)}= \sum_{n,m=0}^{M-1}\int \dPhi C(\phi) e^{i\phi(n-m)}\ketbra{n}{m}.
   \end{align}
  Given $N$ copies of the unitary $V_\phi^{(d)}$, the minimal average cost in Eq.~\eqref{eq:AvgCostMainText} is 
   \begin{align}\label{eq:MinAvgCostSDPMainText}
       C_{\min}^{(d,N)}(\rho)&=  C_{\min}^{(M,1)}(\rho)= \min_{\mathcal{M} \in \MIO} \Tr{Y^{(M)} \mathcal{M}(\rho)},
   \end{align}
   where $M=(d-1)N+1$ and the input and output dimensions of the channel $\mathcal{M}$ are fixed by $\rho$ and $Y^{(M)}$, respectively. The minimal average cost is achieved by the protocol shown in Fig.~\ref{fig:SimpleCircuit}: the $N$ copies of the black-box unitary $V_\phi^{(d)}$ are combined with $\MIO$ operations to implement the unitary $V_\phi^{(M)}$. This unitary is probed with a state $\mathcal{M}(\rho)$, where $\mathcal{M}$ is a $\MIO$ channel that optimizes Eq.~\eqref{eq:MinAvgCostSDPMainText}. The output is subsequently measured in the Fourier basis, and upon outcome $x$, a phase estimate $\hat{\phi}_x=\frac{2\pi x}{M}$ is assigned.
\end{thm}

We emphasize that the optimal phase estimation protocol depicted in Fig.~\ref{fig:SimpleCircuit} depends on the cost function $C$ only via the channel $\mathcal{M}$. The implementation of $V_\phi^{(M)}$, the measurement in the Fourier basis, and the phase estimates are optimal for any cost function. Moreover, using a memory system in the quantum network is never necessary. We also want to mention that the quantum Fourier transform in Fig.~\ref{fig:SimpleCircuit} is not a MIO operation. This does not contradict our premise to consider only MIO networks since we can always absorb it into the subsequent measurement, and the resulting measurement can never create coherence. We choose to depict the protocol including the Fourier transform to highlight the similarity to the usual depiction of phase estimation protocols. The right-hand side of Eq.~\eqref{eq:MinAvgCostSDPMainText} is solvable via a semidefinite program.

Our results generalize Ref.~\cite{vanDam2007}, where optimal phase estimation protocols without coherence constraints have been studied. The authors of Ref.~\cite{vanDam2007} showed that the protocol in Fig.~\ref{fig:SimpleCircuit} where $\mathcal{M}$ is replaced by an optimal unconstrained channel is optimal (for $d=2$). Our results include their findings, as providing enough coherence allows us to implement any channel~\cite{Baumgratz2014}. However, their proof methods relying on dilations of the channels in the network cannot be extended to the scenario that we consider because MIO channels do not possess a free dilation~\cite{Chitambar2016,Chitambar2016b,Chitambar2017}. Our method of showing Theorem~\ref{thm:OptAvgCost} is to relax the optimization over the set of all MIO networks in Eq.~\eqref{eq:AvgCostMainText} to the set of MIO-compatible supermaps introduced in Definition~\ref{def:DeltaConditionsNetwork}. Using symmetry arguments similar to Ref.~\cite[Chap.~4.4]{Holevo2011} to reexpress the optimization problem, we prove that it is sufficient to consider $M$ measurement outcomes and phase estimates $\hat{\phi}_x=\frac{2\pi x}{M}$. 
This allows us to show that the resulting optimization problem can be solved using an SDP. 
Via the dual problem, we then prove that the right-hand side of Eq.~\eqref{eq:MinAvgCostSDPMainText} is a lower bound on the minimal average cost. To conclude the proof, we show that the protocol depicted in Fig.~\ref{fig:SimpleCircuit} achieves this lower bound. This also implies that Theorem~\ref{thm:OptAvgCost} holds if we define $C_{\min}^{(d,N)}(\rho)$ via an optimization over MIO-compatible supermaps.

To understand the relevance of coherence in phase estimation, we first consider two extremal cases: Supplying an incoherent or a maximally coherent state of unbounded dimension. By providing an incoherent input state, according to Theorem~\ref{thm:OptAvgCost}, the optimal average cost reduces to the average of the cost function, i.e., for any $\sigma  \in \I$, we have $C_{\min}^{(M,1)}(\sigma)=\int \dPhi C(\phi)=:C_0$, where $C_0$ is the cost associated with a guess solely based on the uniform prior distribution. Intuitively, this is expected as without coherence, no information about the phase is encoded, which could improve our guesses. If, on the other hand, we supply a maximally coherent state of sufficient dimension, we can prepare any possible quantum state, and in particular, the optimal input states of Refs.~\cite{vanDam2007,vanDam2007b}, such that the minimal average cost is given by the minimal eigenvalue $\lambda_{\min}$ of $Y^{(M)}$. For any coherent state $\rho \notin \I$ between these two extremal cases, we expect to be able to infer some information about the phase and lower the average cost below $C_0$. To quantify this, for an arbitrary but fixed cost function, we introduce the functional 
\begin{align}\label{eq:Advantage}
    \mathcal{A}^{(M)}(\rho):= \!\max_{\mathcal{M}\in \MIO}\! \Tr{\left( \Delta \left(Y^{(M)}\right)\!-\!Y^{(M)}\!\right)\!\mathcal{M}(\rho)},
\end{align}
where $Y^{(M)}$ is defined as in Eq.~\eqref{eq:CostMatrixDefinition}. This directly determines the decrease in average cost compared to random guesses with a cost of $C_0$ since
\begin{align}\label{eq:CostAdvantage}
    C_{\min}^{(M,1)}(\rho)=C_0-\mathcal{A}^{(M)}(\rho).
\end{align}
We now show that the functionals in Eq.~\eqref{eq:Advantage} define a family of coherence monotones~\cite{Streltsov2017}.

\begin{thm}\label{thm:AdvantageCoherence}
    The functionals $\mathcal{A}^{(M)}$ are convex coherence monotones, i.e., $\mathcal{A}^{(M)}(\mathcal{N}(\rho))\leq \mathcal{A}^{(M)}(\rho)$  for all $\mathcal{N} \in \MIO$ and all states $\rho$. For any cost matrix $Y^{(M)}$ with $Y^{(M)}\neq \Delta Y^{(M)}$, the monotones are faithful, i.e.,  $\mathcal{A}^{(M)}(\rho)\geq 0$ with equality iff $\rho \in \I$.
\end{thm}
This, together with Eq.~\eqref{eq:CostAdvantage}, shows that the optimal performance of phase estimation is intrinsically linked to the available coherence. Unless the cost function is constant and thus trivial, there always exists a sufficiently large $M$ such that the corresponding cost matrix $Y^{(M)}$ is not diagonal. This implies that every amount of coherence is useful for phase estimation problems described by a non-trivial cost function, and the operational advantage it provides is directly quantified by the task-tailored coherence measures in Eq.~\eqref{eq:Advantage}. Therefore, coherence is a crucial resource for phase estimation protocols and, thus, for algorithms that use it as a subroutine.

Linking operational advantages directly to employed resources in real applications is a primary motivation for resource theories in general and coherence theory in particular. This has been studied in recent years for various forms of discrimination, exclusion, and detection games~\cite{Napoli2016,Piani2016,Takagi2019,Takagi2019b,Skrzypczyk2019a,Skrzypczyk2019b,Uola2019,Mori2020,Ducuara2020a,Ducuara2020b,Uola2020,Masini2021,Wagner2024}. For example, trying to discriminate between unitaries applying different phases from a discrete set is such a discrimination game, in which coherence emerges as a crucial resource \cite{Napoli2016,Piani2016}. Our work generalizes this to the continuous setting. Moreover, in various applications in computation, see, e.g., Refs.~\cite{Hillery2016,Matera2016,Biswas2017,Ahnefeld2022,Naseri2022,Zhou2024}, the pivotal role of coherence has been studied. Showing that the available coherence directly quantifies the advantage in optimal phase estimation underscores its fundamental role in quantum computation, making the intuition that coherence is a necessary resource quantitative.

Naturally, the measures in Eq.~\eqref{eq:Advantage} are highly tailored to a specific cost function, respectively, the task at hand, and it would be interesting to separate the advantage attributed to coherence into a problem-specific and a coherence-dependent part. If dim$(\rho)=2$ and $M=2$, this is indeed possible and $\mathcal{A}^{(2)}(\rho)= \left(C_0-\lambda_{\min}\left(Y^{(2)}\right)\right)  C_R(\rho)$, where $C_R(\rho)$ is the (generalized) robustness of coherence~\cite{Piani2016, Napoli2016}. For the general case, we provide the following upper bound on $\mathcal{A}^{(M)}$, which holds for any cost function.

\begin{prop}
    The advantage $\mathcal{A}^{(M)}$ is bounded by
        \begin{align}\label{eq:advantageBound}
        \mathcal{A}^{(M)}(\rho)\leq \left(C_0-\lambda_{\min}\left(Y^{(M)}\right)\right)  W(\rho),
    \end{align}    
    where the weight of coherence $W$~\cite{Bu2017,Bu2018} is given by 
    \begin{align}
        W(\rho)= \min_{w\geq 0} \left\{\rho=(1-w)\sigma+w\tau, \sigma \in \I,\tau \text{ a state} \right\}.
    \end{align} 
\end{prop}
Recall that for the optimal phase estimation protocol with unrestricted access to coherence, the best achievable advantage over classical strategies is given by $C_0-\lambda_{\min}\left(Y^{(M)} \right)$~\cite{vanDam2007}. Since $0\le W(\rho)\le 1$~\cite{Bu2018}, the bound in Eq.~\eqref{eq:advantageBound} interpolates between the optimal advantage without coherence and with arbitrarily much coherence. As we demonstrate in the SM, this implies that for commonly used cost functions $\lim_{M\to\infty}C_{\min}^{(M,1)}(\rho)\ge C_0(1-W(\rho))$ which can be strictly larger than zero, while in the unrestricted case~\cite{vanDam2007b}, the limit converges to zero. As such, restricted access to coherence imposes a fundamental bound to the performance of phase estimation even in the limit of unboundedly many copies of the unitary.

\section{Discussion}
In this work, we established a direct quantitative connection between the optimal performance of phase estimation protocols and the coherence employed in the form of a resourceful input state. This provides valuable insight into phase estimation that is twofold: On one hand, lack of coherence poses quantifiable limitations to the accuracy of phase estimation. On the other hand, every state containing coherence also provides a quantifiable advantage. Thus, coherence is a key resource for phase estimation. Together with Refs.~\cite{Hillery2016,Matera2016,Biswas2017,Ahnefeld2022,Naseri2022,Zhou2024}, this establishes a quantitative connection between coherence and advantages in quantum computation. Moreover, our work is an example for the application of resource theories of supermaps, see also Refs.~\cite{Pollock2018,Liu2019,Gour2020,Gour2021b,Berk2021,Berk2023, Taranto2025}, to problems of interest.

Coherence in the context of phase estimation has previously been studied~\cite{MunozLahoz2022,Lecamwasam2024,Ahnefeld2022}.  Ref.~\cite{MunozLahoz2022} investigated the optimal average cost with respect to a fixed cost function, namely, the periodized variance, that can be achieved given an arbitrary state with a fixed amount of coherence measured by the robustness of coherence to which a single copy of the phase gate is applied. This allowed Ref.~\cite{MunozLahoz2022} to obtain asymptotic expressions. In contrast, in our work, we take an operational approach: Optimizing the physical usage of the employed coherence via MIO-compatible supermaps gives rise to tailored families of coherence measures that exactly quantify the performance of phase estimation protocols. This is not only operationally well-defined but also allows us to show that every amount of coherence serves as a resource for phase estimation. 
Moreover, in Ref.~\cite{Lecamwasam2024}, quantum parameter estimation in Bayesian metrology and its connection to coherence was studied. The coherence of state ensembles, defined as the ensemble relative entropy of coherence~\cite{Lecamwasam2024}, was identified to directly quantify the difference between Holevo quantity~\cite{Holevo1973,Wilde2013} and mutual information. 
On the one hand, the results of Ref.~\cite{Lecamwasam2024} are broader in scope than ours in the sense that phase estimation is a special case of parameter estimation. On the other hand, mutual information provides a lower bound on the average mean-squared error~\cite{Hall2012} and thus a specific cost function, while we consider general cost functions. 
In Ref.~\cite{Ahnefeld2022}, the role of coherence in a specific implementation of Shor's algorithm was investigated. This specificity allowed bounding its success probability by dynamical coherence measures~\cite{Theurer2019,Ahnefeld2022} of the employed channels (which include static coherence as studied here as a special case). In the SM, we discuss why here, in contrast, we decided to present our results in terms of static coherence and to which extent this encompasses the dynamical case.
This allows us to determine the role of coherence in general phase estimation problems that are relevant not only for other implementations of Shor's algorithm but also for tasks beyond factoring. In summary, our work unifies and expands key aspects from Refs.~\cite{MunozLahoz2022,Lecamwasam2024,Ahnefeld2022}: We consider arbitrary tasks for which phase estimation is relevant in the form of generic cost functions. Our operational approach of considering the optimal achievable protocol (with respect to the average cost) unveils the fundamental limitations that access to restricted coherence imposes, and is thus not limited to hardware- or implementation specific protocols.

We want to emphasize that, while coherence is a resource that quantifies the performance of phase estimation protocols, this does not imply that it is the only relevant resource.  Importantly, MIO operations can convert coherence to entanglement~\cite{Streltsov2015,Theurer2020}, which is a valuable resource not only for quantum metrology~\cite{Wineland1996,Braunstein1994,Huelga1997,Augusiak2016}, but also for quantum computation~\cite{Linden2001,Jozsa2003,VandenNest2013}, and cryptography~\cite{Bennett2014}. From the proof of Theorem~\ref{thm:OptAvgCost} in the SM, it is apparent that the optimal protocol in Fig.~\ref{fig:SimpleCircuit} uses the coherence provided to entangle the input states of the different copies of the unitary encoding the phase. 
Generalizing our results to non-uniform prior distributions, dynamical resources (including measurements~\cite{Theurer2019}), or resourceful networks, as well as parameter estimation beyond phase estimation, could provide further insights into the role of (dynamical) coherence for quantum advantages and is left for future work.

\textit{Note added:} Shortly after this work appeared on the arXiv, the independent work Ref.~\cite{Chen2025} investigated the limitations that energy constraints impose on parameter estimation (and specifically phase estimation)  using methods similar to ours.

\section{Acknowledgment} 
We thank Alexander Taveira Blomenhofer, Tore Friis, and Koenraad Audenaert for discussions. T. T. acknowledges support from a Postdoc Scholarship on Quantum Algorithms or Quantum Software from the Danish e-infrastructure Consortium (DeiC) and support from the Pacific Institute for the Mathematical Sciences (PIMS). The research and findings may not reflect those of the Institute. For numerics we used Refs.~\cite{cvx,Lofberg2004,Sturm1999}.

\bibliography{bibliography.bib}

\appendix
\maketitle
\pagebreak
\widetext
\newpage


	\begin{center}
		\textbf{Coherence as a resource for optimal phase estimation\\
        Supplemental Material 
        }
	\end{center}
   
\setcounter{secnumdepth}{2} 

We provide the proofs of the results presented in the main text and further discussions on the resource-theoretical concepts we employed to derive these results.

\section{Notation and preliminaries}\label{sec:Preliminaries}
We will restrict ourselves to finite-dimensional quantum systems that we label by Arabic numbers. Quantum states will be denoted by small Greek letters, quantum channels, i.e., linear and completely positive and trace-preserving (CPTP) maps that transform quantum states by calligraphic Latin letters with the exemption of the identity channel which is denoted by $\idChannel$. If required, we explicitly denote the systems on which a quantum channel acts as $\mathcal{M}^{1\leftarrow 0}$ to emphasize that $\mathcal{M}$ is a channel from system $0$ to $1$, i.e., from the set of linear operators $\mathcal{L}(\mathcal{H}_0)$ to $\mathcal{L}(\mathcal{H}_1)$. Quantum supermaps, i.e., linear maps that act on $N$ quantum channels, are denoted by calligraphic letters too, e.g., $\mathcal{S}_N$. 

The (unnormalized) Choi state~\cite{Choi1975} of a quantum channel $\mathcal{N}^{1\leftarrow 0}$ is defined by
\begin{align}\label{eq:ChoiState}
    J_\mathcal{N}:= \left(\idChannel^0 \otimes \mathcal{N}^{1\leftarrow \Tilde{0}} \right) \sum_{n,m} \ketbra{nn}{mm}_{0,\Tilde{0}},
\end{align}
where $\tilde{0}$ is a copy of the input system $0$. The action of a channel can be expressed in the Choi representation as $\mathcal{N}^{1\leftarrow 0}(\rho_0)=\partTr{0}{(\rho_0^T\otimes \id_1)J_\mathcal{N}}$, where $T$ denotes the transpose in the incoherent basis. This provides a one-to-one correspondence between quantum channels and linear operators known as the {\Choi} isomorphism~\cite{Jamiolkowski1972, Choi1975}. If $\mathcal{N}$ is completely positive, then $J_\mathcal{N} \geq 0$, and if $\mathcal{N}$ is trace-preserving, then $\partTr{1}{J_\mathcal{N}}=\id_0$.  Conversely, any (unnormalized) bipartite quantum state that satisfies these conditions corresponds to a quantum channel between the two systems $0$ and $1$. Whenever we speak of a Choi ``state'' we refer to an unnormalized state as in Eq.~\eqref{eq:ChoiState}. We will primarily use the Choi representation of quantum channels throughout this SM. If a formal distinction of the two objects is necessary, we denote the Choi state of a channel $\mathcal{N}$ by $J_\mathcal{N}$ as in Eq.~\eqref{eq:ChoiState}. If clear from the context, we will simply denote both objects by $\mathcal{N}$.

Next, we review how the Choi representation can be extended to supermaps. We will primarily follow Ref.~\cite{Chiribella2009}, which provides a comprehensive picture of the topic. As mentioned in the main text, the most general way to use a set of $N$ channels $\mathcal{N}_1,\ldots, \mathcal{N}_N$ in a laboratory is to plug them into a quantum network~\cite{Chiribella2009}, i.e., to interlace them with a sequence of channels $\mathcal{M}_0,\ldots \mathcal{M}_N$ transforming them to
\begin{align}\label{eq:InterlacingNetworks}
    \mathcal{S}_N[\mathcal{N}_1,\ldots,\mathcal{N}_N]=\mathcal{M}_{N}(\idChannel \otimes \mathcal{N}_N)\mathcal{M}_{N-1}  (\idChannel \otimes \mathcal{N}_1)\mathcal{M}_{N-2}\ldots \mathcal{M}_1 (\idChannel \otimes \mathcal{N}_1)\mathcal{M}_0.
\end{align}
This network implements a so-called supermap $\mathcal{S}_N$, which is depicted on the right side of Fig.~\ref{fig:NetworkvsComb}. One can now assign an analogue of the Choi state, called a quantum comb, to the entire network of quantum channels. Recall that for each of these channels $\mathcal{M}_j$, $J_{\mathcal{M}_j}$ denotes its Choi state, which is an (unnormalized) state on the joint system of its (accessible) input and output $2j$ and $2j+1$ and its (inaccessible) auxiliary systems $A_j,A_{j+1}$, see Fig.~\ref{fig:NetworkvsComb}. The comb $J_{\mathcal{S}_N}$ associated with the network is defined as
\begin{align}\label{eq:NetworkLinkProduct}
    J_{\mathcal{S}_N}&=  \partTr{ A_1,\ldots, A_N}{\left( \id_{A_2,\ldots, A_N, 2,\ldots, 2N+1} \otimes J_{\mathcal{M}_0}^{T_{A_1}}\right) \ldots \Big( \id_{A_1,\ldots, A_{N-1},0,\ldots, 2N-1} \otimes J_{\mathcal{M}_N}\Big)} \nonumber \\
    &=\partTr{ A_1,\ldots, A_N}{\left( \id \otimes J_{\mathcal{M}_0}^{T_{A_1}}\right) \left( \id \otimes J_{\mathcal{M}_1}^{T_{A_1,A_2}}\right) \ldots \Big( \id \otimes J_{\mathcal{M}_N}\Big)},
\end{align}
where $T_X$ denotes the partial transpose on system $X$, and in the second line we suppressed the subscripts on identities  since they are fixed by the Choi state they are tensored with.

\begin{figure}[ht]
    \centering
    \includegraphics[width=1\linewidth]{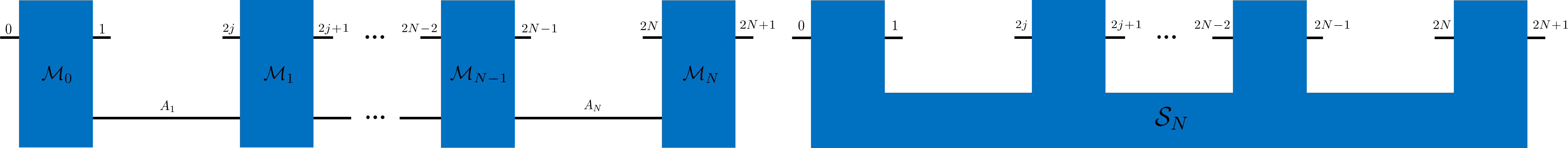}
    \caption{A quantum network comprised of channels $\mathcal{M}_0,\ldots, \mathcal{M}_N$ as depicted on the left side is a specific implementation of a supermap $\mathcal{S}_N$ shown on the right.}
    \label{fig:NetworkvsComb}
\end{figure}

Different networks can lead to the same effective dynamics on the accessible input and output systems. For example, we could insert any unitary followed by its inverse into the auxiliary systems in Fig.~\ref{fig:NetworkvsComb} between two successive channels $\mathcal{M}_j$ and $\mathcal{M}_{j+1}$, and absorb the unitaries into new channels $\tilde{\mathcal{M}}_j$ and $\tilde{\mathcal{M}}_{j+1}$, resulting in the same accessible dynamics as the original network. To depict the supermap $\mathcal{S}_N$ that describes these effective and accessible dynamics, one uses the serrated object on the right side of Fig.~\ref{fig:NetworkvsComb}, where the inaccessible auxiliary systems are hidden. 
Importantly, two networks implement the same supermap iff their combs are the same~\cite{Chiribella2008,Chiribella_2008,Chiribella2009}. This justifies writing $ J_{\mathcal{S}_N}$ in Eq.~\eqref{eq:NetworkLinkProduct}. An operator that is a comb satisfies certain constraints that guarantee that the corresponding supermap can be implemented by a network of quantum channels. As shown in Refs.~\cite{Chiribella2008,Chiribella2009}, an operator $J_{\mathcal{S}_N} \in \mathcal{L} ( \otimes_{j=0}^{2N+1} \mathcal{H}_j )$ corresponds to a deterministic supermap $\mathcal{S}_N$ with definite causal order  with input systems labeled by $\{0,2,\ldots, 2N\}$ and output systems $\{1,3,\ldots, 2N+1\}$ iff there exists a sequence of positive semidefinite operators $J_{\mathcal{S}_N}^{(j)} \in \mathcal{L} ( \otimes_{j=0}^{2j+1} \mathcal{H}_j ) $ with $0\leq j \leq N$  and $J_{\mathcal{S}_N}=J_{\mathcal{S}_N}^{(N)}$ that satisfy the normalization conditions
\begin{subequations}\label{eq:TPcondition}
    \begin{align}
    & \partTr{2j+1}{J_{\mathcal{S}_N}^{(j)}}=\id_{2j} \otimes J^{(j-1)}_\mathcal{N}  \quad \forall\, 1\leq j \leq N  \\
    & \partTr{1}{J_{\mathcal{S}_N}^{(0)}}=\id_0.
\end{align}
\end{subequations}
These conditions are semidefinite constraints and, most importantly, describe the action of a supermap via a single object rather than a product of (Choi representations of) channels. The set of all such quantum combs is denoted as $ \Comb{\mathcal{H}_0,\ldots,\mathcal{H}_{2N} \to \mathcal{H}_1,\ldots,\mathcal{H}_{2N+1}}$. Throughout this work, we will only consider such supermaps with a definite causal order. The system labels appearing in our combs will thus always denote a specific physical system at a fixed point in time.

If we want to compose two supermaps $\mathcal{N}_N, \mathcal{M}_M$ we interlace them, for example as shown in Fig.~\ref{fig:CompositionExample}. The comb of the resulting supermap can be obtained via the so-called link product~\cite{Chiribella2009}. Let $N_{\In}$ and $N_{\Out}$ denote the set of all input systems of $\mathcal{N}_N$  respectively. Moreover, let $N_{\all}=N_{\In} \cup N_{\Out}$ and define the analogous quantities for $\mathcal{M}_M$.
The resulting comb is given by the link product defined as
\begin{align}\label{eq:defLinkProduct}
    J_{\mathcal{N}_N}*J_{\mathcal{M}_M}:=\partTr{ M_{\all} \cap N_{\all}}{\Big( \id_{ N_{\all} \backslash  M_{\all} }\otimes J_{\mathcal{M}_M}^{T_{ M_{\all}\cap N_{\all}}}\Big)\Big(J_{\mathcal{N}_N}\otimes \id_{ M_{\all} \backslash  N_{\all}}\Big)}.
\end{align}
Whenever we compose supermaps, we assume that this leads to a consistent causal order, i.e., outputs of one supermap cannot be used as inputs for another supermap at a previous (time) step.
\begin{figure}[ht]
    \centering
    \scalebox{0.7}{\includegraphics[width=1\textwidth]{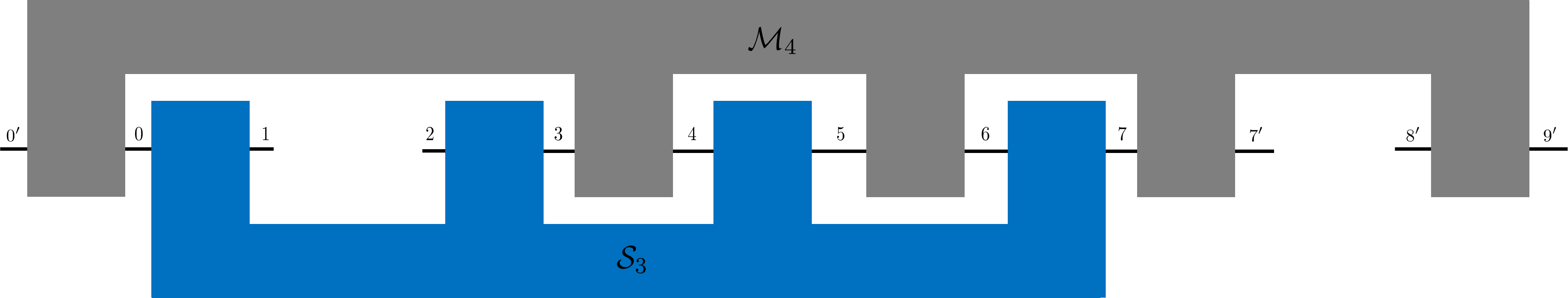}}
    \caption{Schematic depiction of a composition of two supermaps $\mathcal{S}_3$ with input and output systems given by $N_{\In}=\{0,2,4,6\}$ and $N_{\Out}=\{1,3,5,7\}$ and $\mathcal{M}_4$ with $M_{\In}=\{0^\prime, 3, 5,7,8^\prime \}$ and  $M_{\Out}=\{0,4,6,7^\prime,9^\prime\}$. The resulting supermap $\mathcal{W}_O$ has input and output systems $O_{\In}=\{0^\prime,2,8^\prime\}$ and $O_{\Out}=\{1, 7^\prime, 9^\prime\}$.}
    \label{fig:CompositionExample}
\end{figure}

\section{Incoherent quantum supermaps}\label{sec:IncohCombs}
Given a set of free states, a consistent set of free channels cannot create non-free states from the free ones. This is known as the Golden Rule of resource theories~\cite{Chitambar2019,Gour2024}. For the resource theory of coherence, this means that an incoherent channel $\mathcal{M}$ cannot transform an incoherent state into a coherent one, i.e., for any $\sigma \in \I$, we have that $\mathcal{M}(\sigma) \in \I.$ The set of free channels is not unique and different choices are well-motivated in different physical settings, see Ref.~\cite{Streltsov2017} for an overview. Importantly, any set of free operations must be closed under arbitrary compositions in a complete sense including sequential and parallel composition, i.e., $\mathcal{N} \circ \mathcal{M}$ and $\mathcal{N} \otimes \mathcal{M}$ must be free whenever $\mathcal{N}$ and $\mathcal{M}$ are free and have suitable input/output systems, but, e.g., also $\left(\mathcal{S}^{C\leftarrow B_1}\otimes\idChannel^{B_2}\right)\circ \mathcal{T}^{B_1,B_2\leftarrow A}$ must be free whenever $\mathcal{S}$ and $\mathcal{T}$ are. This requirement stems again from the idea that combining free objects should not lead to something valuable. The largest set of channels that satisfies these consistency conditions for coherence theory and which encompasses all other sets of incoherent channels is known as the maximally incoherent operations (MIO), see Refs.~\cite{Aberg2006,Liu2017,Diaz2018}. This set consists of all channels that satisfy the simple condition $\mathcal{M}\circ\Delta=\Delta\circ\mathcal{M}\circ\Delta$. In the following, we will consider MIO as free channels because this will allow us to use the provided coherence in an optimal manner.

As mentioned in the previous section, an optimal strategy for phase estimation will, in general, require the usage of a quantum supermap. To investigate optimal phase estimation with restricted coherence, we must thus also divide the set of supermaps into free and non-free ones. Like for quantum channels, a Golden Rule for free supermaps is that they cannot create resources for free and that the set of supermaps has to be consistent in the sense that it is closed under arbitrary compositions in a complete sense, i.e., even if we concatenate free supermaps only partially. In the context of coherence theory, consider the following Definition from the main text, which we repeat for readability.

\begin{defin}\label{def:MIONetworkSM}     
    A quantum supermap $\mathcal{S}_N$ corresponds to a MIO network and is considered free if its action can be decomposed into a sequence of MIO channels $\mathcal{M}_0,\ldots, \mathcal{M}_N \in \MIO$, i.e., $\mathcal{S}_N[\mathcal{N}_1,\ldots, \mathcal{N}_N]=\mathcal{M}_{N}(\idChannel \otimes \mathcal{N}_N)\mathcal{M}_{N-1}\ldots \mathcal{M}_1 (\idChannel \otimes \mathcal{N}_1)\mathcal{M}_0$ for all channels $\mathcal{N}_1,\ldots, \mathcal{N}_N$.  
\end{defin}

According to the following Proposition, this is a consistent choice of the free supermaps.
\begin{prop}
    The set of MIO networks is closed under arbitrary compositions.
\end{prop}
\begin{proof}
    Composing two supermaps associated with a MIO network is equivalent to interlacing two MIO networks as in Eq.~\eqref{eq:defLinkProduct}. Since the set MIO is closed under arbitrary sequential and parallel compositions (and the fact that the identity channel is in MIO), it immediately follows that the resulting network is again a MIO network.
\end{proof}
This implies in particular that a MIO network can neither create coherence nor transform channels in MIO into channels that are not in MIO. Moreover, this choice of free operations is intuitive in an operational sense: The networks that we can build out of MIO channels and thus for free are exactly the MIO networks. Note that free supermaps can also be constructed in this manner for other sets of free operations, not necessarily related to coherence. However, the set of MIO networks seems to be difficult to characterize. We thus consider the potentially larger set of MIO-compatible superchannels defined in the main text, which we repeat for readability in the Choi representation here.

\begin{defin}\label{def:DeltaConditionsNetworkSM} 
A quantum supermap $\mathcal{S}_N$ is called MIO-compatible if its Choi state $J_{\mathcal{S}_N}$ satisfies
\begin{align}
    \Delta_{0,2,\ldots, 2j} J_{\mathcal{S}_N}= \Delta_{0,2,\ldots, 2j} \Delta_{1,3,\ldots, 2j+1}J_{\mathcal{S}_N} \quad  \forall  j: 0\leq j \leq N,
\end{align}
\end{defin}

As we show in the following Proposition, the set of MIO-compatible supermaps indeed contains the set of MIO networks.


\begin{prop}\label{prop:MIOnetworkSubsetMIOCompatible}
    The set of MIO networks in Definition~\ref{def:MIONetwork} is a subset of the set of MIO-compatible supermaps in Definition~\ref{def:DeltaConditionsNetworkSM}.
\end{prop}
\begin{proof}
    That every MIO network is also MIO-compatible is essentially clear by inspection. Nevertheless, we will provide a technical proof in the following: Consider a network comprised of a sequence of channels $\mathcal{M}_0,\ldots, \mathcal{M}_N \in \MIO$, where we employ our usual system labeling, i.e.,  $\mathcal{M}_j^{2j+1,A_{j+1} \leftarrow 2j, A_j}$ for $0 \leq j\leq N$ and $A_0$ and $A_{N+1}$ are trivial systems. Recall that the comb of the supermap $\mathcal{S}_N$ defined by this network is given by 
    \begin{align}
        J_{\mathcal{S}_N}&=\partTr{ A_1,\ldots, A_N}{\left( \id \otimes J_{\mathcal{M}_0}^{T_{A_1}}\right) \left( \id \otimes J_{\mathcal{M}_1}^{T_{A_1,A_2}}\right) \ldots \Big( \id \otimes J_{\mathcal{M}_N}\Big)}.
    \end{align}
    Recall also that $\mathcal{M}_0,\ldots, \mathcal{M}_N \in \MIO$ is equivalent to $\Delta_0 J_{\mathcal{M}_0}= \Delta_0 \Delta_{1, A_1}J_{\mathcal{M}_0},\ \Delta_{2N,A_N} J_{\mathcal{M}_N}= \Delta_{2N,A_N} \Delta_{2N+1}J_{\mathcal{M}_0}$ and $\Delta_{2j, A_j} J_{\mathcal{M}_j}= \Delta_{2j, A_j} \Delta_{2j+1, A_{j+1}}J_{\mathcal{M}_j}$ for all other $j$. The same holds for the partially transposed Choi states. This allows us to check the conditions in Definition~\ref{def:DeltaConditionsNetworkSM} starting from the first one 
    \begin{align}
        \Delta_0 J_{\mathcal{S}_N} &= \partTr{ A_1,\ldots, A_N}{\left( \id \otimes \Delta_0 J_{\mathcal{M}_0}^{T_{A_1}}\right) \left( \id \otimes J_{\mathcal{M}_1}^{T_{A_1,A_2}}\right) \ldots \Big( \id \otimes J_{\mathcal{M}_N}\Big)} \nonumber \\
        &= \partTr{ A_1,\ldots, A_N}{\left( \id \otimes \Delta_0 \Delta_{1,A_1} J_{\mathcal{M}_0}^{T_{A_1}}\right) \left( \id \otimes J_{\mathcal{M}_1}^{T_{A_1,A_2}}\right) \ldots \Big( \id \otimes J_{\mathcal{M}_N}\Big)} \nonumber \\
        &=\Delta_0 \Delta_1 J_{\mathcal{S}_N}.
    \end{align}
    We now proceed by induction. Assume that $\Delta_{0,\ldots, 2j} J_{\mathcal{S}_N}=\Delta_{0,\ldots, 2j} \Delta_{1,\ldots, 2j+1} J_{\mathcal{S}_N}$ for some $j$. Then, 
    \begin{align}
        \Delta_{0,\ldots, 2j+2} J_{\mathcal{S}_N}&= \Delta_{0,\ldots, 2j+2}\partTr{ A_1,\ldots, A_N}{\left( \id \otimes  J_{\mathcal{M}_0}^{T_{A_1}}\right) \!\ldots\! \left( \id \otimes J_{\mathcal{M}_j}^{T_{A_j,A_{j+1}}}\right) \left( \id \otimes J_{\mathcal{M}_{j+1}}^{T_{A_{j+1},A_{j+2}}}\right)\!\ldots\! \Big( \id \otimes J_{\mathcal{M}_N}\Big)} \nonumber \\
        &= \partTr{ A_1,\ldots, A_N}{\left( \id \otimes  \Delta_0 J_{\mathcal{M}_0}^{T_{A_1}}\right) \!\ldots\! \left( \id \otimes \Delta_{2j} J_{\mathcal{M}_j}^{T_{A_j,A_{j+1}}}\right) \left( \id \otimes \Delta_{2j+2}J_{\mathcal{M}_{j+1}}^{T_{A_{j+1},A_{j+2}}}\right)\!\ldots\! \Big( \id \otimes J_{\mathcal{M}_N}\Big)} \nonumber\\
        &=  \Delta_{0,\ldots, 2j} \Delta_{1,\ldots, 2j+1} \partTr{ A_1,\ldots, A_N}{\left( \id \otimes  \Delta_{0}J_{\mathcal{M}_0}^{T_{A_1}}\right)\! \ldots\! \left( \id \otimes \Delta_{2j+2,A_{j+1}} J_{\mathcal{M}_{j+1}}^{T_{A_{j+1},A_{j+2}}}\right)\!\ldots \!\Big( \id \otimes J_{\mathcal{M}_N}\Big)} \nonumber  \\
        &= \Delta_{0,\ldots, 2j+2} \Delta_{1,\ldots, 2j+3} J_{\mathcal{S}_N},
    \end{align}
where we used the induction hypothesis in the second-to-last line. Thus, we find that $\Delta_{0,2,\ldots, 2j} J_{\mathcal{S}_N}= \Delta_{0,2,\ldots, 2j} \Delta_{1,3,\ldots, 2j+1}J_{\mathcal{S}_N}$ for all $ 0\leq j \leq N.$
\end{proof}

Lastly, the set of MIO-compatible supermaps is closed under arbitrary compositions (that respect causal ordering), see the following Lemma. This guarantees that a MIO-compatible supermap cannot transform MIO channels into non-MIO channels, thereby ensuring that this set is resource-theoretically consistent too.

\begin{lem}\label{lem:ClosedUnderComposition}
    The set of MIO-compatible supermaps in Definition~\ref{def:DeltaConditionsNetworkSM} is closed under \textit{arbitrary} concatenations that respect causality.
\end{lem}
\begin{proof}
 That the set of MIO-compatible supermaps is closed under arbitrary compositions becomes clear by, for example, looking at Fig.~\ref{fig:CompositionExample} and assuming that both supermaps in the composition are MIO-compatible:
    By applying a dephasing to the first input of the resulting supermap, $\tilde{0}$, we obtain a dephasing on $0$ too because $\mathcal{M}$ is MIO-compatible. Since $\mathcal{S}$ is also MIO-compatible, this implies dephasing on $1$ too, and we showed that dephasing on the first input of the resulting supermap implies dephasing on its first output. Now applying a dephasing to $\tilde{0}$ and $2$, i.e., the first two inputs of the resulting supermap, and using the fact that both $\mathcal{M}$ and $\mathcal{S}$ are MIO-compatible, we can recursively add dephasing channels to show that this implies a dephasing channel on system $\tilde{7}$ too. Since this recursive method can be applied to arbitrary concatenations of supermaps, the Lemma follows.
\end{proof}

\section{MIO-compatible supermaps}
Recall that MIO is the largest set of channels that is consistent with incoherent states in the sense that it maps incoherent states to incoherent states and is closed under arbitrary compositions that lead to a new channel. In the following, we will see that the set of MIO-compatible supermaps has a similar property, which endows this set with an operational meaning that justifies its name. To this end, we will need the following Lemma and Corollary. 

\begin{lem}\label{lem:ModifiedMIOTester}
Let $1,2,B$ be quantum systems and let $\ket\psi$ and $\ket\phi$ be normalized states of sytem B such that
\begin{align}
    \ketbra{\psi}{\psi}+\ketbra{\phi}{\phi}=\Delta \left(\ketbra{\psi}{\psi}+\ketbra{\phi}{\phi} \right).
\end{align}
Let $d_X$ denotes the dimension of Hilbert space $X$ and let
\begin{align}
    D_{k,l}=\left( \id -\ketbra{k}{k}-\ketbra{l}{l}\right)_1 \otimes \frac{\id_2}{d_2} \otimes \frac{\id_B}{d_B}.
\end{align}
Then, for every $m,n$ and $k\ne l$, the matrices
     \begin{align}
         &J_\mathcal{N}=  D_{k,l} + \frac{1}{2} \left(\ketbra{kn}{kn}+\ketbra{lm}{lm}\right)_{1,2} \otimes \left(\ketbra{\psi}{\psi}+\ketbra{\phi}{\phi}\right)_B +\frac{1}{2}\left(\ketbra{kn}{lm}+\ketbra{lm}{kn}\right)_{1,2} \otimes \left(\ketbra{\psi}{\psi}-\ketbra{\phi}{\phi}\right)_B\\
         &J_\mathcal{M} = D_{k,l}+ \frac{1}{2}\left(  \ketbra{kn}{kn}  + \ketbra{lm}{lm} \right)_{1,2} \otimes  \left(\ketbra{\psi}{\psi}+\ketbra{\phi}{\phi}\right)_{B} +\frac{1}{2} \left( \ketbra{kn}{lm}-\ketbra{lm}{kn} \right)_{1,2} \otimes \left( \ketbra{\psi}{\phi}-\ketbra{\phi}{\psi}\right)_B
    \end{align}
define the Choi states of channels $\mathcal{N}^{2,B \leftarrow 1} \in \MIO$ and $\mathcal{M}^{2,B \leftarrow 1} \in \MIO$ respectively for all choices of $n,m$. 
\end{lem}
\begin{proof}
For some $0\leq k,l \leq d_1-1$ and $0\leq n,m \leq d_2-1$, let
\begin{align}
    & \ket{\Phi_{k,l,n,m}^{\pm}}_{1,2}=\frac{1}{\sqrt{2}} \left( \ket{kn}_{1,2} \pm \ket{lm}_{1,2}\right), \label{eq:PhiKLNM} \\
    &\ket{\lambda_{k,l,n,m}}_{1,2,B}= \frac{1}{\sqrt{2}} \left(\ket{kn}_{1,2} \otimes \ket{\psi}_{B} + \ket{lm}_{1,2} \otimes \ket{\phi}_{B} \right),\\
    &\ket{\xi_{k,l,n,m}}_{1,2,B}=\frac{1}{\sqrt{2}} \left( \ket{kn}_{1,2} \otimes \ket{\phi}_{B} -\ket{lm}_{1,2} \otimes \ket{\psi}_{B} \right).
\end{align}
Note that
\begin{align}
    &2\ketbra{\Phi_{k,l,n,m}^{+}}{\Phi_{k,l,n,m}^{+}}_{1,2}= \ketbra{kn}{kn}_{1,2}+ \ketbra{lm}{lm}_{1,2} \pm \ketbra{kn}{lm}_{1,2} \pm \ketbra{lm}{kn}_{1,2}, \\
    &2 \ketbra{\lambda_{k,l,n,m}}{\lambda_{k,l,n,m}}_{1,2,B}=  \ketbra{kn}{kn}_{1,2} \!\otimes \ketbra{\psi}{\psi}_B\! +\! \ketbra{lm}{lm}_{1,2} \!\otimes\! \ketbra{\phi}{\phi}_B\! +\!\ketbra{kn}{lm}_{1,2}\! \otimes\! \ketbra{\psi}{\phi}_B\!+ \!\ketbra{lm}{kn}_{1,2} \!\otimes\! \ketbra{\phi}{\psi}_B,\nonumber
\end{align}
and analogously for $\ket{\xi_{k,l,n,m}}_{1,2}$. Combining these, a straightforward calculation reveals that
    \begin{align}
        &J_\mathcal{N}=  D_{k,l} +\ketbra{\Phi_{k,l,n,m}^{+}}{\Phi_{k,l,n,m}^{+}}_{1,2} \otimes \ketbra{\psi}{\psi}_B +\ketbra{\Phi_{k,l,n,m}^{-}}{\Phi_{k,l,n,m}^{-}}_{1,2}\otimes \ketbra{\phi}{\phi}_B \; ,\\
        &   J_\mathcal{M}=D_{k,l}+ \ketbra{\lambda_{k,l,n,m}}{\lambda_{k,l,n,m}}_{1,2,B}+\ketbra{\xi_{k,l,n,m}}{\xi_{k,l,n,m}}_{1,2,B}.
    \end{align}
Thus, $J_\mathcal{N},J_\mathcal{M}\geq 0$ for any $k\neq l$. Moreover, we have $\partTr{2}{J_\mathcal{N}}=\partTr{2}{J_\mathcal{M}}$ and
    \begin{align}
        \partTr{2,B}{J_\mathcal{N}}&=\left( \id -\ketbra{k}{k}-\ketbra{l}{l}\right)_1+  \ketbra{k}{k}_1 +\ketbra{l}{l}_1  =\id_1.
    \end{align}
    Lastly, if $k\neq l$, then $\Delta_1 J_\mathcal{N}=\Delta_1 J_\mathcal{M}$, and if also $\Delta\left(\ketbra{\psi}{\psi}+\ketbra{\phi}{\phi}\right)=\ketbra{\psi}{\psi}+\ketbra{\phi}{\phi}$, then
    \begin{align}
        \Delta_1 J_\mathcal{N} &=\left( \id -\ketbra{k}{k}-\ketbra{l}{l}\right)_1 \otimes \frac{\id_2}{d_2} \otimes \frac{\id_B}{d_B}+\frac{1}{2}\left(\ketbra{kn}{kn}+\ketbra{lm}{lm}\right)_{1,2} \otimes \left(\ketbra{\psi}{\psi}+\ketbra{\phi}{\phi}\right)_B \nonumber \\
        &= \left( \id -\ketbra{k}{k}-\ketbra{l}{l}\right)_1 \otimes \frac{\id_2}{d_2} \otimes \frac{\id_B}{d_B}+\frac{1}{2}\left(\ketbra{kn}{kn}+\ketbra{lm}{lm}\right)_{1,2} \otimes \Delta \left(\ketbra{\psi}{\psi}+\ketbra{\phi}{\phi}\right)_B \nonumber\\
        &= \Delta_1 \Delta_{2,B}J_\mathcal{N}.
    \end{align}
    Thus, $\mathcal{N}^{2,B \leftarrow 1},\mathcal{M}^{2,B \leftarrow 1} \in \MIO$.
\end{proof}

\begin{corollary}\label{cor:SuperMapMIOTesters}
Let $B_j$ denote quantum systems. Let
\begin{align}\label{eq:PsiPm}
    \ket{\psi^\pm}= \frac{1}{\sqrt{2}} \left( \ket{i_{B_1} k_{B_4}\ldots k_{B_{2N}}}\pm  \ket{j_{B_1} l_{B_4}\ldots l_{B_{2N}}}\right)_{B_1,B_4,\ldots, B_{2N}}.
\end{align}
Then, the channels $\mathcal{N}^{2,B_1,B_4\ldots, B_{2N} \leftarrow 1}$ and $\mathcal{M}^{2,B_1,B_4\ldots, B_{2N} \leftarrow 1}$ defined by Choi states $J_\mathcal{N}$ and $J_\mathcal{M}$ as in Lemma~\ref{lem:ModifiedMIOTester} with $\psi$ and $\phi$ chosen as $\psi^\pm$, are contained in MIO for all $i,j,k_{B_4},\ldots, k_{B_{2N}},l_{B_4},\ldots, l_{B_{2N}}$.
\end{corollary}
\begin{proof}
Consider the channels from Lemma~\ref{lem:ModifiedMIOTester}. Choose the system $B$ to be composed of $B_1,B_4,\ldots B_{2N}$ itself and choose $\ket{\psi}$ and $\ket{\phi}$ as $\ket{\psi^\pm}$ as in Eq.~\eqref{eq:PsiPm}, respectively. Clearly,  
\begin{align}
    \ketbra{\psi^+}{\psi^+}+\ketbra{\psi^-}{\psi^-}&= \ketbra{i_{B_1} k_{B_4}\ldots k_{B_{2N}}}{i_{B_1} k_{B_4}\ldots k_{B_{2N}}}+\ketbra{j_{B_1} l_{B_4}\ldots l_{B_{2N}}}{j_{B_1} l_{B_4}\ldots l_{B_{2N}}} \nonumber \\
    &=\Delta \left(\ketbra{\psi^+}{\psi^+}+\ketbra{\psi^-}{\psi^-} \right).
\end{align}
\end{proof}

We are now ready to prove our result concerning the operational interpretation.
\begin{thm}\label{thm:MIOCompatible}
    Let $\mathfrak{S}_0$ be the set of all MIO channels (with arbitrary input and output spaces). Moreover, for $N\ge 1$, let $\mathfrak{S}_N$ be a subset of the supermaps acting on $N$ channels (again with arbitrary input and output spaces), and let $\mathfrak{S}(\mathfrak{S}_N)=\bigcup_{N=0}^\infty \mathfrak{S}_N$.
    If $\mathfrak{S}(\mathfrak{S}_N)$ is closed under arbitrary compositions, then any  $\mathcal{S}_N\in\mathfrak{S}_N$ satisfies
    \begin{align}\label{eq:MIOCompFromMaxSet}
        \Delta_{0,2,\ldots, 2j} J_{\mathcal{S}_N}=\Delta_{0,2,\ldots, 2j} \Delta_{1,3,\ldots, 2j+1} J_{\mathcal{S}_N} \quad \forall j: 0\leq j\leq N.
    \end{align}
\end{thm}
\begin{proof}
    The requirement that $\mathfrak{S}(\mathfrak{S}_N)$ is closed under arbitrary compositions, combined with the fact that it contains all MIO channels implies that all MIO networks are included in $\mathfrak{S}(\mathfrak{S}_N)$. We now show that $\Delta_{0,2,\ldots, 2j} J_{\mathcal{S}_N}=\Delta_{0,2,\ldots, 2j} \Delta_{1,3,\ldots, 2j+1} J_{\mathcal{S}_N}$ holds for any $j:0\leq j\leq N-1$ by inserting the MIO network depicted in Fig.~\ref{fig:SupermapMIOCompatible} into the supermap $\mathcal{S}_N$. Let $B_j$ denote a copy of system $j$ (fixed by the supermap $\mathcal{S}_N$). Consider the MIO channels $\mathcal{N}^{2j+2,B_{2j+1},B_{2j+4}\ldots, B_{2N} \leftarrow 1}$ and $\mathcal{M}^{2j+2,B_{2j+1},B_{2j+4}\ldots, B_{2N}}$ from Corollary~\ref{cor:SuperMapMIOTesters}, where we simply replaced the systems $1,2$ by $2j+1,2j+2$. Let $\mathcal{E}_j$ denote the channel that acts as identity from system $2j-1$ to $B_{2j-1}$ and $B_{2j}$ to $2j$, see Fig.~\ref{fig:SupermapMIOCompatible} where $\mathcal{E}_1$ is plugged in the first slot of $\mathcal{S}_N$.
    Now we construct the two networks comprised of the sequence of channels
\begin{align}
    & \mathcal{E}_1,\ldots, \mathcal{E}_j, \mathcal{N}^{2j+2,B_{2j+1},B_{2j+4}\ldots, B_{2N} \leftarrow 2j+1}, \mathcal{E}_{j+1},\ldots, \mathcal{E}_N
\end{align}
and
\begin{align}
    &\mathcal{E}_1,\ldots, \mathcal{E}_j, \mathcal{M}^{2j+2,B_{2j+1},B_{2j+4}\ldots, B_{2N} \leftarrow 2j+1}, \mathcal{E}_{j+1},\ldots, \mathcal{E}_N.
\end{align}
Here we contract over the systems enumerated by $B_{2j+4},\ldots, B_{2N}$ as depicted in Fig.~\ref{fig:SupermapMIOCompatible}. From now on, we suppress the system indices in the channels $\mathcal{N}$ and $\mathcal{M}$.

\begin{figure}[ht]
    \centering
    \scalebox{1}{\includegraphics[width=1\linewidth]{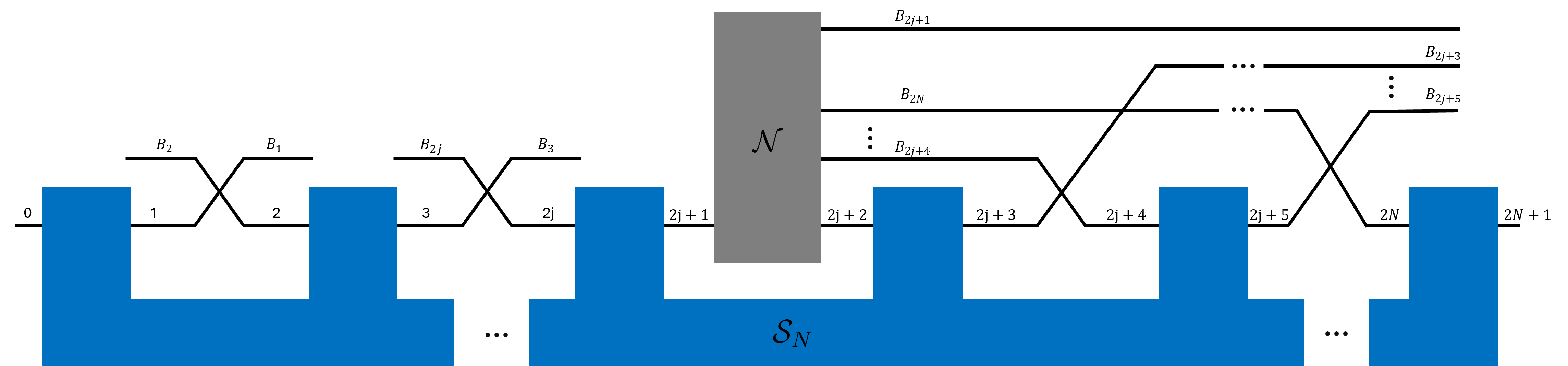}}
    \caption{Composition of the given supermap $\mathcal{S}_N$ with a MIO network comprised of the channels $\mathcal{E}_j$, and the MIO channels  $\mathcal{N}$ (and $\mathcal{M}$ in the place of $\mathcal{N}$) from Corollary~\ref{cor:SuperMapMIOTesters}.}
    \label{fig:SupermapMIOCompatible}
\end{figure}
Note that the channels  $\mathcal{N}$ and $\mathcal{M}$ are in $\MIO$ if the two integers $k\neq l$ (which are arbitrary but fixed for the definition of $\mathcal{N}$ and $\mathcal{M}$) according to Lemma~\ref{lem:ModifiedMIOTester} and Corollary~\ref{cor:SuperMapMIOTesters}. Obviously, the  channels $\mathcal{E}_j$ are contained in MIO too, and thus, such a network is a MIO network. Hence, the composition of the supermap $\mathcal{S}_N$ with each of these networks (either with the channel $\mathcal{N}$ or $\mathcal{M}$ respectively), which we denote by $\mathcal{T}$ and $\mathcal{W}$ respectively, must be a MIO channel if $\mathcal{M}$ and $\mathcal{N}$ were (which is the case if $k\neq l$). This is equivalent to
\begin{subequations}
        \begin{align}
        &\Delta_{0,B_2,\ldots B_{2j}} J_\mathcal{T}=  \Delta_{0,B_2,\ldots B_{2j}}  \Delta_{B_1,\ldots, B_{2j-1}}\Delta_{B_{2j+1}} \Delta_{B_{2j+3},B_{2j+5}\ldots, B_{2N-1}, 2N+1} J_\mathcal{T}, \\
        &\Delta_{0,B_2,\ldots B_{2j}} J_\mathcal{W}=  \Delta_{0,B_2,\ldots B_{2j}} \Delta_{B_1,\ldots, B_{2j-1}}\Delta_{B_{2j+1}} \Delta_{B_{2j+3},B_{2j+5}\ldots, B_{2N-1}, 2N+1} J_\mathcal{W}.
    \end{align}
\end{subequations}
Applying $\left(\idChannel_{B_{2j+1}}-\Delta_{B_{2j+1}}\right)$ to this expression yields
\begin{subequations}
        \begin{align}
        &\left(\idChannel_{B_{2j+1}}-\Delta_{B_{2j+1}}\right) \Delta_{0,B_2,\ldots B_{2j}} J_\mathcal{T}= 0, \\
         &\left(\idChannel_{B_{2j+1}}-\Delta_{B_{2j+1}}\right) \Delta_{0,B_2,\ldots B_{2j}} J_\mathcal{W}= 0.
    \end{align}
\end{subequations}
    By a straightforward calculation, it is easy to check that the composition with a channel $\mathcal{E}_j$ amounts to renaming system labels, i.e., relabeling system $j$ as $B_j$. Using this fact in the definition of the link product in Eq.~\eqref{eq:defLinkProduct} to express $J_\mathcal{T}$ and $J_\mathcal{W}$ respectively, we can insert this into the latter expression to obtain
    \begin{subequations}
    \begin{align}
        &0= \left(\idChannel_{B_{2j+1}}-\Delta_{B_{2j+1}}\right) \Delta_{0,B_2,\ldots B_{2j}} \partTr{2j+1,2j+2,B_{2j+4},\ldots, B_{2N}}{\left( \id  \otimes \left( J_{\mathcal{N}}\right)^{T_{2j+1,2j+2,B_{2j+4},\ldots,B_{2N}}} \right)\left( \id \otimes J_{\mathcal{S}_N}\right)}, \label{eq:TOffDiagsNetwork} \\
        &0=\left(\idChannel_{B_{2j+1}}-\Delta_{B_{2j+1}}\right) \Delta_{0,B_2,\ldots B_{2j}} \partTr{2j+1,2j+2,B_{2j+4},\ldots, B_{2N}}{\left( \id  \otimes \left( J_{\mathcal{M}}\right)^{T_{2j+1,2j+2,B_{2j+4},\ldots,B_{2N}}} \right)\left( \id \otimes J_{\mathcal{S}_N}\right)},\label{eq:WOffDiagsNetwork}
    \end{align}  
    \end{subequations}
    where we suppressed the system indices of $J_{\mathcal{S}_N}$ and the identities for enhanced readability. Next, note that the only contribution on system $B_{2j+1}$ originates from the Choi state of $J_\mathcal{M}$ and $J_\mathcal{N}$ respectively, and are given by
    \begin{subequations}
             \begin{align}
         &\left(\idChannel_{B_{2j+1}}\!-\!\Delta_{B_{2j+1}}\right) \!J_\mathcal{N}^{T_{2j+1,2j+2,B_{2j+4},\ldots,B_{2N}}} \!=\!\frac{1}{2}\! \left( \ketbra{kn}{lm}\!+\!\ketbra{lm}{kn} \right)_{2j+1,2j+2}^{T} \!\otimes \!\left( \ketbra{\psi^+}{\psi^+}\!-\!\ketbra{\psi^-}{\psi^-}\right)_{B_{2j+1},B_{2j+4},\ldots, B_{2N}}^{T_{B_{2j+4},\! \ldots \!, B_{2N}}}, \\
         &\left(\idChannel_{B_{2j+1}}\!-\!\Delta_{B_{2j+1}}\right) \!J_\mathcal{M}^{T_{2j+1,2j+2,B_{2j+4},\ldots,B_{2N}}} \!=\!\frac{1}{2}\! \left( \ketbra{kn}{lm}\!-\!\ketbra{lm}{kn} \right)_{2j+1,2j+2}^{T} \!\otimes \!\left( \ketbra{\psi^+}{\psi^-}\!-\!\ketbra{\psi^-}{\psi^+}\right)_{B_{2j+1},B_{2j+4},\ldots, B_{2N}}^{T_{B_{2j+4},\! \ldots \!, B_{2N}}},
    \end{align}
    \end{subequations}
where $ \ket{\psi^\pm}= \frac{1}{\sqrt{2}} \left( \ket{i_{B_{2j+1}} k_{B_{2j+4}}\ldots k_{B_{2N}}}\pm  \ket{j_{B_{2j+1}} l_{B_{2j+4}}\ldots l_{B_{2N}}}\right)_{B_{2j+1},B_{2j+4},\ldots, B_{2N}}$. Thus, 
\begin{subequations}
    \begin{align}
    &\ketbra{\psi^+}{\psi^+}-\ketbra{\psi^-}{\psi^-}= \ketbra{i_{B_{2j+1}} k_{B_{2j+4}}\ldots k_{B_{2N}}}{j_{B_{2j+1}} l_{B_{2j+4}}\ldots l_{B_{2N}}} + \text{h.c.} \, , \\
    &\ketbra{\psi^+}{\psi^-}-\ketbra{\psi^-}{\psi^+}= \ketbra{j_{B_{2j+1}} l_{B_{2j+4}}\ldots l_{B_{2N}}}{i_{B_{2j+1}} k_{B_{2j+4}}\ldots k_{B_{2N}}} -\text{h.c.} \, .
\end{align}
\end{subequations}
Inserting this into Eqs.~\eqref{eq:TOffDiagsNetwork} and ~\eqref{eq:WOffDiagsNetwork} yields
\begin{subequations}
        \begin{align}
        0&= \Delta_{0,B_2,\ldots B_{2j}} \bra{k_{B_{2j+4}}\ldots k_{B_{2N}}} \left( \bra{l m}_{2j+1,2j+2} J_{\mathcal{S}_N} \ket{kn}_{2j+1,2j+2}  - \bra{k n}_{2j+1,2j+2}J_{\mathcal{S}_N} \ket{lm}_{2j+1,2j+2} \right)\ket{l_{B_{2j+4}}\ldots l_{B_{2N}}}, \\
        0&= \Delta_{0,B_2,\ldots B_{2j}} \bra{k_{B_{2j+4}}\ldots k_{B_{2N}}} \left( \bra{l m}_{2j+1,2j+2} J_{\mathcal{S}_N} \ket{kn}_{2j+1,2j+2}  +\bra{k n}_{2j+1,2j+2}J_{\mathcal{S}_N} \ket{lm}_{2j+1,2j+2} \right)\ket{l_{B_{2j+4}}\ldots l_{B_{2N}}}.
    \end{align}
\end{subequations}
Adding the latter two expression yields
    \begin{align}
        0&= \Delta_{0,B_2,\ldots B_{2j}} \bra{ l m k_{B_{2j+4}}\!\ldots\!k_{B_{2N}}}_{2j+1,2j+2, B_{2j+4},\ldots B_{2N}}   J_{\mathcal{S}_N} \ket{kn l_{B_{2j+4}}\ldots l_{B_{2N}}}_{2j+1,2j+2, B_{2j+4},\ldots B_{2N}} \quad \forall k\neq l,
    \end{align}
    and for all other indices $n,m, k_{B_{2j+4}},l_{B_{2j+4}} \ldots  k_{B_{2N}},l_{B_{2N}} $. After renaming all systems $B_j$ into $j$ again, this is equivalent to
    \begin{align}
        \left(\idChannel_{{2j+1}}-\Delta_{{2j+1}}\right) \Delta_{0,2,\ldots {2j}}  J_{\mathcal{S}_N}=0.
    \end{align}
Since this holds for arbitrary $0\leq j\leq N-1$, this is equivalent to   $\Delta_{0,2,\ldots, 2j} J_{\mathcal{S}_N}=\Delta_{0,2,\ldots, 2j} \Delta_{1,3,\ldots, 2j+1} J_{\mathcal{S}_N} \quad \forall j: 0\leq j\leq N-1.$ Lastly, the condition that $\Delta_{0,\ldots 2N}J_{\mathcal{S}_N} =\Delta_{0,\ldots 2N} \Delta_{1,\ldots 2N+1} J_{\mathcal{S}_N}$ follows since $\mathcal{S}_N$ must be MIO channel from all its input to all its output systems. More technically, this follows by composing the supermap $\mathcal{S}_N$ with a MIO network comprised only of independent SWAP channels.
\end{proof}
Since Eq.~\eqref{eq:MIOCompFromMaxSet} in the above Theorem is exactly the requirement that each $\mathcal{S}_N$ is MIO-compatible, we can invoke Lemma~\ref{lem:ClosedUnderComposition} to conclude that Eq.~\eqref{eq:MIOCompFromMaxSet} is also sufficient for closedness under arbitrary compositions (that respect causal ordering). Moreover, it is straightforward to see that by composing arbitrary MIO-compatible supermaps in such a way that one obtains a channel, one will always end up with a MIO channel. This shows that the MIO-compatible supermaps are the largest set of supermaps that contains precisely MIO as channels and is closed under arbitrary compositions. Using this ``compatibility construction", MIO, i.e., a choice of free channels, fixes the free states ($\I$ is the set of states that correspond to a MIO channel with trivial input) and the free supermaps, and thus a resource theory encompassing all quantum objects. After the network approach, this is the second method by which one can construct a general resource theory purely from the choice of free channels. That there exists such a maximal set is a priori not clear, and it is an interesting open question whether this is the case too in other resource theories (not necessarily related to coherence).

\section{(Completely) MIO-preserving superchannels}\label{sec:MIOpreserving}
In this section, we consider the special case of superchannels, i.e., we consider supermaps $\mathcal{S}_1$. We show that the set of MIO-compatible superchannels is equal to the set of completely MIO-preserving superchannels, which are defined as follows (see Refs.~\cite{Gour2020,Gour2021b} for a discussion of complete resource preservation).

\begin{defin}\label{def:cMIOP}
   A superchannel $\mathcal{S}_1^{1,3\leftarrow 0,2}$ is called completely MIO-preserving (cMIOP) iff it maps every MIO channel into a MIO channel in a complete sense, i.e., for all systems $A,B$
   \begin{align}
       \mathcal{S}_1^{1,3\leftarrow 0,2}[\mathcal{M}^{2,B\leftarrow 1,A}] \in \MIO \quad \forall \mathcal{M}^{2,B\leftarrow 1,A} \in \MIO.
    \end{align}
\end{defin}

\begin{figure}[ht]
    \centering
    \scalebox{0.7}{\includegraphics[width=1\linewidth]{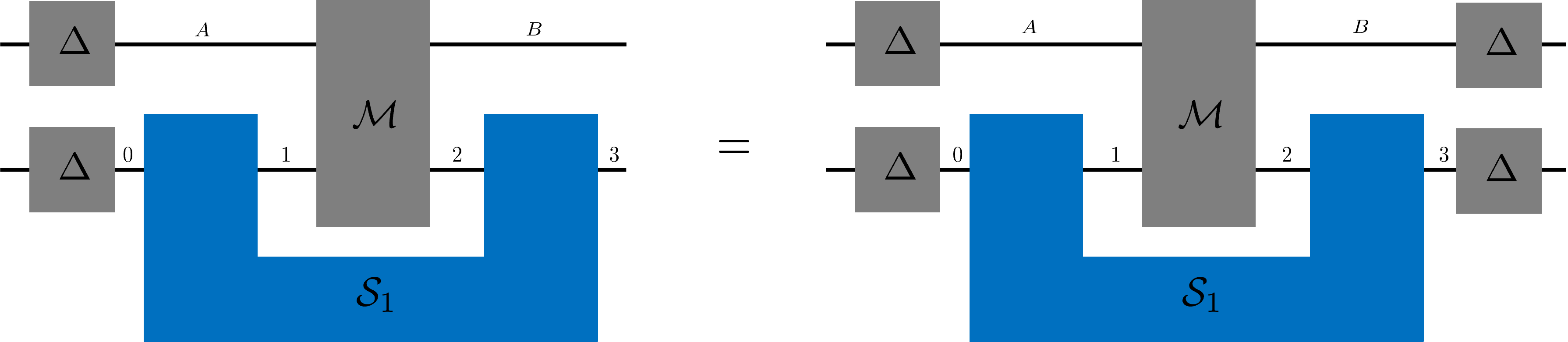}}
    \caption{A superchannel $\mathcal{S}_1$ is completely MIO-preserving iff it satisfies the depicted property for all systems $A,B$  and all $\mathcal{M} \in \MIO$. }
    \label{fig:cMIOPSuperchannel}
\end{figure}
We now proceed to note that cMIOP is equal to MIO-compatible as introduced previously in Definition~\ref{def:DeltaConditionsNetworkSM}.
\begin{corollary}\label{cor:cMIOP}
    The set of completely MIO-preserving superchannels is equal to the MIO-compatible superchannels, i.e., a superchannel $\mathcal{S}_1$ is completely MIO-preserving iff
    \begin{subequations}
    \begin{align}
        &\Delta_0 J_{\mathcal{S}_1}=\Delta_0 \Delta_{1} J_{\mathcal{S}_1}, \\
        &\Delta_{0,2} J_{\mathcal{S}_1}=\Delta_{0,2} \Delta_{1,3}  J_{\mathcal{S}_1}.
    \end{align}
    \end{subequations}
\end{corollary}
\begin{proof}
According to Theorem~\ref{thm:MIOCompatible}, these conditions are necessary. According to Lemma~\ref{lem:ClosedUnderComposition}, they are also sufficient.
\end{proof}

Next, we want to discuss what happens if we do not require closure under \textit{arbitrary} composition of free objects. To this end, we consider the set of MIO-preserving superchannels (see also, e.g., Ref.~\cite{Liu2019} for other such non-complete frameworks) defined as follows.

\begin{defin}\label{def:MIOP}
    A superchannel $\mathcal{S}_1^{1,3\leftarrow 0,2}$ is called MIO-preserving (MIOP) iff
    \begin{align}
       \mathcal{S}_1^{1,3\leftarrow 0,2}[\mathcal{M}^{2\leftarrow 1}] \in \MIO \quad \forall \mathcal{M}^{2\leftarrow 1} \in \MIO.
   \end{align}
\end{defin}

\begin{figure}[ht]
    \centering
    \scalebox{0.7}{\includegraphics[width=1\linewidth]{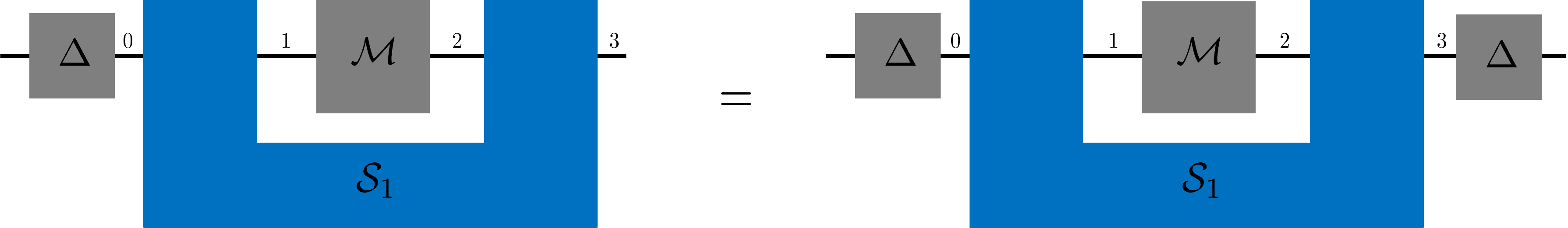}}
    \caption{A superchannel $\mathcal{S}_1$ is MIO-preserving iff it satisfies the depicted property for all $\mathcal{M} \in \MIO$.
    }
    \label{fig:MIOPSuperchannel}
\end{figure}
We now want to demonstrate that we can extract a maximally coherent state from a specific MIO-preserving superchannel if we allow to compose it with only subsystems of a MIO channel. Thereby, we show that the set of cMIOP superchannels is a strict subset of the MIOP superchannels and use this example to emphasize that resource theories of superchannels (-maps) should be completely free. To this end, let $\mathcal{N}^{1,A\leftarrow 0}(\cdot)=\Tr{\cdot} \Phi_{1,A}^+$, where $\ket{\Phi^{\pm}}= \tfrac{1}{\sqrt{2}} (\ket{00}\pm\ket{11})$, and let $\mathcal{K}^{3\leftarrow 2,A}(\cdot)=\sum_i \Tr{M_i (\cdot)} \ketbra{i}{i}_3 $, for some POVM $\{M_i \}_i.$ The superchannel $\mathcal{S}_1$ defined by these two channels as pre- and post-processing channels, respectively, is depicted in Fig.~\ref{fig:MIOPExample}.

\begin{figure}[ht]
    \centering
    \scalebox{0.5}{\includegraphics[width=1\linewidth]{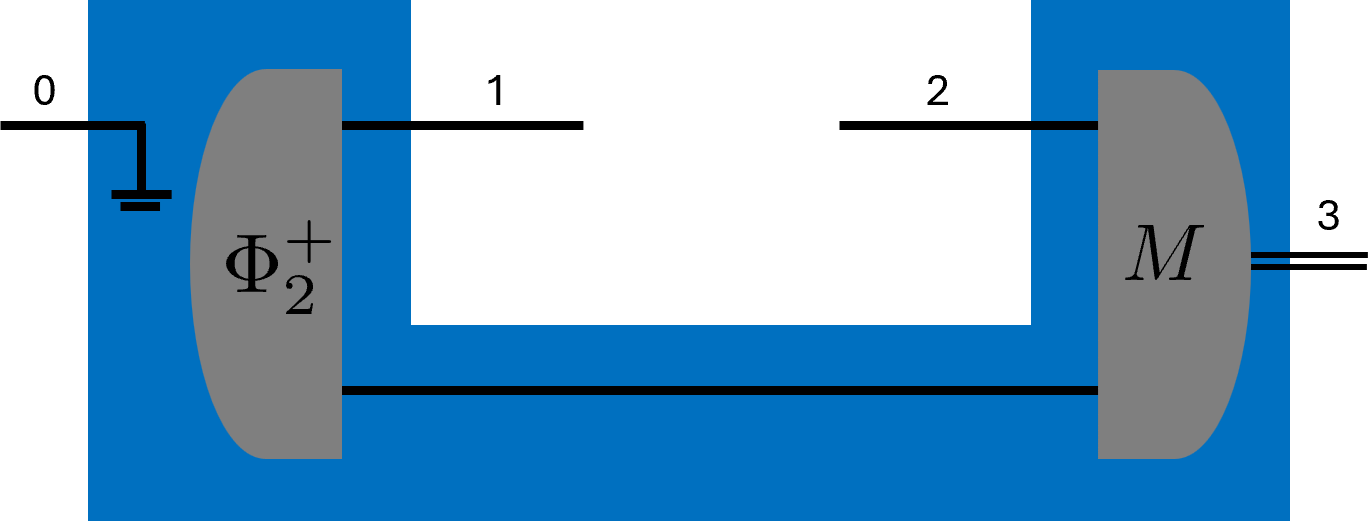}}
    \caption{An example of a MIO-preserving (but not completely MIO-preserving) supermap.}
    \label{fig:MIOPExample}
\end{figure}

Its Choi representation is given by
\begin{align}
    J_{\mathcal{S}_1}&=\partTr{A}{\Big( \id_{2,3} \otimes J_{\mathcal{N}}^{T_A} \Big)\Big( \id_{0,1} \otimes J_\mathcal{K}\Big)} \nonumber  \\
    &= \frac{1}{2}\id_0 \otimes \sum_{n,m} \ketbra{n}{m}_1 \otimes \sum_{k,l} \ketbra{k}{l}_2 \otimes \sum_i \ketbra{i}{i}_3 \Tr{M_i \ketbra{kn}{lm}_{2,A}},
\end{align}
Obviously, the superchannel defined by $J_{\mathcal{S}_1}$ is MIO-preserving, but not completely MIO-preserving since $\Delta_0 J_{\mathcal{S}_1}\neq \Delta_0\Delta_1 J_{\mathcal{S}_1} $, which is necessary according to Corollary~\ref{cor:cMIOP}. Even more so, the superchannel is resource destroying in the sense that for any channel $\mathcal{M}^{2\leftarrow 1}$ (contained in MIO or not), the resulting channel $\mathcal{S}_1[\mathcal{M}^{2\leftarrow 1}] \in \MIO$. Nonetheless, we can extract a maximally coherent bit from this superchannel if we insert a MIO channel $\mathcal{M}^{2,B \leftarrow 1}$ into it. To this end, let $M_i$ be the projective measurement on the Bell states between systems $2$ and $A$. This results in a superchannel (in the Choi representation) of
\begin{align}\label{eq:SuperChannelMIOpNOTcMIOP}
    J_{\mathcal{S}_1}= \frac{1}{2} \id_0 \otimes \Big( \ketbra{\Phi^+}{\Phi^+}_{1,2} \otimes \ketbra{0}{0}_3 +\ketbra{\Phi^-}{\Phi^-}_{1,2} \otimes \ketbra{1}{1}_3 +\ketbra{\Psi^+}{\Psi^+}_{1,2} \otimes \ketbra{2}{2}_3+\ketbra{\Psi^-}{\Psi^-}_{1,2} \otimes \ketbra{3}{3}_3\Big),
\end{align}
where $\ket{\Psi^\pm}=\tfrac{1}{\sqrt{2}} (\ket{01}\pm\ket{10})$. Let us now define a channel $\mathcal{M}^{2,B \leftarrow 1}$ in its Choi representation as
\begin{align}
    J_\mathcal{M}= \ketbra{\Phi^+}{\Phi^+}_{1,2} \otimes \ketbra{+}{+}_B +\ketbra{\Phi^-}{\Phi^-}_{1,2} \otimes \ketbra{-}{-}_B.
\end{align}
 According to Lemma~\ref{lem:ModifiedMIOTester}, this defines the Choi state of a channel $\mathcal{M}^{2,B \leftarrow 1} \in \MIO.$ Inserting this channel into the superchannel above we find
\begin{align}
    J_{\mathcal{S}_1[\mathcal{M}]}&= \partTr{1,2}{\Big(\id_{0,3} \otimes \left(J_\mathcal{M}\right)^{T_{1,2}} \Big) \Big(\id_B \otimes J_{\mathcal{S}_1}\Big)}= \frac{1}{2} \id_0 \otimes \left(\ketbra{+}{+}_B\otimes \ketbra{0}{0}_3+\ketbra{-}{-}_B\otimes\ketbra{1}{1}_3\right).
\end{align}
This corresponds to the replacement channel
\begin{align}
    \mathcal{S}_1[\mathcal{M}]=\frac{1}{2}\left(\ketbra{+}{+}_B\otimes \ketbra{0}{0}_3+\ketbra{-}{-}_B\otimes\ketbra{1}{1}_3\right) \partTr{0}{\cdot}.
\end{align}
If we now perform a projective measurement in the incoherent basis on system $3$, and a classically controlled $Z$ gate based on the measurement outcome on system $B$, we obtain a perfect bit of coherence.

Thus, under compositions on subsystems only, the MIO-preserving (but not completely MIO-preserving) superchannel above can be used to extract a bit of coherence by applying it only to subsystems of the channel $\mathcal{M}.$ This implies that the set MIOP is not closed under \textit{arbitrary} compositions. The fact that we can extract a maximally coherent state from this superchannel highlights the need for closure under arbitrary compositions, and the usage of completely free supermaps is thus crucial in the context of resource theories on quantum supermaps. Therefore, this justifies our initial choice of free supermaps. As we will see later, using the set of MIO-preserving superchannels as free superchannels in phase estimation protocols trivializes the whole problem in the sense that MIO-preserving superchannels can be arbitrarily powerful for phase estimation, \textit{even though we do not apply them to subsystems only}. To leap ahead, we find that for a single copy of the phase gate $V_\phi^{(2)}$, i.e., a qubit unitary encoding the phase $\phi$ and the Holevo cost function given by $C(\phi)=4\sin^2(\phi/2)$, the superchannel $\mathcal{S}_1$ depicted in Fig.~\ref{fig:MIOPExample} achieves the optimal average cost without coherence constraints regardless of the (fixed) supplied coherence. We will address this rigorously after introducing the application to phase estimation in Section~\ref{sec:PhaseEstimationMIOpreserving}.

\section{Phase estimation in quantum algorithms}
We now briefly review why phase estimation serves as a crucial subroutine in many quantum algorithms such as~\cite{Shor1997,Kitaev1995, Harrow2009}. For a comprehensive and rigorous treatment of the topic, see for example Refs.~\cite{Nielsen2010,Cleve1998}, which we follow in this section. Phase estimation as presented in Ref.~\cite{Cleve1998} is formulated as follows. Suppose that we have access to multiple copies of a black-box unitary $U$ and an eigenvector $\ket{\psi}$ of this unitary, such that $U \ket\psi=e^{i\phi} \ket{\psi}.$ The goal is to estimate the phase $\phi$. One can use the quantum circuit depicted in Fig.~\ref{fig:PhaseEstimationCircuitShor} to construct phase estimates of the phase $\phi$.
\begin{figure}[ht]
    \centering
    \scalebox{0.7}{\includegraphics[width=1\linewidth]{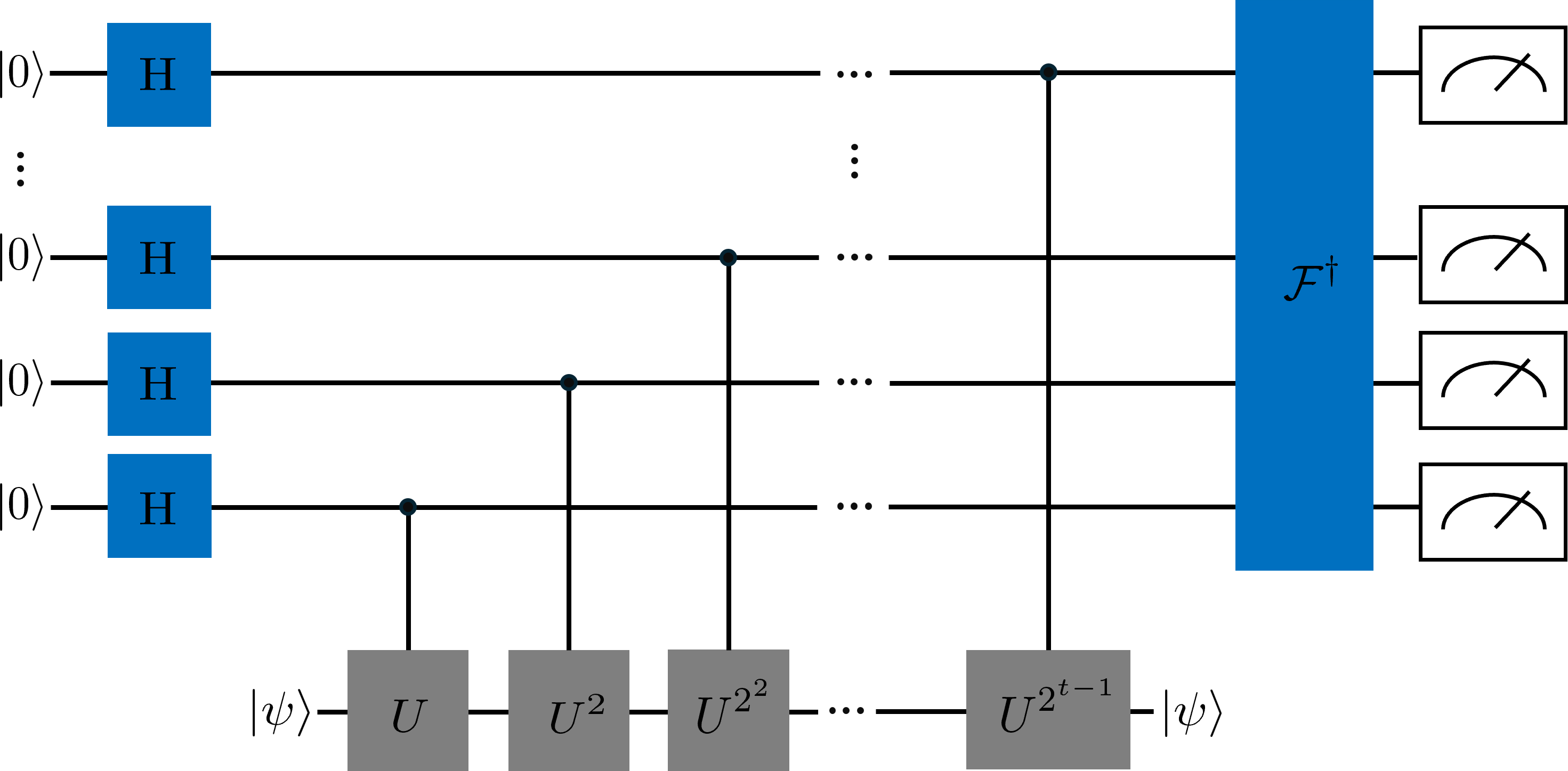}}
    \caption{A phase estimation protocol that prepares $t$ copies of a maximally coherent state via the Hadamard gate H, which are followed by controlled unitaries on an auxiliary system. A measurement is performed in the Fourier basis, i.e., we apply the inverse Fourier transform and measure each qubit in the computational basis, to deduce a phase estimate.}
    \label{fig:PhaseEstimationCircuitShor}
\end{figure}
For any outcome $0\leq k\leq 2^t-1$ (in binary) of the measurement in the circuit, we construct a phase estimate of $\hat{\phi}_k= \frac{2\pi k}{2^t}.$ The accuracy of each estimate and the probability to measure the corresponding outcome depends on the number of qubits $t$ (and also on the number of total copies of the black-box unitary $U$, given by $2^{t}-1$). As shown in Ref.~\cite{Cleve1998}, one obtains a measurement outcome $k$, such that $k$ is the best $t$-bit approximation to $\phi$, i.e., $\left|\phi-\frac{2\pi k}{2^t}\right| <\frac{2\pi}{2^{t+1}}$, with a probability of at least $\frac{4}{\pi^2} \approx 0.405$. Moreover, the probability to obtain the best $t$-bit approximation to $\phi$ can be amplified to arbitrary $1-\epsilon$, by using a total of at least $t+ \lceil \log\left(1+\frac{1}{2\epsilon} \right)\rceil$ qubits, see Refs.~\cite{Cleve1998,Nielsen2010}. Thus, the protocol in Fig.~\ref{fig:PhaseEstimationCircuitShor} can be used to estimate the phase (corresponding to an eigenstate of a given unitary $U$) with arbitrary accuracy and success probability with a suitable choice of circuit size. If one uses phase estimation as a subroutine, the desired accuracy and success probability depends on the specific problem one seeks to solve via phase estimation.\\

We now briefly discuss the example of Shor's factoring algorithm~\cite{Shor1997} and refer to Ref.~\cite{Cleve1998} for the complete discussion. Here, the factorization problem of a composite number $N$ is solved via order-finding, i.e., finding the smallest integer $r$ such that for some (randomly chosen) coprime integer $a$ to $N$ we have $a^r \text{mod} \, N = 1$, which allows to find a factor of $N$ with sufficiently high probability, see~\cite{Shor1997}. This order-finding problem can be solved via phase estimation. To this end, let $U$ define the unitary with $U\ket{n}=\ket{an \, \text{mod}\, N}$. The orthonormal eigenstates and eigenvalues of $U$ are
\begin{align}
    \lambda_k= \exp\left( 2\pi i \tfrac{k}{r}\right), \quad  \ket{\psi_k}= \frac{1}{\sqrt{r}} \sum_{j=0}^{r-1} e^{-2\pi i k \frac{j}{r}} \ket{ x^j \text{mod} N}.
\end{align}
For the moment, assume that we have access to a copy of $ \ket{\psi_1}= \frac{1}{\sqrt{r}} \sum_{j=0}^{r-1} e^{-2\pi i \frac{j}{r}} \ket{ x^j \text{mod} N}$. Our goal is to estimate the phase of its corresponding eigenvalue $\phi=  \frac{2\pi }{r}$. Using the phase estimation algorithm in Fig.~\ref{fig:PhaseEstimationCircuitShor}, where we initialize the auxiliary system in the state $\ket{\psi}=\ket{\psi_1}$, and we choose the number of bits as $t=  \lceil \log N \rceil$, we obtain a measurement outcome $k$ from which we construct a phase estimate of $\hat{\phi}_k= \frac{2\pi k}{r}$ for which  $\left|\frac{2\pi}{r}-\frac{2\pi k}{2^t}\right| <\frac{2\pi}{2^{t+1}}$ with sufficiently high probability. 
For such a measurement outcome, there exists an (efficient) classical post-processing, namely the continued fraction algorithm, which allows us to deduce the order $r$ from the estimate, and thus, find a factor of $N$.
Recall that we assumed that we have access to the eigenstate $ \ket{\psi_1}= \frac{1}{\sqrt{r}} \sum_{j=0}^{r-1} e^{-2\pi i \frac{j}{r}} \ket{ x^j \text{mod} N}$ of the unitary $U$. Given we have no knowledge of $r$ to begin with, this state is hard to prepare in practice.  However, this problem can be circumvented by noticing that
\begin{align}
    \ket{1}=\frac{1}{\sqrt{r}} \sum_{j=0}^{r-1} \ket{\psi_j}.
\end{align}
If we initialize the auxiliary system in Fig.~\ref{fig:PhaseEstimationCircuitShor} with the state $\ket{1}$ rather than $\ket{\psi_1}$, the protocol produces estimates of randomly chosen eigenvalue $\lambda_j =\frac{2\pi j }{r}$ for $0\leq j\leq r-1$. If $j$ and $r$ share a common factor, the classical post-processing will not return $r$, but only a factor of $r$. Fortunately, this is sufficiently rare, and the phase estimation algorithm in Fig.~\ref{fig:PhaseEstimationCircuitShor} combined with the classical post-processing returns the order $r$ with sufficiently high probability to find a factor of $N$. For more details, see Refs.~\cite{Shor1997,Cleve1998}.

Lastly, we want to emphasize that the formulation of phase estimation discussed in this section is equivalent to the one presented in the main text. Note, that the protocol in Fig.~\ref{fig:PhaseEstimationCircuitShor} uses a total of $2^t-1$ copies of the black-box unitary. In total, each basis state of the combined $t$-qubit register undergoes a phase shift of $\ket k \mapsto  e^{i\phi k} \ket{k}$ for any $0\leq k \leq 2^t-1$. This is equivalent to writing the entire middle part of Fig.~\ref{fig:PhaseEstimationCircuitShor} as a unitary $V_\phi^{(M)}$, where $M=2^t$. Moreover, we saw that the cost function relevant to Shor's algorithm is a window function with a width determined by the classical post-processing. Estimates that are close enough to the true value of the phase are useful for the classical post-processing, and thus, do not have to be penalized by the cost function, while others are too inaccurate to be helpful and thus have to be penalized.

\section{Optimal phase estimation with limited coherence}\label{sec:optimalPhaseEstimation}
We consider a scenario where we have access to $N$ copies of a unitary 
\begin{align}
    V_\phi^{(d)}=\sum_{n=0}^{d-1} e^{i\phi n} \ketbra{n}{n},
\end{align}
which encode the phase $\phi$ we wish to estimate. For example, this could be the phase imprinted to our system register via the modular exponentiation in Shor's algorithm, as described in the previous section. The most general way to utilize these unitaries is to interlace them with a sequence of quantum channels, i.e., to plug them into a quantum network as described above and in the main text. In particular, we consider quantum networks comprised of a sequence of MIO channels as in Definition~\ref{def:MIONetwork}. By construction, such networks cannot create coherence, which is necessary to perform phase estimation. Therefore, we need to supply resources, which we do in the form of a resource state $\rho$. It is convenient to describe the network of MIO operations via its Choi representation introduced above. Given the sequence of channels $\mathcal{M}_0,\ldots, \mathcal{M}_N$ the Choi state can be constructed using the link product from Eq.~\eqref{eq:NetworkLinkProduct}.

A phase estimation protocol as depicted in Fig.~\ref{fig:PhaseEstimationProtocolsSM} yields a measurement output $x$, for which we construct a phase estimate according to the assignment rule $x\mapsto \hat{\phi}_x$. The accuracy of such an estimate directly impacts the success of a quantum algorithm that uses phase estimation as a subroutine, and we thus wish to infer estimates that are as close to the true value as possible. Since phase estimation has diverse applications in quantum technologies, the impact of the accuracy on the success of the specific algorithm varies. To assess the performance of phase estimation in a unified manner, we use a cost function $C$ that penalizes deviation of an estimates $\hat{\phi}$ from the true value $\phi$ according to $C(\phi-\hat{\phi})$. We assume that $C\in\mathcal{L}_1([0,2\pi)$) is integrable, $2\pi$-periodic, and non-negative (this can be trivially relaxed to the cost function being bounded from below). 

To assess the performance of a phase estimation protocol, we will use the \textit{minimal average cost} with respect to any  fixed cost function that satisfies the above properties and assuming that our prior distribution of $\phi$ is uniform on $[0,2\pi)$. 
\begin{figure}[ht]
    \centering
   \scalebox{0.8}{\includegraphics[width=1\linewidth]{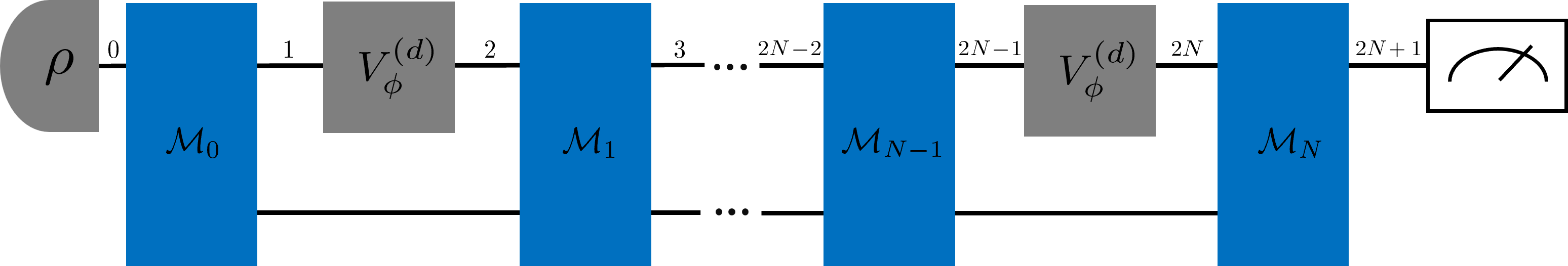}}
    \caption{The quantum part of a phase estimation protocol composed of a network of MIO operations $\mathcal{M}_j$ probing $N$ copies of $V_\phi^{(d)}$, an input state $\rho$, and a measurement.}
    \label{fig:PhaseEstimationProtocolsSM}
\end{figure}
As stated in Eq.~\eqref{eq:AvgCostMainText} of the main text, the minimal average cost is given by 
\begin{align}\label{eq:AvgCostSM}
    C_{\min}^{(d,N)}(\rho)=  \inf \sum_{x} \int_{0}^{2\pi}\dPhi  C(\phi -\hat{\phi}_x) \, p^{(d,N)}(x|\phi,\rho).
\end{align}
The probability distribution $p^{(d,N)}(x|\phi,\rho)$ can be conveniently expressed in the Choi representation. For that, we need the Choi states of the two networks given by the sequence of MIO channels and the identical unitaries $V_\phi^{(d)}$. To this end, let us introduce the systems $\odd=\{1,3,\ldots, 2N-1\}$ and $\even=\{2,4,\ldots, 2N\}$, where the system $0$ through which we supply coherence and the system $2N+1$ which we measure are excluded. Let 
\begin{align}\label{eq:UnitaryExpandedForm}
    U_\phi:=U_\phi^{(d,N)}&=\bigotimes_{j=1}^N \left(V_\phi^{(d)}\right)_{2j-1}\otimes \id_{2j} = \sum_{n=0}^{d^N-1} e^{i\phi H_d(n)} \ketbra{n}{n}_{\odd}\otimes \mathbb{1}_{\even},
\end{align}
where $\ket{n}_{\odd}=\ket{n_1,\ldots , n_N}_{1,3,\ldots , 2N-1}$ with $0\leq n_k\leq d-1$ and $H_d(n)$ denotes the digit sum of the integer $n$, i.e., $H_d(n)=\sum_{k=1}^{N} n_k$. Since the Choi state of a composition of channels that do not share input and output states simply factors as a tensor product of Choi states (see for example Ref.~\cite{Chiribella2009}), the Choi state of the identical unitaries, which we denote by $J_\phi$, is given by
\begin{align}\label{eq:ChoiStateVs}
     J_\phi:=J_\phi^{(d,N)}= \bigotimes_{j=1}^{N} J_{V_\phi^{(d)}}= U_\phi^{(d,N)} \sum_{n,m=0}^{d^N-1} \ketbra{nn}{mm}_{{\odd},{\even}} \left(U_\phi^{(d,N)}\right)^\dagger,
\end{align}
where we suppress the superscript $(d,N)$ if clear form the context. Note that $J_\phi^T=J_{-\phi}$. Let us denote the Choi state of the MIO network as $\mathcal{K}$, where we use the convention to denote Choi states of supermaps with calligraphic letters, just as we do for channels. Using the 
link product, the conditional probability distribution $p^{(d,N)}(x|\phi,\rho)$ thus takes the form
\begin{align}
     p^{(d,N)}(x|\phi,\rho):= \Tr{\left( \rho^T \otimes \left(J_\phi^{(d,N)} \right)^T\otimes M_x\right) \mathcal{K}},
\end{align}
where $\{M_x \}_x$ is the POVM and we again omit the dependence on $(d,N)$ from now on. The optimization problem for the minimal average cost can thus be expressed as
\begin{subequations}\label{eq:MinimalAverageCostRho}
\begin{alignat}{2}
    C_{\min}^{(d,N)}(\rho)= &\inf   \quad && \sum_{x=0}^{\tilde{M}-1} \int \dPhi \,  C\left(\phi -\hat{\phi}_x\right) \Tr{\left( \rho^T \otimes J_\phi^T\otimes M_x\right) \mathcal{K}} \\
        & \suchthat && \tilde{M}\geq 1\\
        & && \lbrace \hat{\phi}_x \rbrace_{0\leq x\leq \tilde{M}-1} \subset [0,2\pi) \\
        & && \sum_{x=0}^{\tilde{M}-1} M_x =\mathbb{1} \\
        & && M_x \geq 0   \\
        & && \mathcal{K} \text{ a comb representation of a MIO network}.\label{eq:MIONetworkConstraint}
\end{alignat}
\end{subequations}
To summarize, we optimize over an arbitrary MIO network as in Fig.~\ref{fig:PhaseEstimationProtocolsSM}, which we represent using the Choi representation of the network, schematically depicted in Fig.~\ref{fig:PhaseEstimationProtocolAsComb},
that processes our resource state $\rho$ and the encoding via the unitaries, the POVM with an unbounded number of measurement outcomes $\tilde{M}$, and all possible assignment rules for the phase estimates. 

\begin{figure}[ht]
    \centering
    \includegraphics[width=1\linewidth]{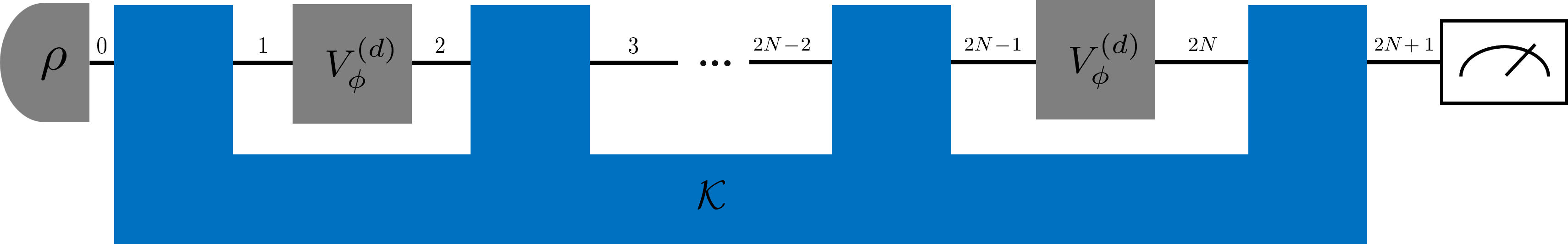}
    \caption{A quantum comb representation $\mathcal{K}$ of the phase estimation protocol comprised of the MIO network in Fig.~\ref{fig:PhaseEstimationProtocolsSM}.}
    \label{fig:PhaseEstimationProtocolAsComb}
\end{figure}

\section{Proof of Theorem~\ref{thm:OptAvgCost} from the main text}
In this section, we provide the proof of Theorem~\ref{thm:OptAvgCost} presented in the main text, which we restate here for readability.

\begin{thm}\label{thm:OptAvgCostSM}
   For any $(d,N)$, let $M=(d-1)N+1$ and let $Y^{(M)}\in \mathbb{C}^{M\times M}$ denote the Toeplitz matrix defined by 
   \begin{align}\label{eq:CostMatrixDefinitionSM}
       Y^{(M)}= \sum_{n,m=0}^{M-1}\int \dPhi C(\phi) e^{i\phi(n-m)}\ketbra{n}{m}.
   \end{align}
   Then, the minimal average cost is given by
   \begin{align}
       C_{\min}^{(d,N)}(\rho)&=  C_{\min}^{(M,1)}(\rho)= \min_{\mathcal{M} \in \MIO} \Tr{Y^{(M)} \mathcal{M}(\rho)},
   \end{align}
   where the input and output dimensions of the channel $\mathcal{M}$ are fixed by $\rho$ and $Y^{(M)}$, respectively.
\end{thm}

As outlined in the main text, we will prove this theorem by establishing a lower bound on the minimal average cost and subsequently constructing an explicit protocol that achieves this bound. To this end, the first step is to relax a constraint of the optimization problem in Eq.~\eqref{eq:MinimalAverageCostRho}, i.e., we optimize over MIO-compatible combs instead of combs corresponding to a network comprised of MIO channels. Recalling the discussion in section~\ref{sec:IncohCombs}, a lower bound to Eq.~\eqref{eq:MinimalAverageCostRho} is thus given by
\begin{subequations}\label{eq:MinimalAverageCostRhoRelaxed}
\begin{alignat}{2}
    C_{\min}^{(d,N)}(\rho)\geq \tilde{C}_{\min}^{(d,N)}(\rho)=&\inf   \quad && \sum_{x=0}^{\tilde{M}-1} \int \dPhi \,  C\left(\phi -\hat{\phi}_x\right) \Tr{\left( \rho^T \otimes J_\phi^T\otimes M_x\right) \mathcal{K}} \\
        &\suchthat && \tilde{M}\geq 1\\
        & && \lbrace \hat{\phi}_x \rbrace_{0\leq x\leq \tilde{M}-1} \subset [0,2\pi) \\
        & && \sum_{x=0}^{\tilde{M}-1} M_x =\mathbb{1} \\
        & && M_x \geq 0   \\
        & && \mathcal{K} \in \Comb{\mathcal{H}_0,\ldots,\mathcal{H}_{2N} \to \mathcal{H}_1,\ldots,\mathcal{H}_{2N+1}}\\
        & &&\Delta_{0,2,\ldots, 2j} \mathcal{K}= \Delta_{0,2,\ldots, 2j} \Delta_{1,3,\ldots, 2j+1}\mathcal{K}\quad  \forall  j: 0\leq j \leq N \label{eq:DeltaConditionsRelaxed}.
\end{alignat}
\end{subequations}
Next, we cast the above optimization problem as an optimization over subcombs, i.e., positive-semidefinite operators that sum to a comb. That this is possible will be obvious for readers familiar with the comb framework. 
\begin{lemma}\label{lem:Subcombs}
    The relaxed minimal average cost in Eq.~\eqref{eq:MinimalAverageCostRhoRelaxed} can be written as an optimization problem over a collection of subcombs $\{\mathcal{K}_x \}_{0\leq x\leq \tilde{M}-1}$, i.e.,
    \begin{subequations}\label{eq:SubCombOptimizationProblem}
            \begin{alignat}{2}
                \tilde{C}_{\min}^{(d,N)}(\rho)= &\inf   \quad && \sum_{x=0}^{\tilde{M}-1} \int \dPhi \,  C\left(\phi -\hat{\phi}_x\right) \Tr{\left(\rho^T \otimes J_\phi^T \right) \mathcal{K}_x} \label{eq:SubcombObjectiveFunction}\\
                & \suchthat && \tilde{M}\geq 1\label{eq:CombCIConstraint}\\
                & && \lbrace \hat{\phi}_x \rbrace_{0\leq x\leq \tilde{M}-1} \subset [0,2\pi) \\
                & && \mathcal{K}_x \geq 0   \\
                & && \sum_{x=0}^{\tilde{M}-1} \mathcal{K}_x \in \Comb{\mathcal{H}_0,\ldots,\mathcal{H}_{2N} \to \mathcal{H}_1,\ldots,\mathcal{H}_{2N-1}} \\
                & &&   \Delta_{0,\ldots, 2j} \mathcal{K}_x = \Delta_{0,\ldots 2j} \Delta_{1,\ldots, 2j+1}\mathcal{K}_x \quad \forall x, \forall j: 0\leq j\leq N-1 \label{eq:DeltaConditionsRelaxedN-1}.
        \end{alignat}
    \end{subequations}

\end{lemma}
\begin{proof}
    For any feasible comb $\mathcal{K}$ and POVM $\lbrace M_x \rbrace$ in Eq.~\eqref{eq:MinimalAverageCostRhoRelaxed}, define $\mathcal{K}_x=\partTr{2N+1}{(\id \otimes M_x) \mathcal{K}}$. Since $\mathcal{K}_x\geq 0$ (as $M_x,\mathcal{K} \geq 0$) and $\sum_x \mathcal{K}_x=\partTr{2N+1}{(\id \otimes \sum_x M_x) \mathcal{K}} = \partTr{2N+1}{\mathcal{K}}=\id\otimes \mathcal{K}^{(N-1)}$, this is a collection of subcombs $\{ \mathcal{K}_x\}_{0\leq x\leq \tilde{M}-1}$. Moreover, the coherence constraints in Eq.~\eqref{eq:DeltaConditionsRelaxedN-1} hold (note that the indices only run until the second to last slot in contrast to Eq.~\eqref{eq:DeltaConditionsRelaxed}) since for any $0\leq j\leq N-1$
    \begin{align}
        \Delta_{0,\ldots, 2j} \mathcal{K}_x &= \partTr{2N+1}{(\id \otimes M_x) \Delta_{0,\ldots, 2j} \mathcal{K}}\overset{\eqref{eq:DeltaConditionsRelaxed}}{=}\partTr{2N+1}{(\id \otimes M_x) \Delta_{0,\ldots, 2j} \Delta_{1,\ldots, 2Nj+1}\mathcal{K}} \nonumber \\
        &=\Delta_{0,\ldots 2j} \Delta_{1,\ldots, 2j+1}\mathcal{K}_x.
    \end{align}
    Conversely, for any feasible collection of subcombs in Eq.~\eqref{eq:SubCombOptimizationProblem}, let $\mathcal{K}=\sum_x \mathcal{K}_x \otimes \ketbra{x}{x}_{2N+1}$, which clearly defines a valid comb with $\Delta_{0,2,\ldots, 2j} \mathcal{K}= \Delta_{0,2,\ldots, 2j} \Delta_{1,3,\ldots, 2j+1}\mathcal{K}$  for all  $ 0\leq j \leq N $. For a measurement $M_x=\ketbra{x}{x}$, we thus have $\Tr{(.)\mathcal{K}_x}=\Tr{\left((.)\otimes M_x \right) \mathcal{K}}$, which completes the proof.
\end{proof}

We now proceed to show that Eq.~\eqref{eq:SubCombOptimizationProblem} can be reduced to a semidefinite optimization problem (SDP). To this end, we make use of the unitary
    \begin{align}\label{eq:Utilde}
    \tilde{U}_\phi&:= \id_0 \otimes U_\phi\overset{\eqref{eq:UnitaryExpandedForm}}{=}\id_0\otimes \bigotimes_{j=1}^N \left(V_\phi^{(d)}\right)_{2j-1}\otimes \id_{2j} = \id_0 \otimes \sum_{n=0}^{d^N-1} e^{i\phi H_d(n)} \ketbra{n}{n}_{\odd}\otimes \mathbb{1}_{\even}.
    \end{align}
The key step here is to show via symmetry arguments (available because of a uniform prior and a cost function that only depends on the difference between estimates and true value), that for a fixed number of copies $N$ of the unitary, it is always sufficient to consider a finite number $M$ of measurement outcomes $x$ with associated uniformly distributed phase estimates $\hat{\phi}_x=\tfrac{2\pi x}{M}$.

\begin{lemma}\label{lem:CovariantCombs}
    Let $M:= (d-1)N+1$ and 
    \begin{align}
        X^{(d,N)}:=  \int_{0}^{2\pi}\dPhi C(\phi) {J_\phi^{(d,N)}}^T= \sum_{n,m=0}^{d^N-1} \ketbra{nn}{mm}_{\odd,\even} \int_{0}^{2\pi}\dPhi C(\phi) e^{-i\phi\left(H_d(n)-H_d(m)\right)} .
    \end{align}
    The relaxed minimal average cost in Eq.~\eqref{eq:SubCombOptimizationProblem} can be expressed as the semidefinite optimization problem
    \begin{subequations}\label{eq:MinimalCostCovariant}
    \begin{alignat}{2}
        \tilde{C}_{\min}^{(d,N)}(\rho)= &\min   \quad &&  M \Tr{\left( \rho^T \otimes X^{(d,N)}\right)\mathcal{K}_0} \\
            & \suchthat && \mathcal{K}_0\geq 0 \\
            & && \sum_{x=0}^{M-1} \tilde{U}_{\frac{2\pi x}{M}}^\dagger \mathcal{K}_0 \tilde{U}_{\frac{2\pi x}{M}}\in \Comb{\mathcal{H}_0,\ldots,\mathcal{H}_{2N} \to \mathcal{H}_1,\ldots,\mathcal{H}_{2N-1}} \\
            & && \Delta_{0,2,\ldots, 2j} \mathcal{K}_0= \Delta_{0,2,\ldots, 2j} \Delta_{1,3,\ldots, 2j+1}\mathcal{K}_0 \quad  \forall  j: 0\leq j \leq N-1.
    \end{alignat}
    \end{subequations}
\end{lemma}
\begin{proof}
    Starting from Lemma~\ref{lem:Subcombs}, we first note that, without loss of generality, we can restrict ourselves to optimize over measurements with $\tilde{M}\geq M=(d-1)N+1$ outcomes. If the minimum is achieved for fewer than $M$ outcomes, we can simply add outcomes that appear with probability zero (see also Ref.~\cite[Remark~16]{vanDam2007}). Next, we show that for a fixed $\tilde{M}$, the minimum of the optimization problem in Eq.~\eqref{eq:SubCombOptimizationProblem} is achieved on a set of subcombs and phase estimates $\lbrace (\mathcal{K}_x, \hat{\phi}_x) \rbrace_{0\leq x\leq \tilde{M}-1}$, where the subcombs are covariant, i.e., each 
    \begin{align}
        \mathcal{K}_x=  \tilde{U}_{\frac{2\pi x}{M}} \mathcal{K}_0 \tilde{U}_{\frac{2\pi x}{M}}^\dagger
    \end{align}
    for some $\mathcal{K}_0$, and each associated phase estimate is $\hat{\phi}_x  = \tfrac{2\pi x}{M}$. To finish the proof, we will then show that it is sufficient to consider $M=(d-1)N+1$ measurement outcomes. 
    
    Let $\lbrace \mathcal{L}_x\rbrace_{0\leq x\leq \tilde{M}-1}$ be an arbitrary collection of subcombs that is feasible in the optimization problem in Eq.~\eqref{eq:SubCombOptimizationProblem} and let 
     \begin{align}
      C(\mathcal{L}_x,\hat{\phi}_x) &=\int_{0}^{2\pi} \dPhi \, \sum_{x=0}^{\tilde{M}-1} \mathcal{C}(\phi-\hat{\phi}_x) \Tr{\left(\rho^T\otimes J_\phi^T \right) \mathcal{L}_x}.
    \end{align}
    This is the objective function evaluated on $\lbrace \mathcal{L}_x\rbrace_{0\leq x\leq \tilde{M}-1}$ and a specific choice of estimates. 

    We define another set of subcombs $\lbrace \mathcal{M}_x \rbrace_{0\leq x\leq \tilde{M}-1}$ by
    \begin{align}\label{eq:SubcombM}
     \mathcal{M}_x&:=\frac{1}{\tilde{M}}\sum_{r=0}^{\tilde{M}-1} \tilde{U}_{-\hat{\phi}_{x+r} +\tfrac{2\pi x}{\tilde{M}}}^\dagger\mathcal{L}_{x+r} \tilde{U}_{-\hat{\phi}_{x+r} +\tfrac{2\pi x}{\tilde{M}}},
    \end{align}
    where addition on indices is understood to be modulo $\tilde{M}$, i.e., $\mathcal{L}_{x+r}=\mathcal{L}_{(x+r)\,  \text{mod}  \tilde{M}}$ and $\hat{\phi}_{x+r}=\hat{\phi}_{(x+r)\,  \text{mod}  \tilde{M}}$.  We will show later in the proof that this defines feasible collection of subcombs.
    Using Eq.~\eqref{eq:ChoiStateVs} to rearrange phases and the periodicity of the cost function to express the integration variable as $\phi^\prime =\phi-\tfrac{2\pi x}{\tilde{M}}$, we obtain
    \begin{align}
        C(\mathcal{M}_x,\tfrac{2\pi x}{\tilde{M}})&\overset{\eqref{eq:ChoiStateVs}}{=} \frac{1}{\tilde{M}} \sum_{r=0}^{\tilde{M}-1} \int_{0}^{2\pi} \dPhi \, \sum_{x=0}^{\tilde{M}-1} C\left(\phi-\frac{2\pi x}{\tilde{M}}\right) \Tr{\left( \rho^T \otimes J_{-\frac{2\pi x }{\tilde{M}}+\phi +\hat{\phi}_{x+r} }^T \right)\mathcal{L}_{x+r} } \nonumber  \\
         &= \frac{1}{\tilde{M}} \sum_{r=0}^{\tilde{M}-1} \int_{0}^{2\pi} \dPhi' \, \sum_{x=0}^{\tilde{M}-1} C\left(\phi^\prime \right) \Tr{\left(\rho^T \otimes J_{\phi^\prime +\hat{\phi}_{x+r} }^T \right)\mathcal{L}_{x+r} }\nonumber  \\
         &\overset{y=x+r}{=} \frac{1}{\tilde{M}} \sum_{r=0}^{\tilde{M}-1} \int_{0}^{2\pi} \dPhi' \, \sum_{y=0}^{\tilde{M}-1} C\left(\phi^\prime \right) \Tr{\left( \rho^T \otimes J_{\phi^\prime +\hat{\phi}_{y} }^T \right)\mathcal{L} _{y} } \nonumber \\
         &=\int_{0}^{2\pi}\dPhi' \, \sum_{y=0}^{\tilde{M}-1} C\left(\phi^\prime \right) \Tr{\left(\rho^T \otimes J_{\phi^\prime +\hat{\phi}_{y} }^T \right)\mathcal{L}_{y} } \nonumber \\
         &= \int_{0}^{2\pi} \dPhi' \, \sum_{y=0}^{\tilde{M}-1} C\left(\phi^\prime-\hat{\phi}_{y} \right) \Tr{\left( \rho^T \otimes J_{\phi^\prime }^T \right)\mathcal{L}_{y} } \nonumber \\
         &=C(\mathcal{L}_x,\hat{\phi}_x).
    \end{align}
This shows that the objective function evaluated on the two pairs of subcombs and estimates is equal. Estimates $\hat{\phi}_x=\frac{2\pi x}{\tilde{M}}$ are thus optimal (for fixed $\tilde{M}$). Using the same method again, we construct another collection of subcombs $\lbrace \mathcal{K}_x \rbrace_{0\leq x\leq \tilde{M}-1}$ as
\begin{align}
    \mathcal{K}_x&:=\frac{1}{\tilde{M}}\sum_{r=0}^{\tilde{M}-1} \tilde{U}_{-\frac{2\pi r }{\tilde{M}}}^\dagger\mathcal{M}_{x+r}  \tilde{U}_{-\frac{2\pi r }{\tilde{M}}} \nonumber  \\
    &= \tilde{U}_{\frac{2\pi x}{\tilde{M}}}^\dagger  \left\lbrace\frac{1}{\tilde{M}} \sum_{y=0}^{\tilde{M}-1} \tilde{U}_{-\frac{2\pi y }{\tilde{M}}}^\dagger\mathcal{M}_{y} \tilde{U}_{-\frac{2\pi y}{\tilde{M}}} \right\rbrace \tilde{U}_{\frac{2\pi x}{\tilde{M}}} \nonumber \\
    &=: \tilde{U}_{\frac{2\pi x}{\tilde{M}}}^\dagger \mathcal{K}_0 \tilde{U}_{\frac{2\pi x}{\tilde{M}}}.
\end{align}
Using again Eq.~\eqref{eq:ChoiStateVs}, this leaves the objective function invariant too,
\begin{align}\label{eq:subcombK}
    C(\mathcal{K}_x,\tfrac{2\pi x}{\tilde{M}})=& \frac{1}{\tilde{M}} \sum_{r,x=0}^{\tilde{M}-1} \int_{0}^{2\pi}  \dPhi \,  C\left(\phi-\frac{2\pi x}{\tilde{M}}\right) \Tr{\left( \rho^T \otimes J_{\frac{2\pi r}{\tilde{M}}+\phi }^T \right)\mathcal{M}_{x+r} } \nonumber \\
    =& \frac{1}{\tilde{M}} \sum_{r,y=0}^{\tilde{M}-1} \int_{0}^{2\pi}  \dPhi \,  C\left(\phi-\frac{2\pi (y-r)}{\tilde{M}}\right) \Tr{\left(\rho^T \otimes J_{\phi+\frac{2\pi r}{\tilde{M}}}^T \right)\mathcal{M}_{y}}\nonumber \\
    =& \frac{1}{\tilde{M}} \sum_{r,y=0}^{\tilde{M}-1} \int_{0}^{2\pi}  \frac{\diff{\varphi}}{2\pi} \,  C\left(\varphi-\frac{2\pi y}{\tilde{M}}\right) \Tr{\left( \rho^T \otimes J_{\varphi}^T \right)\mathcal{M}_{y}}\nonumber \\
    =& C(\mathcal{M}_x,\tfrac{2\pi x}{\tilde{M}}).
\end{align}
where in the second line we made the substitution $y=x+r$, and in the second-to-last line we substituted $\varphi=\phi+\frac{2\pi r}{\tilde{M}}$. Moreover, with $X^{(d,N)}=\int \dPhi C(\phi) J_\phi^T$, it follows that 
\begin{align}
  C(\mathcal{K}_x,\tfrac{2\pi x}{\tilde{M}})&=  \sum_{x=0}^{\tilde{M}-1} \int_{0}^{2\pi} \dPhi C\left(\phi-\frac{2\pi x}{\tilde{M}}\right)  \Tr{\left( \rho^T \otimes J_\phi^T \right) \tilde{U}_{\frac{2\pi x}{\tilde{M}}}^\dagger  \mathcal{K}_0  \tilde{U}_{\frac{2\pi x}{\tilde{M}}}} \nonumber \\
  &=\sum_{x=0} ^{\tilde{M}-1}\int_{0}^{2\pi} \dPhi C\left(\phi-\frac{2\pi x}{\tilde{M}}\right) \Tr{\left( \rho^T \otimes U_{\frac{2\pi x}{\tilde{M}}} J_{\phi}^T U_{\frac{2\pi x}{\tilde{M}}}^\dagger \right)  \mathcal{K}_0 } \nonumber \\
  &\overset{\eqref{eq:ChoiStateVs}}{=} \sum_{x=0} ^{\tilde{M}-1}\int_{0}^{2\pi} \dPhi^\prime C(\phi^\prime)  \Tr{\left( \rho^T \otimes J_{\phi^\prime}^T \right)  \mathcal{K}_0 }\nonumber  \\
  &= \tilde{M} \Tr{\left( \rho^T \otimes X^{(d,N)} \right) \mathcal{K}_0}.
\end{align}
As promised, we now show that $\lbrace \mathcal{M}_x \rbrace$ and $\lbrace \mathcal{K}_x \rbrace$ are feasible subcombs (given that $\lbrace \mathcal{L}_x \rbrace$ is). By a direct calculation, it is easy to check that if $\mathcal{M}=\sum_x \mathcal{M}_x \in \Comb{\mathcal{H}_2,\ldots, \mathcal{H}_{2N} \to \mathcal{H}_1,\ldots, \mathcal{H}_{2N-1}}$, then
\begin{align}
    \mathcal{K}&=\sum_{x=0}^{\tilde{M}-1} \mathcal{K}_x= \frac{1}{\tilde{M}}\sum_{r=0}^{\tilde{M}-1} \tilde{U}_{-\frac{2\pi r}{\tilde{M}}}^\dagger \sum_{x=0}^{\tilde{M}-1} \mathcal{M}_{x+r}\tilde{U}_{-\frac{2\pi r}{\tilde{M}}} = \frac{1}{\tilde{M}}\sum_{r=0}^{\tilde{M}-1} \tilde{U}_{-\frac{2\pi r}{\tilde{M}}}^\dagger \mathcal{M} \tilde{U}_{-\frac{2\pi r}{\tilde{M}}} \nonumber \\
    &\in \Comb{\mathcal{H}_2,\ldots, \mathcal{H}_{2N} \to \mathcal{H}_1,\ldots, \mathcal{H}_{2N-1}}.
\end{align}
This follows immediately from the fact that $\tilde{U}$ factors and acts only non-trivially on systems enumerated by an odd index (see Eq.~\eqref{eq:Utilde}). Since the normalization conditions trace out odd systems, they are unaffected by such unitaries, and the collection of subcombs $\lbrace \mathcal{K}_x \rbrace$ is normalized if $\lbrace \mathcal{M}_x \rbrace$ was. Analogously, we can conclude that $\lbrace \mathcal{M}_x \rbrace$ defines a valid collection of subcombs if $\lbrace \mathcal{L}_x \rbrace$ does. By a direct calculation, we can see that for all $x$ and $0\leq j \leq N-1$
\begin{align}
     \Delta_{0,2,\ldots, 2j}\mathcal{M}_x &=  \frac{1}{\tilde{M}}\sum_{r=0}^{\tilde{M}-1} \tilde{U}_{-\hat{\phi}_{x+r} +\tfrac{2\pi x}{\tilde{M}}}^\dagger \Delta_{0,2,\ldots, 2j}\mathcal{L}_{x+r}\tilde{U}_{-\hat{\phi}_{x+r} +\tfrac{2\pi x}{\tilde{M}}} \nonumber \\
     &= \frac{1}{M}\sum_{r=0}^{\tilde{M}-1} \tilde{U}_{-\hat{\phi}_{x+r} +\tfrac{2\pi x}{\tilde{M}}}^\dagger  \Delta_{0,2,\ldots, 2j}\Delta_{1,3,\ldots, 2j+1}\mathcal{L}_{x+r}\tilde{U}_{-\hat{\phi}_{x+r} +\tfrac{2\pi x}{\tilde{M}}} \nonumber \\
     &=\frac{1}{\tilde{M}}\sum_{r=0}^{\tilde{M}-1}   \Delta_{0,2,\ldots, 2j}\Delta_{1,3,\ldots, 2j+1}\mathcal{L}_{x+r} \nonumber \\
     &=  \Delta_{0,2,\ldots, 2j}\Delta_{1,3,\ldots, 2j+1}\mathcal{M}_{x},
\end{align}
and analogously, the same holds for the $\mathcal{K}_x$. Since the $\mathcal{K}_x$ are covariant, enforcing this condition on $\mathcal{K}_0$ is sufficient.
The relaxed average cost is thus given by
\begin{subequations}\label{eq:SDPUnboundedM}
\begin{alignat}{2}
    \tilde{C}_{\min}^{(d,N)}= &\inf   \quad &&  \tilde{M} \Tr{\left( \rho^T \otimes X^{(d,N)}\right)\mathcal{K}_0} \\
        & \suchthat && \tilde{M}\geq M \\
         & && \mathcal{K}_0\geq 0 \\
        & && \sum_{x=0}^{\tilde{M}-1} \tilde{U}_{\frac{2\pi x}{\tilde{M}}}^\dagger \mathcal{K}_0 \tilde{U}_{\frac{2\pi x}{\tilde{M}}}\in \Comb{\mathcal{H}_0,\ldots,\mathcal{H}_{2N} \to \mathcal{H}_1,\ldots,\mathcal{H}_{2N-1}}\label{eq:combCov} \\
        & && \Delta_{0,2,\ldots, 2j} \mathcal{K}_0= \Delta_{0,2,\ldots, 2j} \Delta_{1,3,\ldots, 2j+1}\mathcal{K}_0 \quad  \forall  j: 0\leq j \leq N-1. \label{eq:DephasingConstraint}
\end{alignat}
\end{subequations}
Lastly, we must show that we can always restrict the number of measurement outcomes to $M=(d-1)N+1$. Let $\tilde{M}\geq M$ be any fixed number of measurement outcomes and let $\tilde{\mathcal{K}}_0$ be such that it is feasible in the above optimization problem in Eq.~\eqref{eq:SDPUnboundedM} (for the fixed $\tilde{M}$). Recall that for any $0\leq a \leq d^N-1$, we have $0\leq H_d(a)\leq (d-1)N$. Expanding $\tilde{\mathcal{K}}_0$ in the incoherent basis, i.e.,
\begin{align}
    \tilde{\mathcal{K}}_0= \sum_{n,m=0}^{d_0-1} \sum_{a,b,c,d=0}^{d^N-1}  \left(\tilde{K}_0\right)_{nmabcd}  \ketbra{n}{m}_0\otimes\ketbra{a}{b}_\text{odd} \otimes \ketbra{c}{d}_\text{even},
\end{align}
where $d_0=\text{dim}(\rho)$. The constraint in Eq.~\eqref{eq:combCov} reads
\begin{align}\label{eq:subCombNormalizationIndex} 
   &\sum_{x=0}^{\tilde{M}-1}\sum_{n,m=0}^{d_0-1} \sum_{a,b,c,d=0}^{d^N-1}  \left(\tilde{K}_0\right)_{nmabcd} e^{-i\tfrac{2\pi x}{\tilde{M}} [H_d(a)-H_d(b)]}  \ketbra{n}{m}_0\otimes\ketbra{a}{b}_\text{odd} \otimes \ketbra{c}{d}_\text{even} \nonumber \\
    &= \tilde{M} \sum_{n,m=0}^{d_0-1}\sum_{\substack{a,b,c,d=0 \\ H_d(a)=H_d(b)}}^{d^N-1} \left(\tilde{K}_0\right)_{nmabcd}  \ketbra{n}{m}_0\otimes\ketbra{a}{b}_\text{odd} \otimes \ketbra{c}{d}_\text{even}  \nonumber \\
    & \in  \Comb{\mathcal{H}_0,\ldots, \mathcal{H}_{2N} \to \mathcal{H}_1,\ldots, \mathcal{H}_{2N-1}}.
\end{align} 
Defining 
\begin{align}
    \mathcal{K}_0:= \frac{\tilde{M}}{M}& \tilde{\mathcal{K}}_0, 
\end{align}
a comparison with Eq.~\eqref{eq:subCombNormalizationIndex} reveals that the pair $M,\mathcal{K}_0$ satisfies the comb constraint too. Obviously, it satisfies the dephasing constraints in Eq.~\eqref{eq:DephasingConstraint} and
\begin{align}
    \tilde{M}\Tr{\left(\rho^T \otimes X^{(d,N)}\right)\tilde{\mathcal{K}}_0}=M\Tr{\left(\rho^T \otimes X^{(d,N)}\right)\mathcal{K}_0}.
\end{align}
In summary, we have shown that
\begin{subequations}
\begin{alignat}{2}
    \tilde{C}_{\min}^{(d,N)}= &\min   \quad &&  M \Tr{\left( \rho^T \otimes X^{(d,N)}\right)\mathcal{K}_0} \\
        & \suchthat && \mathcal{K}_0\geq 0 \\
        & && \sum_{x=0}^{M-1} \tilde{U}_{\frac{2\pi x}{M}}^\dagger \mathcal{K}_0 \tilde{U}_{\frac{2\pi x}{M}}\in \Comb{\mathcal{H}_0,\ldots,\mathcal{H}_{2N} \to \mathcal{H}_1,\ldots,\mathcal{H}_{2N-1}} \\
        & && \Delta_{0,2,\ldots, 2j} \mathcal{K}_0= \Delta_{0,2,\ldots, 2j} \Delta_{1,3,\ldots, 2j+1}\mathcal{K}_0 \quad  \forall  j: 0\leq j \leq N-1,
\end{alignat}
\end{subequations}
which finishes the proof.
\end{proof}

Note that the above proof depends crucially on the assumption of a uniform prior distribution of the phase since otherwise the constructed collection of subcombs in Eqs.~\eqref{eq:SubcombM} and~\eqref{eq:subcombK} will not leave the average cost invariant. We also used the assumption that the cost function only depends on the difference between the true phase and the estimate. Although we showed in the previous Lemma that the number of measurement outcomes can be limited to $M= (d-1)N+1$, we emphasize that considering a larger number of measurement outcomes than $M$ leads to the same minimal average cost. This means that we can use a measurement acting on a system composed of multiple qubits, which can ease the practical implementation of an optimal protocol using standard implementations of the quantum Fourier transform, see Ref.~\cite{Nielsen2010}. 

We are now ready to prove Theorem~\ref{thm:OptAvgCostSM} for the special case of $N=1$, i.e., the case that we have only access to a single copy of the unitary $V_\phi^{(d)}$. This will then allow us to reduce the multiple-copy case to the single-copy case.

\begin{figure}[ht]
    \centering
    \scalebox{0.5}{\includegraphics[width=1\linewidth]{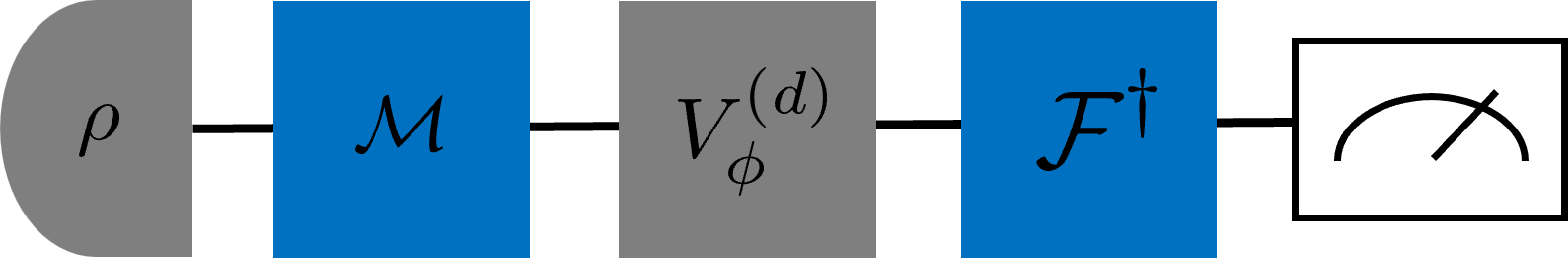}}
    \caption{Optimal phase estimation protocol for a single copy of the unitary $V_\phi^{(d)}$ of arbitrary dimension $d$. Here $\mathcal{F}$ denotes the quantum Fourier transform as defined in Eq.~\eqref{eq:QFT}, the measurement is in the incoherent basis, i.e., $M_j=\ketbra{j}{j}$, the associated phase estimates are $\hat{\phi}_j=\tfrac{2\pi j}{d}$, and $\mathcal{M} \in \MIO$ is an optimal pre-processing channel as detailed in Lemma~\ref{lem:ConstrainedOneUnitarySM}. }
    \label{fig:OneUnitaryOptimal}
\end{figure}

\begin{lemma}\label{lem:ConstrainedOneUnitarySM}
Let $Y^{(d)}$ be as in Eq.~\eqref{eq:CostMatrixDefinitionSM}, i.e.,
   \begin{align}
       Y^{(d)}= \sum_{n,m=0}^{d-1}\int \dPhi C(\phi) e^{i\phi(n-m)}\ketbra{n}{m}.
   \end{align}
Using a single copy of a unitary $V_\phi^{(d)}$, the optimal achievable average cost is given by
    \begin{align}\label{eq:OneUnitaryPrimmalSDP}
         C_{\min}^{(d,1)}(\rho)= \min_{\mathcal{M} \in \MIO} \Tr{Y^{(d)} \mathcal{M}(\rho)}.
    \end{align}
If $\mathcal{M}\in\MIO$ is an optimizer on the right-hand side, than $C_{\min}^{(d,1)}(\rho)$ is obtained by the protocol in Fig.~\ref{fig:OneUnitaryOptimal}. The measurement and phase estimates are universal for every cost function and there is no need for a memory channel.

\end{lemma}
\begin{proof}
We start from the SDP for the relaxed average cost in Eq.~\eqref{eq:MinimalCostCovariant}. For a single copy of a unitary ($N=1$ ) we can make use of the identity 
\begin{align}
    \sum_{x=0}^{d-1} \tilde{U}_{\frac{2\pi x}{d}} Z \tilde{U}_{\frac{2\pi x}{d}}^\dagger&= d\,\Delta_1 Z,
\end{align}
which follows from a straightforward calculation from Eq.~\eqref{eq:Utilde}. Therefore, the optimization problem for the relaxed minimal average cost is given by
\begin{subequations}\label{eq:SubCombSDPOneUnitary}
\begin{alignat}{2}
    \tilde{C}_{\min}^{(d,1)}(\rho)= &\min   \quad && d  \Tr{\left( \rho^T \otimes X^{(d,1)}\right) \mathcal{K}} \\
        & \suchthat && \mathcal{K}\geq 0 \\
        & && d\,\Delta_1 \mathcal{K}\in \Comb{\mathcal{H}_0,\mathcal{H}_2 \to \mathcal{H}_1}\\
        & && \Delta_{0} \mathcal{K}= \Delta_{0,1} \mathcal{K}.
\end{alignat}
\end{subequations}
Comparing this with Eq.~\eqref{eq:MinimalCostCovariant}, we remark that we denote the subcomb $\mathcal{K}_0$ as $\mathcal{K}$ here to avoid confusion with the system enumeration. We now relax the problem in Eq.~\eqref{eq:SubCombSDPOneUnitary} even further. For each feasible $\mathcal{K}$ define 
\begin{align}
    &T:=\sum_{i=0}^{d-1} \ket{ii}_{1,2} \bra{i}_{1}, \\
    &\mathcal{M}:=d (\id^0\otimes T^\dagger) \mathcal{K}(\id^0\otimes T).
\end{align}
This allows us to rewrite the objective function as
\begin{align}
    d  \Tr{\left( \rho^T \otimes X^{(d,1)}\right) \mathcal{K}} = \Tr{\left( \rho^T \otimes Y^{(d)}\right) \mathcal{M}}.
\end{align}
For any feasible $\mathcal{K}$, i.e., any $\mathcal{K}\geq 0$ with $\Delta_0 \mathcal{K}=\Delta_{0,1} \mathcal{K}$ and
$d \Delta_1\mathcal{K}=Z_{01}\otimes \id_2$ for some $Z\geq 0$ with $\partTr{1}{Z_{01}}=\id_0$, we obviously have $\mathcal{M}\geq 0$. Moreover, 
\begin{align}
    \partTr{1}{\mathcal{M}}&= d \sum_{i} (\id_0 \otimes \bra{ii}_{1,2}) \mathcal{K}(\id_0 \otimes \ket{ii}_{1,2}) \nonumber \\
    &= d \sum_{i} (\id_0 \otimes \bra{ii}_{1,2}) \Delta_{1}(\mathcal{K})(\id_0 \otimes \ket{ii}_{1,2}) \nonumber \\
    &=\partTr{1}{Z} =\id_0,
\end{align}
and 
\begin{align}
    \Delta_0 \mathcal{M} &= d \sum_{i,j} \ketbra{i}{j}_{1} \otimes (\id_0 \otimes \bra{ii}_{1,2}) \Delta_0(\mathcal{K})(\id_0 \otimes \ket{jj}_{1,2}) \nonumber \\
    &=d \sum_{i,j} \ketbra{i}{j}_{1} \otimes(\id_0 \otimes \bra{ii}_{1,2}) \Delta_{01}(\mathcal{K})(\id_0 \otimes \ket{jj}_{1,2}) \nonumber \\
    &=d \sum_{i} \ketbra{i}{i}_{1} \otimes(\id_0 \otimes \bra{ii}_{1,2}) \Delta_{01}(\mathcal{K})(\id_0 \otimes \ket{ii}_{1,2}) \nonumber \\
    &=  \Delta_{0,1} \mathcal{M}.
\end{align}
Thus, $\mathcal{M}$ is the Choi state of a MIO channel. Recall that the optimization problem in Eq.~\eqref{eq:MinimalCostCovariant} is itself a lower bound on the minimal average cost, and thus
\begin{align}\label{eq:MIOChannel}
     C_{\min}^{(d,1)}(\rho)\geq \tilde{C}_{\min}^{(d,1)}(\rho) \geq \min_{\mathcal{M} \in \MIO} \Tr{Y^{(d)} \mathcal{M}(\rho)}.
\end{align}
Next, we will show that this bound is attainable. Let $\mathcal{M}^\star$ be an optimal point of
the optimization problem on the right-hand side of Eq.~\eqref{eq:MIOChannel}. Using the protocol depicted in Fig.~\ref{fig:OneUnitaryOptimal} results in an outcome distribution
\begin{align}
    p^{(d,1)}(x|\phi,\rho)=\bra{x} \, \mathcal{F}^\dagger V_\phi^{(d)}\mathcal{M}^\star(\rho) {V_\phi^{(d)}}^\dagger\mathcal{F}\ket{x}, 
\end{align}
where $\mathcal{F}$ denotes the quantum Fourier transform on a system of dimension $d$, i.e.,  
\begin{align}\label{eq:QFT}
    \mathcal{F} \ket{n}= \frac{1}{\sqrt{d}} \sum_{k=0}^{d-1} e^{i 2\pi \tfrac{ kn}{d}} \ket{k}.
\end{align}
If we write $\mathcal{M}^\star(\rho)=\sum_{n,m} M_{nm} \ketbra{n}{m}$, then
\begin{align}
    p^{(d,1)}(x|\phi,\rho)&= \frac{1}{d} \sum_{n,m=0}^{d-1} M_{nm} e^{i(n-m)(\phi-\tfrac{2\pi x}{d})},
\end{align}
and
\begin{align}\label{eq:explitProtocol}
    C_{\min}^{(d,1)}(\rho) &\leq \sum_{x=0}^{d-1} \int \dPhi C(\phi-\tfrac{2\pi x}{d}) p^{(d,1)}(x|\phi,\rho) \nonumber \\
    &=d \int \dPhi' C(\phi^\prime) \frac{1}{d}\sum_{n,m} M_{nm} e^{i\phi (n-m) \phi^\prime} \nonumber \\
    &= \sum_{n,m} M_{nm} \bra{m} Y^{(d)}\ket{n} \nonumber \\
    &= \Tr{Y^{(d)} \mathcal{M}^\star(\rho)}.
\end{align}
Together with Eq.~\eqref{eq:MIOChannel}, this shows that $ C_{\min}^{(d,1)}(\rho) = \min_{\mathcal{M} \in \MIO} \Tr{Y^{(d)} \mathcal{M}(\rho)}$.
\end{proof}

Next, we use the previous Lemma to obtain a lower bound on the minimal achievable average cost with arbitrary $N$.

\begin{thm}\label{thm:OptimalCostN-LowerBound}
For any pair $(d,N)$, let $M=(d-1)N+1$ and let $Y^{(M)}$ be defined as in Eq.~\eqref{eq:CostMatrixDefinitionSM}. 
Then, the average cost is bounded by
    \begin{align}
        C_{\min}^{(d,N)}(\rho)\geq C_{\min}^{(M,1)}(\rho)= \min_{\mathcal{M} \in \MIO} \Tr{Y^{(M)} \mathcal{M}(\rho)}.
    \end{align}
\end{thm}
\begin{proof}
For all $j\in\{0,\ldots, N-1\}$, let us define the maps 
\begin{align}\label{eq:DMaps}
    \mathcal{D}_j(\cdot)=\left(\Delta_{0,\ldots, 2j}-\Delta_{0,\ldots, 2j}\Delta_{1,\ldots, 2j+1}\right)(\cdot)
\end{align} 
The dual of the semidefinite program describing the relaxed average cost in Eq.~\eqref{eq:MinimalCostCovariant} is given by
\begin{subequations}\label{eq:DualSDPCost}
\begin{alignat}{2}
    \tilde{C}_{\min}^{(d,N)}(\rho)= &\max   \quad &&   \Tr{B_1} \\
        & \suchthat && B_i=B_i^\dagger \quad \forall i \\
        & && A_i=A_i^\dagger \quad \forall i \\
        & && B_j\otimes \id^{(2j-1)}-\partTr{2j}{B_{j+1}}=0\ \forall j=1,\ldots, N \label{eq:BConstraints}\\
        & && M\, \rho^T\otimes X^{(d,N)}-\sum_{x=0}^{M-1} \tilde{U}_{\frac{2\pi x}{M}} B_{N+1} \tilde{U}_{\frac{2\pi x}{M}}^\dagger-\sum_{j=0}^{N-1} \mathcal{D}_j(A_j)\geq 0.\label{eq:InqConstraint}
\end{alignat}
\end{subequations}
We start from Lemma~\ref{lem:ConstrainedOneUnitarySM} with a single unitary $V_\phi^{(M)}$ (note the dimension $M$ here). In the following, we will use the Roman numbers $\I,\II$ to denote the input and output systems of $V_\phi^{(M)}$ and numbers $1,2,\ldots, 2N$ to denote input and output systems of $V_\phi^{(d)}$ to avoid ambiguities. According to Lemma~\ref{lem:ConstrainedOneUnitarySM}, the dual of the (relaxed) average cost for a single unitary $V_\phi^{(M)}$ is given by
\begin{subequations}\label{eq:BigUnitaryDualY}
 \begin{alignat}{2}
    C_{\min}^{(M,1)}(\rho)= \tilde{C}_{\min}^{(M,1)}(\rho)=\min_{\mathcal{M} \in \MIO} \Tr{Y^{(M)}\mathcal{M}(\rho)} =& \max_{A,B}\quad  && \Tr{B} \\
        & \suchthat && A^\dagger=A\\
        & && B^\dagger=B \\
        & && \rho^T\otimes  Y^{(M)} - B\otimes \id - \mathcal{D}_{0}(A)\ge0,
\end{alignat}    
\end{subequations}
where the dual variables  $A \in \mathcal{L}(\mathcal{H}_0 \otimes \mathcal{H}_I)$, $B \in \mathcal{L}(\mathcal{H}_0)$, and $\tilde{\mathcal{D}}_{0}(A)=(\Delta_0-\Delta_{0,{\I}}) A$. Let $(A^\star, B^\star)$ denote an optimal feasible point of this problem. In the following, we will use this optimal feasible point to construct a feasible point of the optimization problem in Eq.~\eqref{eq:DualSDPCost}, and thus a lower bound on $C_{\min}^{(d,N)}(\rho)$

Recall that for $N$ qudit systems with $\ket{n}=\ket{n_1,\ldots, n_N}$ the digit sum is $H_d(n)=n_1+\ldots +n_N$. Each $n_k$ can take values $0\leq n_k \leq d-1$ and thus $0 \leq H_d(n)\leq (d-1)N$. This implies $H_d(n)$ takes at most $M=N(d-1)+1$ different values. Let $s_N(i)$ denote the set of all such digit-strings of length $N$ that have a digit sum (as defined around Eq.~\eqref{eq:UnitaryExpandedForm}) of $i$, i.e., for $\Vec{a}\in s_N(i)$ we have $H_d(\Vec{a})=i$. Define
\begin{align}
    &\Pi:=\sum_i \ketbra{ii}{ii}_{\I,\II},\label{eq:DefinitionPI} \\
    & W:=\sum_{i=0}^{M-1} \sum_{\Vec{j}\in s_N(i)}\ket{\Vec{j}}_{\odd} \ket{\Vec{j}}_{\even}\bra{ii}_{\I,\II} . \label{eq:DefinitionW}
\end{align}
To explicitly construct a feasible solution to Eq.~\eqref{eq:DualSDPCost} we define the following matrices. Let
\begin{align}\label{eq:BN+1}
    B_{N+1}=B^\star \otimes W\Pi W^\dagger=B^\star \otimes \sum_{i=0}^{M-1} \sum_{\Vec{j},\Vec{k}\in s_{N}(i)}\ket{\Vec{j}}_{\odd} \ket{\Vec{j}}_{\even} \bra{\Vec{k}}_{\odd} \bra{\Vec{k}}_{\even},
\end{align}
and for $2\leq l \leq N$ let
\begin{align}\label{eq:Bl}
    B_{l}= B^\star \otimes \sum_{i=0}^{(d-1)(l-1)} \sum_{\Vec{j},\Vec{k}\in s_{l-1}(i)}\ket{\Vec{j}}_{1,...,2l-3} \ket{\Vec{j}}_{2,...,2l-2} \bra{\Vec{k}}_{1,...,2l-3} \bra{\Vec{k}}_{2,...,2l-2}.
\end{align}
Moreover, for $T=\sum_{i=0}^{M-1}\ket{ii}_{\I,\II}\bra{i}_{\I}$, let 
\begin{align}\label{eq:DefTildeA}
    \tilde{A}:=M (\id^{(0)}\otimes WT) A^\star (\id^{(0)}\otimes WT).
\end{align}
We now proceed to show that the tuples $(B_1=B^\star,B_2, \ldots, B_{N+1})$ and $(\tilde{A},\ldots, \tilde{A})$, i.e., all $A_j=\tilde{A}$, are a feasible point for the optimization problem in Eq.~\eqref{eq:DualSDPCost}. We start by showing that $B_j\otimes \id^{(2j-1)}-\partTr{2j}{B_{j+1}}=0\ \forall j=1,\ldots, N$.  Let us denote a string of $l$ digits as $\vec{j}=j_1,\ldots,j_{l}$, then recall that $\ket{\Vec{j}}_{1,...,2l-1} \ket{\Vec{j}}_{2,...,2l}=\ket{j_1,\ldots,j_{l}}_{1,...,2l-1} \ket{j_1,\ldots,j_{l}}_{2,...,2l}$. Let $\vec{j}^\prime=j_1,\ldots,j_{l-1}$ denote a string of length $l-1$ that shares all its digits with the first $l-1$ digits of $\vec{j}$. If two strings $\vec{j},\vec{k}\in s_{l}(i)$ (of fixed length $l$) share the same digit sum, and their last digit coincides, i.e., $j_l=k_l$, then their truncated strings $\vec{j}^\prime, \vec{k}^\prime$ have the same digit sum too, and $\vec{j}^\prime, \vec{k}^\prime \in s_{l-1}(i-j_l)$  whenever $0\leq i-j_l\leq (d-1)(l-1)$. Moreover, for any $i<0$ and $i> (d-1)(l-1) $, $S_{l-1}(i)= \emptyset$. We can therefore write
\begin{align}\label{eq:BlGuess}
    \partTr{2l}{B_{l+1}}&= B^\star \otimes \sum_{i=0}^{(d-1)l} \sum_{\Vec{j},\Vec{k}\in s_{l}(i)}\ket{\Vec{j}^\prime}_{1,...,2l-3} \ket{\Vec{j}^\prime}_{2,...,2l-2} \bra{\Vec{k}^\prime }_{1,...,2l-3}  \bra{\Vec{k}^\prime}_{2,...,2l-2} \otimes \ketbra{j_{l}}{j_l}_{2l-1} \,\delta_{j_l,k_l} \nonumber \\
    &= B^\star \otimes \sum_{i=0}^{(d-1)l} \sum_{j_l=0}^{d-1} \sum_{\vec{j}^\prime,\vec{k}^\prime \in s_{l-1}(i-j_l)} \ket{\Vec{j}^\prime}_{1,...,2l-3} \ket{\Vec{j}^\prime}_{2,...,2l-2} \bra{\Vec{k}^\prime }_{1,...,2l-3}  \bra{\Vec{k}^\prime}_{2,...,2l-2} \otimes \ketbra{j_{l}}{j_l}_{2l-1} \nonumber \\
     &= B^\star \otimes\sum_{j_l=0}^{d-1} \ketbra{j_{l}}{j_l}_{2l-1} \otimes \sum_{i=-j_l}^{(d-1)l-j_l} \sum_{\vec{j}^\prime,\vec{k}^\prime \in s_{l-1}(i)} \ket{\Vec{j}^\prime}_{1,...,2l-3} \ket{\Vec{j}^\prime}_{2,...,2l-2} \bra{\Vec{k}^\prime }_{1,...,2l-3}  \bra{\Vec{k}^\prime}_{2,...,2l-2}  \nonumber \\
     &= B^\star \otimes\sum_{j_l=0}^{d-1}  \ketbra{j_{l}}{j_l}_{2l-1} \otimes \sum_{i=0}^{(d-1)(l-1)} \sum_{\vec{j}^\prime,\vec{k}^\prime \in s_{l-1}(i)} \ket{\Vec{j}^\prime}_{1,...,2l-3} \ket{\Vec{j}^\prime}_{2,...,2l-2} \bra{\Vec{k}^\prime }_{1,...,2l-3}  \bra{\Vec{k}^\prime}_{2,...,2l-2} \nonumber \\
     &= \id_{2l-1} \otimes B_l,
\end{align}
where in the second-to-last line we used that $s_{l-1}(i)$ is an empty set for some of the summands. In addition,
\begin{align}
    \partTr{2}{B_2}&= \partTr{2}{B^\star \otimes \sum_{j,k=0}^{d-1} \ketbra{jj}{kk}_{1,2}} =B^\star \otimes \id_1.
\end{align}
Therefore, the tuple $(B_1=B^\star,B_2, \ldots, B_{N+1})$ satisfies the constraint $B_j\otimes \id_{2j-1}-\partTr{2j}{B_{j+1}}=0\ \forall j=1,\ldots, N$. Next, we show that the construction satisfies the constraint 
\begin{align}
    M\, \rho^T\otimes X^{(d,N)}-\sum_{x=0}^{M-1} \tilde{U}_{\frac{2\pi x}{M}} B_{N+1} \tilde{U}_{\frac{2\pi x}{M}}^\dagger-\sum_{j=0}^{N-1} \mathcal{D}_j(A_j)\geq 0.
\end{align}
To this end, first notice that in Eq.~\eqref{eq:BigUnitaryDualY}, $A$ only appears in the term $\tilde{\mathcal{D}}_0(A) $. Without loss of generality, we can thus assume that $A^\star= \tfrac{1}{M} \sum_i \ketbra{i}{i}_0 \otimes A_{\I}^{(i)}$, where the $A_{\I}^{(i)}$ are off-diagonal matrices on system $\I$, such that $\tilde{\mathcal{D}}_0(A^\star)=(\Delta_0-\Delta_{0,\I})(A^\star)=A^\star$ (if $A^\star$ had any other non-vanishing elements, they would be irrelevant for the optimization problem). 
Denoting $\ket{\vec{k}}=\ket{k_1,\ldots,k_N}$, we can rewrite Eq.~\eqref{eq:DefTildeA} as
\begin{align}\label{eq:ATilde}
    \tilde{A}&= \sum_i \ketbra{i}{i}_0 \otimes \sum_{n \neq m} (A_i)_{nm} \sum_{\vec{k} \in s_N(n)} \sum_{\vec{l} \in s_N(m)} \ketbra{\vec{k},\vec{k}}{\vec{l},\vec{l}}_{\text{odd,even}} \nonumber \\
    &= \sum_i \ketbra{i}{i}_0 \otimes \sum_{n \neq m} (A_i)_{nm} \sum_{\vec{k} \in s_N(n)} \sum_{\vec{l} \in s_N(m)} \ketbra{k_1,\ldots, k_N}{l_1,\ldots, l_N}_{2,\ldots,2N} \otimes \ketbra{k_1,\ldots, k_N}{l_1,\ldots, l_N}_{1,\ldots,2N-1}.
\end{align}
It follows directly that
\begin{subequations}
    \begin{align}
  & \Delta_0 \tilde{A}=\tilde{A},\label{eq:DeltaRelationZero}\\
  &\Delta_{2j}\tilde{A}=\Delta_{2j-1}\tilde{A}  \qquad \forall \, j=1,\ldots, N, \label{eq:DeltaRelation} \\
  & \Delta_{0,1,2,3,\cdots,2N}\tilde{A}=0. \label{eq:FullyDephasedATIlde}
\end{align}
\end{subequations}
Remember that $A_0=A_1=\ldots=A_{N-1}=\tilde{A}$. Using the structure of $\tilde{A}$, and, in particular Eq.~\eqref{eq:DeltaRelation}, notice that for any $j=0,\ldots, N-1$ we have
\begin{align}\label{eq:DjDeltaIdentity}
    \mathcal{D}_j(A_j)&=\left(\Delta_{0,\ldots, 2j}-\Delta_{0,\ldots, 2j}\Delta_{1,\ldots, 2j+1}\right)\tilde{A} \overset{\eqref{eq:DeltaRelation}}{=} \left(\Delta_{0,\ldots, 2j} \Delta_{1,\ldots, 2j-1}-\Delta_{0,\ldots, 2j}\Delta_{1,\ldots, 2j+1}\right)\tilde{A} \nonumber\\
    &= \Delta_0 \Delta_{12} \ldots \Delta_{2j-1,2j} \left(\idChannel-\Delta_{2j+1}\right) \tilde{A}.
\end{align}
Using the telescoping sum,
\begin{align}\label{eq:IdentitySumA}
    \sum_{j=0}^{N-1}\mathcal{D}_j(\tilde{A})&= \mathcal{D}_0(\tilde{A}) + \sum_{j=1}^{N-1} \Delta_0 \Delta_{12} \ldots \Delta_{2j-1,2j} \tilde{A} - \sum_{j=1}^{N-1}\Delta_0 \Delta_{12} \ldots \Delta_{2j-1,2j} \Delta_{2j+1,2j+2} \tilde{A} \nonumber \\
    &= \Delta_0 (\idChannel-\Delta_{12}) \tilde{A} + \Delta_0 \Delta_{12} \tilde{A} -  \Delta_{0}\Delta_{1,2}\ldots \Delta_{2N-1,2N} \tilde{A} \nonumber \\
    &\overset{\eqref{eq:FullyDephasedATIlde}}{=}\Delta_0 \tilde{A} \nonumber\\
    &\overset{\eqref{eq:DeltaRelationZero}}{=}\tilde{A}.
\end{align}
Moreover, we have the following identity for $W$ (as defined in Eq.~\eqref{eq:DefinitionW}) and $U_{\phi}^{(d,N)}$ (as defined in Eq.~\eqref{eq:UnitaryExpandedForm})
\begin{align}\label{eq:CommutingWandU}
    U_{\phi}^{(d,N)} W =\sum_{k=0}^{M-1} e^{i\phi k} \sum_{\Vec{j}\in s_N(k)}\ket{\Vec{j}}_{\odd} \ket{\Vec{j}}_{\even}\bra{kk}_{\I,\II}= W U_{\phi}^{(M,1)},
\end{align}
where $U_{\phi}^{(M,1)}$ is analogously defined to Eq.~\eqref{eq:UnitaryExpandedForm}, but acting on systems $\I,\II$ (of dimension $M$ each) only, i.e.,
\begin{align}
    U_{\phi}^{(M,1)}= \sum_{k=0}^{M-1} e^{i\phi k} \ketbra{n}{n}_{\I} \otimes \mathbb{1}_{\II}.
\end{align}
Thus,
\begin{align} \label{eq:SumUnitaryB}
    \sum_{x=0}^{M-1} \tilde{U}_{\frac{2\pi x}{M}}^{(d,N)} B_{N+1} \left( \tilde{U}_{\frac{2\pi x}{M}}^{(d,N)} \right)^{\dagger} &= \sum_{x=0}^{M-1} \left(\id \otimes U_{\frac{2\pi x}{M}}^{(d,N)}\right) (B^\star \otimes W\Pi W^\dagger)\left(\id \otimes U_{\frac{2\pi x}{M}}^{(d,N)}\right)^\dagger \nonumber \\
    &=B^\star  \otimes \sum_{x=0}^{M-1}  U_{\frac{2\pi x}{M}}^{(d,N)}W\Pi W^\dagger \left(U_{\frac{2\pi x}{M}}^{(d,N)}\right)^{\dagger} \nonumber\\
    &\overset{\eqref{eq:CommutingWandU}}{=} B^\star \otimes W\sum_{x=0}^{M-1} U_{\frac{2\pi x}{M}}^{(M,1)} \Pi \left(U_{\frac{2\pi x}{M}}^{(M,1)}\right)^{\dagger} W^\dagger \nonumber\\
    &=B^\star \otimes M W \Pi W^\dagger,
\end{align}
where in the last line, we used that $\Pi=\sum_{i} \ketbra{ii}{ii}_{\I,\II}$, and as such all $U_{\frac{2\pi x}{M}}^{(M,1)}$ leave it invariant. Recall that for $T=\sum_{i=0}^{M-1} \ket{ii}_{\I,\II}\bra{i}_{\I}$, we have that $\Pi=TT^\dagger$ and $X^{(M,1)}=TY^{(M)}T^\dagger$, and thus, $X^{(d,N)}=WTY^{(M)}T^\dagger W^\dagger$.
Therefore, the inequality constraint in Eq.~\eqref{eq:InqConstraint} can be rewritten as
\begin{align}
    &\rho^T\otimes MX^{(d,N)}-\sum_{x=0}^{M-1} \tilde{U}_{\frac{2\pi x}{M}} B_{N+1} \tilde{U}_{\frac{2\pi x}{M}}^\dagger -\sum_{j=0}^{N-1} \mathcal{D}_j(A_j) \overset{\eqref{eq:IdentitySumA}}{=}\rho^T\otimes MX^{(d,N)}-\sum_{x=0}^{M-1} \tilde{U}_{\frac{2\pi x}{M}} B_{N+1} \tilde{U}_{\frac{2\pi x}{M}}^\dagger -\tilde{A} \nonumber \\
    &\overset{\eqref{eq:SumUnitaryB}}{=} M \, (\id\otimes WT)\left( \rho^T\otimes Y^{(M)}-B^\star \otimes \id \right) (\id\otimes WT)^\dagger -\tilde{A} \nonumber \\
    &\geq M \, (\id\otimes WT) \tilde{\mathcal{D}}_{0}(A^\star) (\id \otimes WT)^\dagger -\tilde{A} \nonumber \\
    &= M \, (\id\otimes WT)A^\star (\id \otimes WT)^\dagger -\tilde{A} \nonumber\\
    &\overset{\eqref{eq:DefTildeA}}{=}0,
\end{align}
where the inequality follows since $(A^\star, B^\star)$ is a feasible solution to Eq.~\eqref{eq:BigUnitaryDualY}. As such, all the constraints in Eq.~\eqref{eq:DualSDPCost} are satisfied by the variables we constructed. Therefore, the tuples $(B^\star,B_2, \ldots, B_{N+1})$ and $(\tilde{A},\ldots, \tilde{A})$ provide a feasible solution to Eq.~\eqref{eq:DualSDPCost}. Since the objective function of the dual problem in Eq.~\eqref{eq:DualSDPCost} can be computed from $B_1=B^\star$ only, we find that

\begin{align}
    \tilde{C}_{\min}^{(d,N)}(\rho)\geq \Tr{B^\star} \overset{\eqref{eq:BigUnitaryDualY}}{=} \min_{\mathcal{M} \in \MIO} \Tr{Y^{(M)} \mathcal{M}(\rho)}  = C_{\min}^{(M,1)}(\rho),
\end{align}
where $M=(d-1)N+1$.
\end{proof}

Lastly, we will have to show that the bound in Theorem~\ref{thm:OptimalCostN-LowerBound} is achievable. We do this by explicit construction. The lower bound on the average cost in Theorem~\ref{thm:OptimalCostN-LowerBound} corresponds to the average cost of a phase estimation protocol that probes a higher-dimensional unitary $V_\phi^{(M)}$. Our next step is thus to construct a quantum network that can implement such a unitary from the $N$ copies of $V_\phi^{(d)}$. We do this in the following Lemma, which is a generalization of the case with $d=2$ presented in Ref.~\cite{vanDam2007b}.

\begin{lemma}\label{lem:PhaseShiftCurcuit}
    There exists a network of MIO channels that converts $N$ copies of $V_\phi^{(d)}$ into a single copy of $V_\phi^{(M)}$, where $M=(d-1)N+1$. The network is shown in Fig.~\ref{fig:PhaseShiftCurcuitD}. 
\end{lemma}

\begin{proof}
The idea of the proof is to show that the circuit in Fig.~\ref{fig:PhaseShiftCurcuitD} satisfies the claims of the Lemma. To this end, we begin by describing its components.  First, define the integer
    \begin{align}
    &l:= \max \left\{ k\geq 0: \sum_{i=0}^{k-1} d^i= \frac{d^{k}-1}{d-1} \leq N  \right\} \label{eq:DefinitionLQudits}.
\end{align}
As such we have $\tfrac{d^l-1}{d-1} \leq N<  \tfrac{d^{l+1}-1}{d-1}$. This implies that 
the system on which $V_\phi^{(M)}$ acts, i.e., $\text{span}\{\ket{k}: 0\leq k\leq  (d-1)N \}$, has a lower dimension than the space consisting of $l+1$ qudits (of dimension $d$). Expressing any $0\leq k\leq  (d-1)N$ in base $d$ as $k=k_0+k_1d^1+\ldots+ k_l d^l$, we define an embedding isometry, which is the first step of the network in Fig.~\ref{fig:PhaseShiftCurcuitD}, as
\begin{align}
    W_e=\sum_{k=0}^{(d-1)N}\ket{k_0,k_1,\cdots,k_l}_{0,\ldots, l}\bra{k}_{\I}.
\end{align}

Next, we apply a permutation on the computational basis (not on the individual digits $k_i$ of $k$), 
\begin{align}
    U_\pi=\sum_{k=0}^{d^{l+1}-1} \ketbra{\pi(k)}{k}_{0,\ldots, l},
\end{align}
where $\pi$ denotes a permutation that we will fix later.
In the next step, for $i\in\{0,\cdots,l-1\}$, we apply $V^{(d)}_\phi$ $d^i$-times to the $i$-th qudit, see Fig.~\ref{fig:PhaseShiftCurcuitD}. By our choice of $l$, this uses less than $N$ copies of $V^{(d)}_\phi$. The remaining $M^\prime:=N-\tfrac{d^l-1}{d-1}$ copies are applied to qudit $l$. Note that this can always be done using a network consisting of SWAP gates and identity channels, which are in MIO. After an application of $U_{\pi}^\dagger$, we go back to the system on which $V_\phi^{(M)}$ acts, via the measurement implementing the inverse of $W_e$,
\begin{align}
    \mathcal{M}(\rho)=\Pi_1\rho\Pi_1 + \Pi_2\rho \Pi_2,
\end{align}
where $\Pi_1=\sum_{k=0}^{(d-1)N} \ketbra{k}{k_0,k_1,\cdots,k_l}$ and $\Pi_2=\sum_{k=(d-1)N+1}^{d^{l+1}-1} \ketbra{0}{k_0,k_1,\cdots,k_l}$. Verifying that this measurement is a MIO channel is straightforward.

\begin{figure}[ht]
    \centering
    \scalebox{0.6}{\includegraphics[width=1\linewidth]{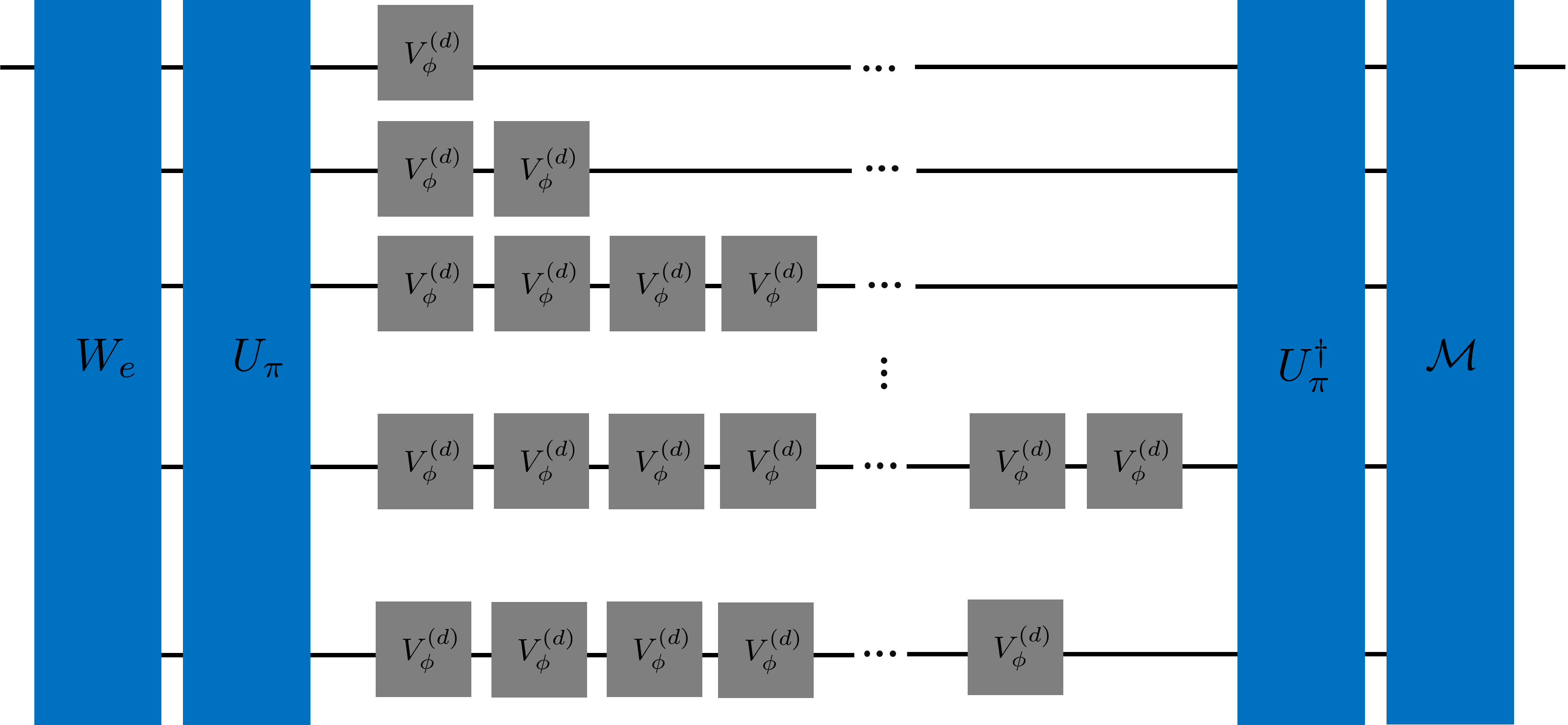}}
    \caption{Implementation of the unitary $V_\phi^{(M)}$ via a MIO network using  $N$ copies of $V_\phi^{(d)}$ in total. For more details, see the proof of Lemma~\ref{lem:PhaseShiftCurcuit}. }
    \label{fig:PhaseShiftCurcuitD}
\end{figure}

We will now show that there exists indeed a permutation $\pi$ such that the circuit we described above implements $V_\phi^{(M)}$. Let $m(k):=\pi(k)$ be expressed in base $d $ as $m(k):=m_0(k)+m_1(k) d^1+\ldots +m_{l-1}(k)d^{l-1}+m_l(k) d^l$. Any basis state $\ket{k}$ with $0\le k\le (d-1) N$ undergoes the following evolution
\begin{align}\label{eq:PhaseShiftCurcuitPhase}
    \ket{k}& \overset{W_e}{\mapsto} \ket{k_0,k_1,\ldots,k_l}\nonumber \\
    & \overset{U_\pi}{\mapsto} \ket{m_0(k),\ldots, m_{l-1}(k), m_l(k)} \nonumber\\
    &\mapsto e^{i\phi [m_0(k)+m_1(k) d^1+\ldots +m_{l-1}(k)d^{l-1}+m_l(k) M^\prime]}  \ket{m_0(k),\ldots, m_{l-1}(k), m_l(k)} \nonumber \\
    &\overset{U_\pi^\dagger}{\mapsto}e^{i\phi [m_0(k)+m_1(k) d^1+\ldots +m_{l-1}(k)d^{l-1}+m_l(k) M^\prime]} \ket{k_0,k_1,\ldots,k_l}\nonumber \\
    &\overset{\mathcal{M}}{\mapsto}e^{i\phi [m_0(k)+m_1(k) d^1+\ldots +m_{l-1}(k)d^{l-1}+m_l(k) M^\prime]} \ket{k}.
\end{align}
Let $H$ denote the set of integers that the multipliers for the phase in the exponent in Eq.~\eqref{eq:PhaseShiftCurcuitPhase}  attain, i.e.,
\begin{align}
    H&= \{ n\geq 0: n=m_0(k)+m_1(k) d^1+\ldots +m_{l-1}(k)d^{l-1}+m_l(k) M^\prime, 0\leq m_i(k) \leq d-1\} \nonumber \\
    &= \{ n: 0\leq n \leq (d-1) N \},
\end{align}
where last line follows directly, since, by construction, $0\leq m_0+m_1 d^1+\ldots +m_{l-1}d^{l-1} \leq d^l-1$ and $0\leq m_l M^\prime \leq (d-1) \left(N-\tfrac{d^l-1}{d-1}  \right)$. Hence, the circuit in Fig.~\ref{fig:PhaseShiftCurcuitD} produces exactly the same phases as $V_\phi^{(M)}$. 

It remains to ensure that they are applied to the correct basis states. To this end, we choose a $\pi$ that maps any $0\le k \le (d-1)N$ to a $\pi(k)=m(k)$ such that $m_0(k)+m_1(k) d^1+\ldots +m_{l-1}(k)d^{l-1}+m_l(k) M^\prime=k$. Such a permutation always exists: We simply match the total of $M$ basis states $\ket{k}$ with the correct phase $m_0(k)+m_1(k) d^1+\ldots +m_{l-1}(k)d^{l-1}+m_l(k) M^\prime=k \in H$. For an explicit construction of the permutation in the case of $d=2$, see Ref.~\cite{vanDam2007b}.\qedhere

\end{proof}

To conclude the proof of Theorem~\ref{thm:OptAvgCostSM}, we now combine Lemma~\ref{lem:ConstrainedOneUnitarySM}, Lemma~\ref{lem:PhaseShiftCurcuit}, and Theorem~\ref{thm:OptimalCostN-LowerBound}.

\begin{lem}
    The bounds in Theorem~\ref{thm:OptimalCostN-LowerBound} are achievable, i.e., the minimal achievable average cost for any pair $(d,N)$ and input state $\rho$ is given by
    \begin{align}
        C_{\min}^{(d,N)}(\rho) = \min_{\mathcal{M}\in \MIO} \Tr{Y^{(M)} \mathcal{M}(\rho)},
    \end{align}
    where $M=(d-1)N+1$ and
    \begin{align}
        Y^{(M)}=\sum_{n,m=0}^{M-1} \ketbra{n}{m} \int \dPhi C(\phi) e^{i\phi(n-m)}.
    \end{align}
\end{lem}
\begin{proof}
    Starting from our initial generic phase estimation protocol depicted in Fig.~\ref{fig:PhaseEstimationProtocolsSM}, whose minimal average cost is given by Eq.~\eqref{eq:MinimalAverageCostRho}, we can always choose a specific protocol to obtain an upper bound. Let $\mathcal{M}^\star \in \MIO$ be such that
    \begin{align}
        \Tr{Y^{(M)} \mathcal{M}^\star(\rho)}=\min_{\mathcal{M}\in \MIO} \Tr{Y^{(M)} \mathcal{M}(\rho)},
    \end{align}
    where the dimensions of $\mathcal{M}$ are fixed by $\rho$ and $Y^{(M)}$. Choose a measurement $\{M_k\}_{0\leq k\leq (d-1)N}=\{\mathcal{F} \ketbra{k}{k} \mathcal{F}^\dagger \}_{0\leq k\leq (d-1)N}$, where $\mathcal{F}$ is the quantum Fourier transformation of dimension $M$ as defined in Eq.~\eqref{eq:QFT}. Using the network from Lemma~\ref{lem:PhaseShiftCurcuit}, allows us to implement the unitary $V_\phi^{(M)}$, which we probe with the state $\mathcal{M}^\star(\rho)$ and subsequently measure with $\{M_k\}$ to produce measurement outcomes with distribution $p^{(d,N)}(x|\phi, \rho)=\bra{x} \, \mathcal{F}^\dagger V_\phi^{(M)}\mathcal{M}^\star(\rho) {V_\phi^{(M)}}^\dagger\mathcal{F}\ket{x}$. Let the phase estimates be given by $\{ \hat{\phi}_x=\tfrac{2\pi x}{M} \}$. Then the average cost in Eq.~\eqref{eq:MinimalAverageCostRho} is upper bounded by
    \begin{align}
        C_{\min}^{(d,N)}(\rho) \leq \sum_{x=0}^{M-1} \int \dPhi C(\phi-\tfrac{2\pi x}{M}) p^{(d,N)}(x|\phi,\rho) =  \Tr{Y^{(M)} \mathcal{M}^\star(\rho)},
    \end{align}
    where the last equality has been shown in Lemma~\ref{lem:ConstrainedOneUnitarySM} in Eq.~\eqref{eq:explitProtocol} (if one sets $d=M$). Together with Theorem~\ref{thm:OptimalCostN-LowerBound}, this yields
    \begin{align}
          \min_{\mathcal{M}\in \MIO} \Tr{Y^{(M)} \mathcal{M}(\rho)} \leq \tilde{C}_{\min}^{(d,N)}(\rho)\leq C_{\min}^{(d,N)}(\rho) \leq   \Tr{Y^{(M)} \mathcal{M}^\star(\rho)}=\min_{\mathcal{M}\in \MIO} \Tr{Y^{(M)} \mathcal{M}(\rho)},
    \end{align}
    which completes the proof.
\end{proof}

\section{Optimal phase estimation without coherence constraints}\label{sec:OptimalPhaseEstimationNoConstraints}
As discussed in the main text, our results include the results of Ref.~\cite{vanDam2007} on optimal phase estimation protocols without any coherence constraints as a special case. In their scenario, they consider the optimization problem in Eq.~\eqref{eq:MinimalAverageCostRho} but drop the constraint that the network is composed of MIO channels. Equivalently, if we supply a sufficient amount of coherence (for example, many copies of maximally coherent states), each MIO channel can, together with one of the maximally coherent states, implement an arbitrary channel, and we recover the scenario of Ref.~\cite{vanDam2007}. In particular, one does not need an unbounded amount of coherence to achieve this. As a direct consequence of Theorem~\ref{thm:OptAvgCostSM}, the most amount of coherence (for any cost function) needed to achieve optimal phase estimation is given by a maximally coherent state of dimension $\ket{\Psi_M^+}=\tfrac{1}{\sqrt{M}} \sum_{i=0}^{M-1} \ket{i}$, since this allows to prepare any state, see Ref.~\cite{Baumgratz2014}, thus
\begin{align}\label{eq:unconstraint}
    C_{\min}^{(M,1)}(\Psi_M^+)= \min_{\mathcal{M} \in \MIO} \Tr{Y^{(M)} \mathcal{M}(\Psi_M^+)}= \min_{\tau\geq 0, \Tr{\tau}=1} \Tr{Y^{(M)} \tau} =\lambda_{\min}\left( Y^{(M)}\right).
\end{align}
As we demonstrate in Section~\ref{sec:HolevoCostFunction}, supplying a maximally coherent state to achieve the optimal average cost is not necessarily required.

\section{Optimal phase estimation under MIO-preserving superchannels}\label{sec:PhaseEstimationMIOpreserving}

As mentioned in Section~\ref{sec:MIOpreserving}, we want to briefly discuss the consequences of considering MIO-preserving superchannels instead of completely MIO-preserving superchannels in their application to phase estimation. In particular, consider a single-qubit unitary $V_\phi^{(2)}$ and the Holevo cost function $C(\phi)=4\sin^2(\phi/2)$ (periodized variance, see also Ref.~\cite{Holevo2011}). According to Section~\ref{sec:OptimalPhaseEstimationNoConstraints}, the optimal input state is given by $\ket{+}$ and produces a minimal average cost of $C_{\min}^{(2,1)}=1$. Revisiting the optimization problem for the average cost posed in Eq.~\eqref{eq:MinimalAverageCostRhoRelaxed}, we now define the minimal average cost under MIO-preserving superchannels as
\begin{subequations}
\begin{alignat}{2}
    \hat{C}_{\min}^{(2,1)}(\rho)=&\inf   \quad && \sum_{x=0}^{\tilde{M}-1} \int \dPhi \,  C\left(\phi -\hat{\phi}_x\right) \Tr{\left( \rho^T \otimes J_\phi^T\otimes M_x\right) J_{\mathcal{S}_1}} \\
        & \suchthat && \tilde{M}\geq 1\\
        & && \lbrace \hat{\phi}_x \rbrace_{0\leq x\leq \tilde{M}-1} \subset [0,2\pi) \\
        & && \sum_{x=0}^{\tilde{M}-1} M_x =\mathbb{1} \\
        & && M_x \geq 0   \\
        & && J_{\mathcal{S}_1}\quad \text{the Choi state of a MIO-preserving superchannel}.
\end{alignat}
\end{subequations}
However, by considering MIO-preserving but not completely MIO-preserving superchannels we trivialize the problem of phase estimation in our setting, in the sense that we can always perform ideal phase estimation regardless of the supplied coherence in $\rho$. To see why this is the case, consider the MIO-preserving superchannel defined in Eq.~\eqref{eq:SuperChannelMIOpNOTcMIOP}. Additionally, we can map the outcomes encoded on system $3$ to a qubit system and obtain (the also MIO-preserving superchannel)
\begin{align}
    J_{\mathcal{S}_1}= \frac{1}{2} \id_0 \otimes \left( \ketbra{\Phi^+}{\Phi^+}_{1,2}+\ketbra{\Psi^+}{\Psi^+}_{1,2} \right) \otimes \ketbra{0}{0}_3 + \left(\ketbra{\Phi^-}{\Phi^-}_{1,2} +\ketbra{\Psi^-}{\Psi^-}_{1,2} \right)\otimes \ketbra{1}{1}_3.
\end{align}
Now choose phase estimates $\hat{\phi}_x= \pi  x$ and $M_x=\ketbra{x}{x}$ to obtain
\begin{align}
    &\hat{C}_{\min}^{(2,1)}(\rho) \leq    \sum_{x=0}^{1} \int \dPhi \,  C\left(\phi -\pi x\right) \Tr{\left( \rho^T \otimes J_\phi^T\otimes M_x\right) J_{\mathcal{S}_1}} \nonumber \\
    &= \frac{1}{2}\int \dPhi \,  C\left(\phi\right) \Tr{J_{\phi}^T \left( \ketbra{\Phi^+}{\Phi^+}_{1,2}+\ketbra{\Psi^+}{\Psi^+}_{1,2} \right)}+\frac{1}{2}\int \dPhi \,  C\left(\phi\right) \Tr{J_{\phi+\pi x}^T \left( \ketbra{\Phi^-}{\Phi^-}_{1,2}+\ketbra{\Psi^-}{\Psi^-}_{1,2} \right)} \nonumber \\
    &= \Tr{ \left(\int \dPhi J_\phi^T \right) \ketbra{\Phi^+}{\Phi^+}_{1,2} }= \Tr{T Y^{(2)} T^\dagger \ketbra{\Phi^+}{\Phi^+}_{1,2}}  \nonumber\\
    &= \Tr{Y^{(2)} \ketbra{+}{+}} =1, 
\end{align}
where we used that $T=\sum_i \ketbra{ii}{i}$ in the second-to-last line. This upper bound coincides with the lower bound, which we obtain by considering no coherence constraints at all (see Section~\ref{sec:OptimalPhaseEstimationNoConstraints}). Thus, we find that $\hat{C}_{\min}^{(2,1)}(\rho)=1$ irrespective of $\rho.$ This highlights the fact that we do not necessarily require local coherence to do optimal phase estimation as long as we encode the phase to an entangled state (and perform a global measurement) as in this example.

It is rather surprising that, despite never ``pulling" something out of the superchannel in this application to phase estimation, we encounter a stark difference between MIO-preserving and completely MIO-preserving superchannels. As long as we do not constrain the internal workings of such an algorithm, we can generate, manipulate, and detect all possible quantum resources internally, encode the final outcome in a classical system, and return it, which is exactly what happens if we consider the MIO-preserving superchannel in the example above. In contrast, imposing that the superchannel has to be completely MIO-preserving, we naturally constrain the superchannel in a way that allows us to describe the application to phase estimation in a meaningful manner. More generally, since any quantum algorithm in computing can be written in a way that has a classical input and a classical output only, we believe that this is required to investigate the role of quantum resources in computation in a more general setting.

\section{Task-tailored coherence monotones}
As introduced in the main text, for an arbitrary but fixed cost function $C$ we define the functional
\begin{align}\label{eq:AdvantageSM}
    \mathcal{A}^{(M)}(\rho)&:= \max_{\mathcal{M}\in \MIO} \Tr{\left( \Delta Y^{(M)}-Y^{(M)}\right)\mathcal{M}(\rho)} =  C_0-C_{\min}^{(M,1)}(\rho),
\end{align}
where $Y^{(M)}$ is defined as in Eq.~\eqref{eq:CostMatrixDefinitionSM}. In this section, we show that every bit of coherence is a resource for phase estimation and provide a bound on the minimal achievable average cost.

\begin{thm}\label{thm:AdvantageCoherenceSM}
    The functionals $\mathcal{A}^{(M)}$ are convex coherence monotones, i.e., $\mathcal{A}^{(M)}(\mathcal{N}(\rho))\leq \mathcal{A}^{(M)}(\rho)$  for all $\mathcal{N} \in \MIO$ and all states $\rho$. For any cost matrix $Y^{(M)}$ with $Y^{(M)}\neq \Delta Y^{(M)} $, the monotones are faithful, i.e.,  $\mathcal{A}^{(M)}(\rho)\geq 0$ with equality iff $\rho \in \I$.
\end{thm}
\begin{proof}
    Let $\mathcal{N} \in \MIO$. Since MIO is closed under concatenations of MIO channels, we have monotonicity under MIO channels of the advantage $\mathcal{A}^{(M)}$ as
    \begin{align}
        \mathcal{A}^{(M)}(\mathcal{N}\rho)=\max_{\mathcal{M}\in \MIO} \Tr{\left( \Delta Y^{(M)}-Y^{(M)}\right)\mathcal{M}\mathcal{N}(\rho)} \leq  \max_{\mathcal{M}\in \MIO} \Tr{\left( \Delta Y^{(M)}-Y^{(M)}\right)\mathcal{M}(\rho)}= \mathcal{A}^{(M)}(\rho).
    \end{align}
    These monotones are  convex since
    \begin{align}\label{eq:convex}
        \mathcal{A}^{(M)}(s\sigma+t\tau)&=\max_{\mathcal{M}\in\MIO} \Tr{\left( \Delta Y^{(M)}-Y^{(M)}\right) \mathcal{M}(s\sigma+t\tau)} \nonumber\\
        &= \max_{\mathcal{M}\in\MIO} \left(s \Tr{\left( \Delta Y^{(M)}-Y^{(M)}\right) \mathcal{M}(\sigma)}+t \Tr{\left( \Delta Y^{(M)}-Y^{(M)}\right) \mathcal{M}(\tau)} \right)\nonumber\\
        &\leq \max_{\mathcal{M}\in\MIO} s \Tr{\left( \Delta Y^{(M)}-Y^{(M)}\right) \mathcal{M}(\sigma)}+\max_{\mathcal{M}\in\MIO} t \Tr{\left( \Delta Y^{(M)}-Y^{(M)}\right) \mathcal{M}(\tau)}\nonumber \\
        &=s \mathcal{A}^{(M)}(\sigma)+t \mathcal{A}^{(M)}(\tau).
    \end{align}
Non-negativity of the monotones follows since $\mathcal{M}=\Delta\in \MIO$, thus
\begin{align}
	\mathcal{A}^{(M)}(\rho)=\max_{\mathcal{M}\in\MIO} \Tr{\left( \Delta Y^{(M)}-Y^{(M)}\right) \mathcal{M}(\rho)}\ge \Tr{\left( \Delta Y^{(M)}-Y^{(M)}\right) \Delta \rho}=0.
\end{align}
If $\sigma\in\text{I}$, then
\begin{align}
	\mathcal{A}^{(M)}(\sigma)=\max_{\mathcal{M}\in\MIO} \Tr{\left( \Delta Y^{(M)}-Y^{(M)}\right) \mathcal{M}(\sigma)}= \max_{\mathcal{M}\in\MIO} \Tr{\left( \Delta Y^{(M)}-Y^{(M)}\right) \Delta \mathcal{M}(\sigma)}=0.
\end{align}
If $\rho\notin \text{I}$ and $\Delta Y^{(M)} \neq Y^{(M)} $, we now show that $\mathcal{A}^{(M)}(\rho)>0$. From these two assumptions follows that there exists a pair $j\ne k$ such that $\rho_{j,k}=|\rho_{j,k}|e^{-i\theta}\ne0$ and an $l>0$ such that $Y_{0,l}^{(M)}=|Y_{0,l}^{(M)}|e^{-i\varphi}\ne0$ (it is sufficient to consider elements in the first row since $Y^{(M)}$ is a Toeplitz matrix). In particular, we can choose the integers $j,k,l$ such that
\begin{subequations}
\begin{align}
    &|\rho_{j,k}|= \max_{j\neq k} |\rho_{j,k}|, \label{eq:RhoMaxElements} \\
    &|Y_{0,l}^{(M)}| =  \max_{0<l\le M-1} |Y_{0,l}^{(M)}|.
\end{align}
\end{subequations}
Now let us define a channel $\mathcal{N}$ by Kraus operators
\begin{subequations}
    \begin{align}
		K=&\ketbra{j}{j}+\ketbra{k}{k},  \\
		L_i=&\ketbra{i}{i} \quad \forall i\notin  \{j,k\}.
	\end{align}
\end{subequations}
Then 
	\begin{align}
		K^\dagger K+\sum_{i\notin\{j,k\}}L_i^\dagger L_i=\mathbb{1}.
	\end{align}
This defines a MIO channel. Without loss of generality let $j\ne l$ and $k\ne 0$. Let $P_\pi$ be the unitary corresponding to the permutation $\pi$ that exchanges the zeroth and the $j$-th element and the $k$-the and $l$-th and leaves all other elements unaffected. Moreover, let 
	\begin{align}
		U=-e^{i(\theta+\varphi)} \ketbra{0}{0}+\sum_{i>0}\ketbra{i}{i}.
	\end{align}
    Both unitaries $P_\pi$ and $U$ are incoherent. Therefore, since the concatenation of free channels is free,
    \begin{align}
        \mathcal{A}^{(M)}(\rho)&= \max_{\mathcal{M}\in \MIO} \Tr{\left( \Delta Y^{(M)}-Y^{(M)}\right)\mathcal{M}(\rho)} \nonumber\\
        &\geq \Tr{\left( \Delta Y^{(M)}-Y^{(M)}\right) \mathcal{U}\mathcal{P}_\pi \mathcal{N}(\rho)} \nonumber\\
        &\overset{\eqref{eq:RhoMaxElements}}{=} \Tr{\left( \Delta Y^{(M)}-Y^{(M)}\right) \mathcal{U}\mathcal{P}_\pi \left(\rho_{j,k}\ketbra{j}{k}+\rho_{k,j}\ketbra{k}{j}\right)} \nonumber\\
        &= -\Tr{Y^{(M)} \mathcal{U}  \left(\rho_{j,k}\ketbra{0}{l}+\rho_{k,j}\ketbra{l}{0}\right)}\nonumber \\
        &= \Tr{Y^{(M)} e^{i(\theta+\varphi)}  \left(\rho_{j,k}\ketbra{0}{l}+\rho_{k,j}\ketbra{l}{0}\right)} \nonumber\\
        &= 2|\rho_{j,k}| |Y_{0,l}^{(M)}| \nonumber\\
        &= 2 \max_{0< l\le M-1} \left|Y_{0,l}^{(M)}\right| \max_{i\ne j} |\rho_{i,j}| \label{eq:BoundAdvantageElements}\nonumber\\
        &>0,
    \end{align}
    where the second-to-last line follows from our choice of matrix elements.
\end{proof}

Unless the cost function is constant and thus trivial, there always exists a sufficiently large $M$ such that the corresponding cost matrix $Y^{(M)}$ is not diagonal. Hence, for any cost function, there exists a minimal $M$ (equivalently a pair $(d,N)$) such that every bit of coherence is useful for phase estimation and, thus, for algorithms that use it as a subroutine. The operational advantage coherence provides is directly quantified by Eq.~\eqref{eq:AdvantageSM}.

Next, we provide a lower bound on the minimal achievable average cost, which separates into a coherence-dependent and a problem-specific part.

\begin{prop}\label{prop:lowerBoundWeight}
Let $D$ denote the set of all quantum states. The generalized robustness~\cite{Napoli2016} and weight~\cite{Bu2018} of coherence are given by 
\begin{align}
    &W(\rho)= \min_{\tau, \sigma} \left\lbrace w\geq 0 : \rho=w\tau+(1-w)\sigma, \sigma \in \I, \tau \in D  \right\rbrace\, , \\
    &C_R(\rho)= \min_{\tau, \sigma} \left\lbrace r\geq 0 : \rho +r \tau= (1+r) \sigma, \in \I, \tau \in D  \right\rbrace.
\end{align}
Then, the average cost is lower bounded by 
    \begin{align}
        C_{\min}^{(M,1)}(\rho)\geq \lambda_{\min}\left(Y^{(M)}\right)+\left(C_0-\lambda_{\min}\left(Y^{(M)}\right) \right)\left(1-W(\rho)\right).
    \end{align}
    For $M=2$, and $\rho$ a qubit state, the average cost is exactly given by 
    \begin{align}
         C_{\min}^{(2,1)}(\rho)= C_0 - \left(C_0-\lambda_{\min}\left(Y_C^{(d)}\right)\right) C_{R}(\rho).
    \end{align}
\end{prop}
\begin{proof}

We start with
    \begin{align}
        C_{\min}^{(M,1)}(\rho)&= \min_{\mathcal{M} \in \MIO} \Tr{Y\mathcal{M}(\rho)} \nonumber\\
        &= \min_{\mathcal{M} \in \MIO} \Tr{ \left(\lambda_{\min}\left(Y^{(M)}\right)\id+Y^{(M)}-\lambda_{\min}\left(Y^{(M)}\right)\id\right)\mathcal{M}(\rho)}\nonumber\\
        &=\lambda_{\min}\left(Y^{(M)}\right)+ \min_{\mathcal{M}} \left\lbrace \Tr{Z \rho} : Z=\mathcal{M}^\dagger\left(Y^{(M)}-\lambda_{\min}\left(Y^{(M)}\right)\id\right),\, \mathcal{M}\in \MIO\right\rbrace \nonumber\\
        &\geq \lambda_{\min}\left(Y^{(M)}\right)+ \min_{Z} \left\lbrace \Tr{Z\rho} : Z \geq 0, \Delta Z= \left(C_0-\lambda_{\min}\left(Y^{(M)}\right)\right) \id\right\rbrace \label{eq:RelaxedWeightBound}
    \end{align}
where we relaxed the optimization in the last line. This follows from using that $Z\geq 0$ since $Y^{(M)}-\lambda_{\min}\left(Y^{(M)}\right) \geq 0$. Moreover, invoking the fact that $\mathcal{M} \in \MIO$ and thus
$\Delta Z=\Delta\mathcal{M}^\dagger(Y^{(M)}-\lambda_{\min}\left(Y^{(M)}\right) \id)=\Delta\mathcal{M}^\dagger \Delta(Y^{(M)}-\lambda_{\min}\left(Y^{(M)}\right) \id)=(C_0-\lambda_{\min}\left(Y^{(M)}\right)) \id$ leads to the second constraint. Next, note that the weight of coherence can be computed using the SDP~\cite{Bu2018} 
\begin{align}\label{eq:dualSDPWight}
    W(\rho)&= \max \left\lbrace \Tr{(-Z)\rho}+1 | Z\geq 0, \Delta(Z)=\id  \right\rbrace
\end{align}
which can be rewritten as
\begin{align}
    W(\rho)&= \max \left\lbrace \Tr{(-Z)\rho}+1 : Z\geq 0, \Delta(Z)=\id  \right\rbrace =1 - \min \left\lbrace \Tr{Z\rho} : Z\geq 0, \Delta(Z)=\id  \right\rbrace.
\end{align}
Rescaling the optimization problem in Eq.~\eqref{eq:RelaxedWeightBound} yields 
    \begin{align}
        C_{\min}^{(M,1)}(\rho)\geq \lambda_{\min}\left(Y^{(M)}\right)+\left(C_0-\lambda_{\min}\left(Y^{(M)}\right) \right)\left(1-W(\rho)\right),
    \end{align}
which finishes the first part of the proof. 

For the second part, note that according to the Gershgorin circle theorem, there exists an $i\in\{0,\cdots,M-1\}$ such that
\begin{align}\label{eq:Gershgorin}
    C_0-\lambda_{\min}\left(Y_C^{(M)}\right) \le& \sum_{ j\ne i} \left|Y_{i,j}^{(M)}\right|  \nonumber\\
    \le& (M-1) \max_{j\ne i} \left|Y_{i,j}^{(M)}\right|  \nonumber\\
    \le& (M-1) \max_{0< l\le M-1}  \left|Y_{0,l}^{(M)}\right|.
\end{align}
Let $d=\dim(\rho)$, then using the $l_1$-norm of coherence~\cite{Baumgratz2014} yields
\begin{align}
    C_{l_1}(\rho)=\sum_{i\ne j} |\rho_{i,j}| \le d(d-1) \max_{i\ne j } |\rho_{i,j}|,
\end{align}
which is achieved for $d=2$. Therefore, using Eq.~\eqref{eq:BoundAdvantageElements}, we have
\begin{align}\label{eq:AdvantageLowerBoundElements}
    \mathcal{A}^{(M)}(\rho) \overset{\eqref{eq:BoundAdvantageElements}}{\geq}2 \max_{0\le l\le M-1} \left|Y_{0,l}^{(M)} \right| \max_{i\ne j} |\rho_{i,j}| \ge \left(C_0-\lambda_{\min}\left(Y_C^{(M)}\right)\right)\frac{2C_{l_1}(\rho)}{d(d-1)(M-1)}.
\end{align}

Note that $\left|Y_{0,1}^{(2)}\right| = \left(C_0-\lambda_{\min}\left(Y_C^{(2)}\right)\right)$ and that the robustness of coherence and the $l_1$-norm of coherence coincide for qubits, see Ref.~\cite[Theorem 4]{Piani2016}.
For $M=d=2$, we thus have that
\begin{align}\label{eq:AdvantageUpperBound2}
   \mathcal{A}^{(2)}(\rho)=& \max_{\mathcal{M}\in \MIO} \Tr{\left( \Delta Y^{(2)}-Y^{(2)}\right)\mathcal{M}(\rho)}\nonumber\\
    \leq& \max_{\mathcal{M}\in \MIO} \sum_{k\neq l} |Y_{k,l}^{(2)} (\mathcal{M}(\rho))_{l,k}| \nonumber\\
    =&  2 \left|Y_{0,1}^{(2)}\right| \max_{\mathcal{M}\in \MIO}  |(\mathcal{M}(\rho))_{1,0}| \nonumber\\
    =&  \left|Y_{0,1}^{(2)}\right|\max_{\mathcal{M}\in \MIO}  C_{l_1}(\mathcal{M}(\rho)) \nonumber\\
    =&  \left|Y_{0,1}^{(2)}\right|\max_{\mathcal{M}\in \MIO}  C_{R}(\mathcal{M}(\rho)) \nonumber\\
    =&\left(C_0-\lambda_{\min}\left(Y_C^{(2)}\right)\right) C_{R}(\rho).
\end{align}
Combining Eq.~\eqref{eq:AdvantageLowerBoundElements} and Eq.~\eqref{eq:AdvantageUpperBound2} for $d=M=2$ concludes the proof.
\end{proof}
In particular, we note that the bound in Proposition~\ref{prop:lowerBoundWeight} is sharp. To this end, let $\rho= p\ketbra{\Psi^+}{\Psi^+} +(1-p) \frac{\id}{M}$, where $\Psi^+$ denotes the maximally coherent state of dimension $M$. Clearly, we have $W(\rho)=p$, and 
\begin{align}
        C_{\min}^{(M,1)}(\rho)= p \lambda_{\min}\left(Y^{(M)}\right)+(1-p)C_0 = \lambda_{\min}\left(Y^{(M)}\right) W(\rho) +C_0 \left((1-W(\rho)\right),
\end{align}
which is exactly the bound in Proposition~\ref{prop:lowerBoundWeight}. Note that the weight of coherence has the peculiar property that for every incoherent pure state the weight is zero, whilst for any other pure state the weight is one, which trivializes the bound in Proposition~\ref{prop:lowerBoundWeight} for all pure states. Lastly, let us consider the asymptotic behavior of the lower bound in Proposition~\ref{prop:lowerBoundWeight}. In particular, we consider a fixed resource state $\rho$ with fixed dimension and take the limit of $M\to \infty$ (or equivalently the limit of infinitely many copies of the unitary). Assuming that the cost function has $ \essinf_\phi C(\phi)=0$ (which is the case for common cost functions) the minimal eigenvalue of the Toeplitz matrix $Y^{(M)}$ asymptotically goes to zero, see Ref.~\cite{Serra1996}. Thus, using the bound from Proposition~\ref{prop:lowerBoundWeight} we find that
\begin{align}
    \lim_{M\to \infty} C_{\min}^{(M,1)}(\rho) &\geq   \lim_{M\to \infty} \left[\lambda_{\min}\left(Y^{(M)}\right)+\left(C_0-\lambda_{\min}\left(Y^{(M)}\right) \right)\left(1-W(\rho)\right) \right] \nonumber \\
    &= C_0\left(1-W(\rho)\right).
\end{align}
Thus, the weight of coherence provides a fundamental lower bound to the minimal average cost in the limit of infinitely many copies of the unitary $V_\phi^{(d)}$. In contrast, in the unconstrained case, it follows from Eq.~\eqref{eq:unconstraint} that the minimal average cost approaches zero in the limit $M\to \infty.$

\section{Example: Holevo cost function}\label{sec:HolevoCostFunction}
In this section, we consider the specific example of the Holevo cost function $C(\phi)=4\sin^2(\phi/2)$, also known as periodized variance (see also Ref.~\cite{Holevo2011}). In Fig.~\ref{fig:WeightBoundTightness}, we present a comparison between the bound in Proposition~\ref{prop:lowerBoundWeight} and the weight of coherence for random states.
\begin{figure}[ht]
\begin{subfigure}{0.5\textwidth}
    \centering
    \includegraphics[width=1\linewidth]{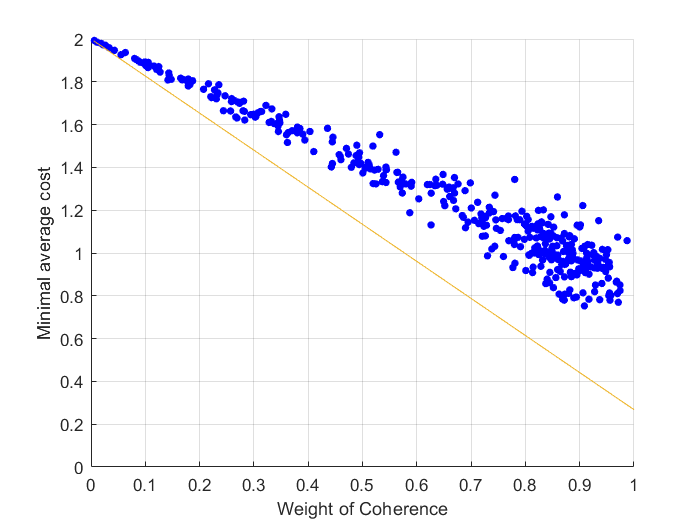}
    \caption{Comparison of the bound in Proposition~\ref{prop:lowerBoundWeight} (solid line) to the exact values of the minimal average cost of random states of dimension $d=5$ (blue dots) for the cost function $C(\phi)=4\sin^2(\phi/2)$ and $M=5$. }
    \label{fig:WeightBoundTightness}
\end{subfigure}%
\begin{subfigure}{0.5\textwidth}
    \centering
    \includegraphics[width=1\linewidth]{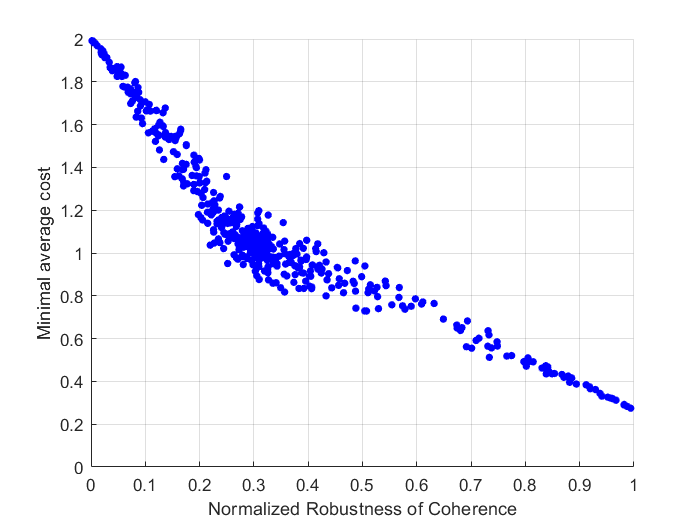}
    \caption{Comparison between the exact minimal average cost for the cost function $C(\phi)=4\sin^2(\phi/2)$ and $M=5$ and the normalized robustness, i.e., $C_R(\rho)/(d-1)$, for random states of dimension $d=5$.}
    \label{fig:Robustness}
\end{subfigure}
\end{figure}
The proportionality of the minimal average cost to the (normalized) robustness of coherence breaks down in higher dimensions as depicted in Fig.~\ref{fig:Robustness}. We leave the question whether there exists a similar bound to Proposition~\ref{prop:lowerBoundWeight} in terms of the robustness (for specific classes of cost functions) to future work. Lastly, we want to emphasize that a maximally coherent state is not necessarily required to achieve the minimal average cost. Consider the Holevo cost function, then the cost matrix is $Y^{(M)}=2\id -  \sum_{k=0}^{M-1} \left( \ketbra{k}{k+1}+\ketbra{k+1}{k} \right)$. As shown for example in Ref.~\cite{vanDam2007b}, the minimal eigenvalue (and thus the coherence-unconstrained minimal average cost) and its corresponding eigenstate $\ket{\nu}$ are given by
\begin{subequations}
    \begin{align}
    &\lambda_{\min}\left(Y^{(M)}\right)=2-2\cos\left(\frac{\pi}{d+1}\right)=4 \sin^2\left( \frac{\pi}{2(d+1)}\right), \\
    &\ket{\nu}= \sqrt{\frac{2}{d+1}}\sum_{j=0}^{d-1} \sin\left( \pi \frac{j+1}{d+1}\right).
\end{align}
\end{subequations}
Clearly $\ket{\nu}$ is not a maximally coherent state.

\section{Extended discussion}
Firstly, we want to discuss two assumptions on the phase estimation problem we made in Section~\ref{sec:optimalPhaseEstimation}. 
Like previous works~\cite{vanDam2007,vanDam2007b} do, we only consider cost functions of the form $C(\phi-\hat{\phi}_x)$. For many practical applications, such as in computation or in metrology, this is sufficient. However, in general, one could consider cost functions of the form $C(\phi,\hat{\phi}_x)$ that do not depend only on the difference between the two variables. Moreover, we assumed a uniform prior distribution of the phase. Our methods rely heavily on these assumptions. More technically, we use that the prior distribution $p(\phi)$ and cost function $C(\phi,\hat{\phi}_x)$ can be expressed as $p(\phi)C(\phi,\hat{\phi}_x)=f(\phi-\hat{\phi}_x)$ for some function $f$. It is unclear whether the minimal average cost with respect to arbitrary cost functions $C(\phi,\hat{\phi}_x)$ and non-uniform priors $p(\phi)$ can still be computed by a semidefinite program if this is no longer the case. It would be interesting to further explore this question.

As promised in the main text, we now discuss the role of dynamic coherence, i.e., non-free channels or even non-free supermaps, in phase estimation and why we choose to focus on static resources. To this end, let us first consider a scenario where we are given a single non-free channel $\mathcal{N}$ together with multiple copies of the black-box unitary $V_\phi^{(d)}$. Just like we did for static coherence (a coherent input state $\rho$), we insert the copies of the black-box unitary and the coherent channel $\mathcal{N}$ into a free supermap to optimally infer information on the phase. Since quantum channels cannot be stored, the placement of the channels in the (causally-ordered) supermap is crucial here: If we were to place the channel that can generate coherence after the unitaries, the coherence it generates cannot be used to encode any information. By extension, the optimal placement for the coherent channel is therefore before all the unitaries. Thus, the only change compared to the scenario in which we consider static coherence is that rather than considering a single incoherent channel $\mathcal{M}$ in Theorem~\ref{thm:OptAvgCostSM}, which optimizes $C_{\min}(\rho) = \min_{\mathcal{M} \in \MIO} \Tr{\mathcal{M}(\rho) Y^{(M)}}$, we must consider an incoherent superchannel $\mathcal{S}_1$ (with trivial input system $0$) processing the supplied coherent channel, thereby optimizing $\Tr{\mathcal{S}_1[\mathcal{N}]Y^{(M)}}$.

More generally, if we supply multiple non-free channels (or even resourceful supermaps), the optimal placement, and thus their usage, is not clear. The intuition that one simply uses all the available dynamical resources at the beginning of the protocol to generate as much coherence as possible and then proceeds to probe the black-box unitaries is not correct. For instance, if we are given a superchannel that has a non-zero coherence generation capacity at the first tooth but is resource-destroying in the second tooth (e.g., it returns the fully dephased input), it is better to use the coherence that can be generated in the beginning to encode information about the phase by probing one copy of the black-box unitary with it and to leave the second tooth of the superchannel as it is.

Lastly, we want to briefly discuss the role of (multi-partite) entanglement in the optimal phase estimation protocol that achieves the minimal average cost in Theorem~\ref{thm:OptAvgCostSM}. Recalling the implementation of the unitary $V_\phi^{(M)}$ (from $N$ copies of $V_\phi^{(d)}$) in Lemma~\ref{lem:PhaseShiftCurcuit}, and in particular the embedding isometry $W_e$, we see that this isometry transforms (local) coherence into multi-partite entanglement. The role of entanglement in phase (and more generally parameter estimation) has been studied extensively, particularly in the context of metrology~\cite{Wineland1996,Braunstein1994,Huelga1997,Augusiak2016}, and has been identified as a resource. We want to emphasize that our results are not in conflict with this since coherence and entanglement can be interconverted, which is exactly what the optimal phase estimation protocol here does.

\end{document}